\documentclass[a4paper,UKenglish,cleveref, autoref, thm-restate, numberwithinsect]{lipics-v2021}

\pdfoutput=1 
\hideLIPIcs  

\title{Improving Reachability in Vector Addition Systems through Pumpability} 


\author{Weijun Chen}
{BASICS, Shanghai Jiao Tong University}
{cwj2018@sjtu.edu.cn}
{https://orcid.org/0009-0007-2611-0338}
{}

\author{Yuxi Fu}
{BASICS, Shanghai Jiao Tong University}
{fu-yx@cs.sjtu.edu.cn}
{https://orcid.org/0000-0001-6429-7550}
{}

\author{Yangluo Zheng}
{BASICS, Shanghai Jiao Tong University}
{wunschunreif@sjtu.edu.cn}
{https://orcid.org/0009-0000-1028-5458}
{}

\authorrunning{W. Chen, Y. Fu, and Y. Zheng} 

\Copyright{Weijun Chen, Yuxi Fu, and Yangluo Zheng} 



\category{} 

\relatedversion{} 

\funding{All the authors are supported by the NSFC (National Natural Science Foundation of China) grant No.\ 62572319.}

\nolinenumbers 

\EventEditors{John Q. Open and Joan R. Access}
\EventNoEds{2}
\EventLongTitle{42nd Conference on Very Important Topics (CVIT 2016)}
\EventShortTitle{CVIT 2016}
\EventAcronym{CVIT}
\EventYear{2016}
\EventDate{December 24--27, 2016}
\EventLocation{Little Whinging, United Kingdom}
\EventLogo{}
\SeriesVolume{42}
\ArticleNo{23}

\usepackage{bm}
\usepackage{stmaryrd}

\begin{document}

\newcommand{\defproblem}[3]{
    \;\par
    \noindent\fbox{
        \begin{minipage}{0.96\columnwidth}
            \underline{#1}
            \begin{description}
                \item[Input:] {#2}
                \item[Question:] {#3}
            \end{description}
        \end{minipage}
    }
    \par\;
}

\newcommand{\OpName}[1]{\mathop{\operatorname{\textup{#1}}}\nolimits}
\newcommand{\OpNameSC}[1]{\mathop{\operatorname{\textup{\textsc{#1}}}}\nolimits}

\renewcommand{\Vec}[1]{\bm{#1}}
\newcommand{\norm}[1]{\left\Vert{#1}\right\Vert}

\newcommand{\Supp}{\OpName{supp}}
\newcommand{\CycleSpace}{\OpNameSC{cyc}}
\newcommand{\Len}{\OpNameSC{len}}
\newcommand{\Span}{\OpName{span}}
\newcommand{\Drop}{\OpNameSC{drop}}
\newcommand{\Cone}{\OpName{cone}}
\newcommand{\SeqCone}{\operatorname{SeqCone}\nolimits}
\newcommand{\Next}{\operatorname{next}\nolimits}
\newcommand{\Rank}{\operatorname{rank}\nolimits}
\newcommand{\RankFull}{\operatorname{rank}\nolimits_{\text{full}}}
\newcommand{\ExpF}[1]{{\normalfont\textsf{exp}({#1})}}
\newcommand{\PolyF}[1]{{\normalfont\textsf{poly}({#1})}}
\newcommand{\Clean}[1]{{\normalfont\textrm{clean}({#1})}}
\newcommand{\Dec}[1]{{\normalfont\textrm{dec}({#1})}}
\newcommand{\Facc}{\textsc{Facc}}
\newcommand{\Bacc}{\textsc{Bacc}}

\newcommand{\gdim}{\OpName{gdim}}
\newcommand{\dimcyc}{\dim_{\textup{cyc}}}
\newcommand{\dimcom}{\dim_{\textup{com}}}
\newcommand{\Reach}{\OpName{Reach}}
\newcommand{\Source}{\OpName{src}}
\newcommand{\Target}{\OpName{trg}}
\newcommand{\Argmin}{\mathop{\operatorname{arg\, min}}}
\newcommand{\Argmax}{\mathop{\operatorname{arg\, max}}}
\newcommand{\Size}{\OpNameSC{size}}

\newcommand{\PSPACE}{{\normalfont\textsf{PSPACE}}}
\newcommand{\NP}{{\normalfont\textsf{NP}}}
\newcommand{\NL}{{\normalfont\textsf{NL}}}
\newcommand{\ELEM}{{\normalfont\textsf{ELEMENTARY}}}

\newcommand{\bracket}[1]{\left\langle{#1}\right\rangle}
\renewcommand{\Bar}[1]{\overline{#1}}
\newcommand{\lelex}{\le_{\text{lex}}}
\newcommand{\ltlex}{<_{\text{lex}}}

\newcommand{\LPSSystem}[1]{\mathcal{E}_{\normalfont\text{LPS}}({#1})}
\newcommand{\LPSSystemHomo}[1]{\mathcal{E}^0_{\normalfont\text{LPS}}({#1})}

\newcommand{\KLMSystem}[1]{\mathcal{E}({#1})}
\newcommand{\KLMSystemHomo}[1]{\mathcal{E}^0({#1})}

\newcommand{\RotTo}{\mathbin{\curvearrowright}}
\newcommand{\RotToEq}{\mathbin{\underline{\curvearrowright}}}
\newcommand{\NotRotToEq}{\mathbin{\not\kern -0.3em \RotToEq}}
\newcommand{\Inner}[1]{\langle {#1} \rangle}
\newcommand{\UBCounters}{\operatorname{\textsc{ub}}}

\maketitle

\begin{abstract}
Vector addition systems (VAS) constitute an important model of computation and concurrency that is equally expressive as the Petri net model. Recently, a lot of research has been conducted on vector addition systems with states (VASS), which are VASes equipped with a finite state control. Results on VASS naturally carry over to VAS, but no straightforward improvement is available. In this paper, we investigate the reachability problem in VAS in fixed dimensions. Based on a pumpability analysis of VAS that refines Rackoff's extraction for VASS, we obtain an $\textsf{F}_{d-2}$ upper bound for the $d$-dimensional VAS reachability problem, improving the $\textsf{F}_d$ upper bound inherited from the $d$-dimensional VASS reachability problem. Low-dimensional VASes are also considered. In particular, we establish a $\textsf{PSPACE}$ upper bound for reachability in 4-dimensional VAS and an $\textsf{ELEMENTARY}$ upper bound for 5-dimensional VAS, while the same upper bounds were known only for 2-VASS and 3-VASS, respectively. The result for 4-VAS particularly hinges on a simplified projection technique developed for geometrically 2-dimensional VASSes, whose reachability problem is shown to be equivalent to 2-VASS. 

\end{abstract}

\section{Introduction}
\label{sec:introduction}

Vector addition systems (VAS), being equivalent to Petri nets, are a well-established model of concurrency. A VAS is simply defined by a finite set of integer vectors called actions. Configurations, which are non-negative integer vectors, may move along these actions. The reachability problem, asking whether there is a run from one configuration to another, is the central algorithmic problem in the study of VAS and Petri nets that has found various applications \cite{DBLP:journals/tie/ZurawskiZ94}. Decidability of the reachability problem in VAS has been known for over 40 years \cite{DBLP:journals/siamcomp/Mayr84, DBLP:conf/stoc/Kosaraju82, DBLP:journals/tcs/Lambert92}. Its complexity was settled to be Ackermann-complete only recently \cite{DBLP:conf/lics/LerouxS19, DBLP:conf/focs/Leroux21, DBLP:conf/focs/CzerwinskiO21}. Zooming in to the $d$-dimensional VAS reachability problem with $d$ fixed, a significant gap still remains between the complexity lower bound $\textsf{F}_{O(d/2)}$ \cite{DBLP:conf/fsttcs/CzerwinskiJ0LO23} and upper bound $\textsf{F}_d$ \cite{DBLP:conf/lics/LerouxS19}, where $\textsf{F}_d$ is the $d$-th level of the fast-growing complexity hierarchy \cite{DBLP:journals/toct/Schmitz16}.

In the literature, the enhanced model ``vector addition system with states (VASS)'' turns out to be more popular. VASS augments VAS with a finite state control, so that each state has its own set of actions. Indeed, the Ackermann-completeness result mentioned above was proved directly for VASS. Considering the expressive power, Hopcroft and Pansiot \cite{DBLP:journals/tcs/HopcroftP79} showed that the finite state control can be simulated by a VAS with three additional dimensions. In general, this implies problems in VASS are inter-reducible with those in VAS if the dimension is not a fixed parameter. In fixed dimensions, one might naturally expect better complexity bounds for VAS than for VASS. However, such results are not known for most problems. For the $d$-dimensional VAS reachability problem, the best upper bound we have is $\textsf{F}_d$ inherited from $d$-dimensional VASS \cite{DBLP:conf/icalp/FuYZ24}. Gaps between VAS and VASS are clearer in low dimensions. For example, it was shown in \cite{DBLP:journals/tcs/HopcroftP79} that reachability relations in both 2-dimensional VASS and 5-dimensional VAS are effectively semilinear. Inspired by this semilinearity, \PSPACE{}-completeness is established for reachability in 2-VASS \cite{DBLP:journals/jacm/BlondinEFGHLMT21}. However, for 5-VAS, currently we do not even have an elementary upper bound. This paper is devoted to narrowing the gap between VAS and VASS.

\paragraph*{Contribution.}
In this paper, we mainly focus on the reachability problem in fixed-dimensional VAS. Our first result shows that in the same fixed dimension $d$, the complexity of reachability in VAS could be lower than that in VASS.

\begin{theorem}
    Reachability in $(d+2)$-VAS is in $\mathsf{F}_d$ for every $d\ge 3$.
\end{theorem}

This result is based on a sharper analysis of pumpability in VAS. A $d$-dimensional configuration $\Vec{s}$ is said to be pumpable if $\Vec{s}$ can reach some $\Vec{s}' \ge \Vec{s} + \Vec{1}$. Intuitively, when both the source and the target configurations are pumpable, one can bound the length of the shortest runs between them by a polynomial. On the other hand, when pumpability fails, a known technique called Rackoff's extraction \cite{DBLP:journals/tcs/Rackoff78} exhibits one bounded coordinate on any run. The dimension can then be reduced by one. In \autoref{sec:pumpability}, we improve this technique and show that for a specific type of VAS, the failure of pumpability yields two bounded coordinates. This allows us to reduce the $(d+2)$-VAS reachability problem to ``almost'' $d$-VASS reachability: the reduced VASS might have dimension more than $d$, but the \emph{geometric dimension} of each of its strongly connected components is at most $d$. Geometric dimension, defined as the dimension of the vector space spanned by effects of simple cycles, is a metric on VASS gaining importance in recent research \cite{DBLP:conf/icalp/FuYZ24,DBLP:conf/icalp/CzerwinskiJ0O25,DBLP:conf/concur/Zheng25,DBLP:conf/lics/GuttenbergCL25}. A closer inspection of \cite{DBLP:conf/icalp/FuYZ24} shows that the $\textsf{F}_d$ upper bound for $d$-VASS indeed holds for VASSes produced by our reduction where each strongly connected component has geometric dimension at most $d$. 

Our second result focuses on low-dimensional VAS, especially 4-VAS and 5-VAS.

\begin{theorem}
    Reachability in 5-VAS is in $\ELEM$. Reachability in 4-VAS is in $\PSPACE$ under binary encoding and is $\NL$-complete under unary encoding.
\end{theorem}

To the best of our knowledge, this provides the first elementary upper bound for VAS of dimension greater than $3$ (we remark that a \PSPACE{} upper bound for $3$-VAS could be derived from \cite[Lemma 2]{DBLP:conf/icalp/CzerwinskiJ0O25}). For low-dimensional VAS, we need a fine-grained characterization of shortest runs. We show that the shortest runs in $d$-dimensional VAS are captured either by $(d-2)$-dimensional VASS, or by a VASS consisting of two sub-VASSes whose geometric dimensions are at most $d - 3$. For 4- and 5-VAS, this allows us to use as building blocks the existing results on 2-VASS \cite{DBLP:journals/jacm/BlondinEFGHLMT21}, 3-VASS \cite{DBLP:conf/icalp/CzerwinskiJ0O25}, and geometrically 2-dimensional VASS \cite{DBLP:conf/icalp/FuYZ24,DBLP:conf/concur/Zheng25}. We remark that a projection technique converting geometrically 2-dimensional VASS to 2-VASS is also provided in this paper, simplifying that of \cite{DBLP:conf/concur/Zheng25}. 
The technique produces tight complexity bounds for reachability in geometrically 2-dimensional VASS under unary encoding, which is not answered in \cite{DBLP:conf/concur/Zheng25}. To be specific, we show that this problem is \NL{}-complete when the actual dimension is fixed, and is \NP{}-complete otherwise.

\paragraph*{Organization.} 
In \autoref{sec:preliminary}, we fix notation and introduce important concepts in VAS. Section \ref{sec:pumpability} exhibits our main technical contribution---a weaker pumpability condition based on fine-tuning Rackoff's extraction. In \autoref{sec:d-vas} we apply the pumpability condition to show $\textsf{F}_{d-2}$ upper bound for $d$-dimensional VAS where $d > 4$. In \autoref{sec:low-dim-vas-boundedness} we show a fine-grained analysis of VAS runs using similar techniques. This will be useful in \autoref{sec:low-dim-vas}, where we derive better bounds for 4- and 5-VAS. Finally, \autoref{sec:conclusion} proposes possible future directions.

\section{Preliminaries}
\label{sec:preliminary}

Let $\mathbb{N}, \mathbb{Z}, \mathbb{Q}$ denote the sets of natural numbers, integers, and rational numbers, respectively.
We use, for example, $\mathbb{Q}_{\ge 0}$ for the set of rationals restricted to non-negative values. Let $m \le n$ be two integers, we write $[m, n]$ for the set $\{m, m + 1, \ldots, n\}$. The interval $[1, n]$ is abbreviated as $[n]$.
Vectors are written in boldface $\Vec{u}, \Vec{v}, \Vec{x}, \Vec{y}$, etc. For a vector $\Vec{x} \in \mathbb{Q}^d$, we write $\Vec{x}(i)$ for its $i$-th component, where $i \in [d]$. We compare vectors component-wise, so that $\Vec{x} \le \Vec{y}$ means $\Vec{x}(i) \le \Vec{y}(i)$ for all $i \in [d]$. We use the max norm for vectors, denoted by $\norm{\Vec{x}} := \max_{i \in [d]}|\Vec{x}(i)|$. For a finite set $X \subseteq \mathbb{Q}^d$ we write $\norm{X} := \max_{\Vec{x} \in X}\norm{\Vec{x}}$. The \emph{inner product} of two vectors $\Vec{x}, \Vec{y}\in \mathbb{Q}^d$ is written as $\Inner{\Vec{x}, \Vec{y}}:=\sum_{i \in [d]} \Vec{x}(i)\cdot \Vec{y}(i)$. When the dimension $d$ is clear from the context, we write $\Vec{e}_i \in \mathbb{Q}^d$ for the unit vector with a $1$ in its $i$-th coordinate. We write $\Vec{1} := (1, 1, \ldots, 1)$ and $\Vec{0} := (0, 0, \ldots, 0)$. For a finite set $X \subseteq \mathbb{Q}^d$ of vectors where $X = \{\Vec{x}_1, \ldots, \Vec{x}_n\}$, the \emph{cone} generated by $X$ is $\Cone(X) := \{\lambda_1 \Vec{x}_1 + \cdots + \lambda_n \Vec{x}_n \mid \lambda_1, \ldots, \lambda_n \in \mathbb{Q}_{\ge 0}\}$.

\paragraph*{VASS, VAS, and Sequential VAS}

A $d$-dimensional \emph{vector addition system with states} ($d$-VASS) is a pair $V = (Q, T)$ where $Q$ is a finite set of \emph{states} and $T \subseteq Q \times \mathbb{Z}^d \times Q$ is a finite set of \emph{transitions}. By a \emph{configuration} of $V$ we mean a pair $(p, \Vec{x})$, often denoted by $p(\Vec{x})$, of a state $p \in Q$ and a non-negative vector $\Vec{x} \in \mathbb{N}^d$. A transition $t = (p, \Vec{a}, q)$ induces one-step runs $p(\Vec{x}) \xrightarrow{t} q(\Vec{y})$ where $\Vec{x}, \Vec{y} \in \mathbb{N}^d$ satisfy $\Vec{x} + \Vec{a} = \Vec{y}$. Paths are words over transitions $\pi = (p_1, \Vec{a}_1, q_1)\ldots (p_k, \Vec{a}_k, q_k) =: t_1\ldots t_k \in T^*$ such that $q_i = p_{i+1}$ for $i \in [k-1]$. If further $q_k = p_1$, we say $\pi$ is a \emph{cycle} in $V$. A cycle consisting of a single transition is called a \emph{self-loop}. The effect of such a path $\pi$ is $\Delta(\pi) := \sum_{i = 1}^k \Vec{a}_k$. The length of $\pi$ is denoted by $|\pi|:=k$. Runs induced by paths are concatenations of the one-step runs induced by the transitions, e.g.,
\begin{equation}
  \pi : p_1(\Vec{x}_1) \xrightarrow{t_1} q_1(\Vec{y_1}) \xrightarrow{t_2} \cdots \xrightarrow{t_k} q_k(\Vec{y}_k).
\end{equation}
We emphasize that $\Vec{x}_1, \Vec{y}_1, \ldots, \Vec{y}_k \ge \Vec{0}$ must be non-negative.
We write $p(\Vec{x}) \xrightarrow{\pi} q(\Vec{y})$ if there is a run induced by $\pi$ from $p(\Vec{x})$ to $q(\Vec{y})$. We also write $p(\Vec{x}) \xrightarrow{*} q(\Vec{y})$, quantifying the path existentially. Two paths $\pi_1, \pi_2$ can be concatenated into $\pi_1\pi_2$ if the target state of $\pi_1$ matches the source state of $\pi_2$. If $\pi$ is a cycle, we write $\pi^m$ for the $m$-fold concatenation of $\pi$. Sometimes we also consider $\mathbb{Z}$-configurations $p(\Vec{z})$, where the vector $\Vec{z} \in \mathbb{Z}^d$ can take negative values. We define $\mathbb{Z}$-runs similarly over $\mathbb{Z}$-configurations, denoted by $p(\Vec{z}) \xrightarrow{\pi}_{\mathbb{Z}} q(\Vec{w})$. 
Consider a run $\pi = t_1\ldots t_k$ from state $p$ to $q$.
We define the \emph{drop} of $\pi$ to be the vector $\Drop(\pi)\in\mathbb{N}^d$ where for each $i \in [d]$, $\Drop(\pi)(i)$ is the maximum decrease of the $i$-th component along this path, i.e.,
\begin{equation}
    \Drop(\pi)(i):= -\min(0,\min_{j\in [k]}\Delta(t_1\dots t_j)(i)).
\end{equation}
It is straightforward that for any $\Vec{z}\ge \Drop(\pi)$, $p(\Vec{z}) \xrightarrow{\pi} q(\Vec{w})$ holds for some $\Vec{w} \in \mathbb{N}^d$.

The (unary-encoded) \emph{size} of $V$ is defined by $\Size(V) := |Q|+ d \cdot |T| \cdot (\norm{T} + 1)$, where $\norm{T} := \max_{(p, \Vec{a}, q) \in T}\norm{\Vec{a}}$.
The \emph{reverse} of $V$ is given by $V^{\text{rev}} = (Q, T^{\text{rev}})$ where $T^{\text{rev}} = \{ (q, -\Vec{a}, p) \mid (p, \Vec{a}, q) \in T \}$.
The norm of a ($\mathbb{Z}$-)configuration $p(\Vec{x})$ is defined as $\norm{p(\Vec{x})}:=\|\Vec{x}\|$.

\emph{Vector addition systems} (VAS) are the main focus of this paper. A $d$-VAS is a $d$-VASS $V = (Q, T)$ where $Q$ contains a single state only. Transitions in $T$ are determined by its effect, and we can view $T$ as a subset of $\mathbb{Z}^d$. The unique state is often omitted. So a configuration is regarded as a vector in $\mathbb{N}^d$.

In the analysis of VAS runs, we may identify some transitions that can be fired only a bounded number of times. In order to distinguish these transitions from others, we introduce the model of \emph{sequential VAS}. Let $V$ be a $d$-VAS whose transition set is $A \subseteq \mathbb{Z}^d$, and let $\Vec{a}_1, \ldots, \Vec{a}_k \in \mathbb{Z}^d$ be vectors. The sequential VAS $S=V[\Vec{a}_1, \ldots, \Vec{a}_k]$ is defined to be the VASS $(Q, T)$ where $Q := \{ q_0, q_1, \ldots, q_k \}$ and 
\begin{equation}
  T := \{ (q_j, \Vec{a}, q_j) \mid j \in [0, k], \Vec{a} \in A \} \cup \{ (q_{j - 1}, \Vec{a}_j, q_j) \mid j \in [k] \}.
\end{equation}
We call $q_0$ the \emph{source state} and $q_k$ the \emph{target state} in $S$. Vectors $\Vec{a}_1, \ldots, \Vec{a}_k$ are called the \emph{bridges} in $S$.
{By a source/target configuration, we mean a configuration whose state is the source/target state in $S$.}
Intuitively, $V[\Vec{a}_1, \ldots, \Vec{a}_k]$ can be viewed as the VAS $V$ enhanced with $k$ special transitions which can be fired exactly once in the order they are specified.

\paragraph*{Geometric Dimension and Reachability Problems}

Let $V$ be a $d$-VASS. Its \emph{cycle space} $\CycleSpace(V)\subseteq \mathbb{Q}^d$ is the vector space spanned by the effects of simple cycles in $V$, i.e., $\CycleSpace(V) := \Span\{\Delta(\theta) \mid \theta \text{ is a simple cycle in } V\}$. Geometric dimension was introduced recently as an important parameter of VASSes. Here, we consider two types of geometric dimension:
The \emph{cycle-dimension} 
$
  \dimcyc(V) := \dim(\CycleSpace(V))
$
is the dimension of the whole cycle space. This was known as \emph{geometric dimension} in \cite{DBLP:conf/icalp/FuYZ24} and \cite{DBLP:conf/concur/Zheng25}.
The \emph{component-dimension} 
$
  \dimcom(V) := \max_{\text{SCC } V' \text{ in } V} \dim(\CycleSpace(V'))
$
is the maximum value among the dimensions of the cycle spaces of strongly connected components in $V$. This coincides with the dimension introduced in \cite{DBLP:conf/lics/GuttenbergCL25}.
By convention, we say that a VASS $V$ is geometrically $d$-dimensional if $\dimcyc(V)\le d$. Notice that for a VAS $V$, both $\dimcyc(V)$ and $\dimcom(V)$ are equal to the dimension of the vector space spanned by the effects of transitions of $V$, as every transition is a self-loop. {We refer the readers to \cite{exploring-geo} for further research about these two parameters.}

The reachability problem in VASSes can be formulated as follows:
\begin{quotation}
    \begin{description}
        \item[Input:] A VASS $V = (Q, T)$ with configurations $p(\Vec{x}), q(\Vec{y})$.
        \item[Question:] Does it hold $p(\Vec{x}) \xrightarrow{*} q(\Vec{y})$?
    \end{description}
\end{quotation}
Complexity bounds for reachability problems are usually obtained from length bounds on runs witnessing the reachability. Let $V$ be a VASS and $s, t$ be two configurations of $V$. {We also call the triple $(V, s, t)$ a VASS, with $s$ and $t$ understood as the source and target configurations.} We write $\Len(V, s, t)$ for the set of lengths of runs from $s$ to $t$, and $\Size(V, s, t)$ for $\Size(V) + d\cdot (\norm{s} + \norm{t} + 1)$. The following bound can be obtained from the famous KLMST decomposition algorithm.

\begin{theorem}[{\cite[Lemma 6.2]{DBLP:conf/icalp/FuYZ24}}]\label{thm:d-vass-in-fd}
  Let $(V,s,t)$ be a VASS of component-dimension $m \ge 3$. Then $\min \Len(V, s, t) \le F_m(g(\Size(V, s, t)))$ if $\Len(V, s, t) \ne \emptyset$, where $g$ is some elementary function independent of $m$, and $(V,s,t)$.
\end{theorem}

The function $F_m$ is the $m$-th fast-growing function, which will be introduced below. We remark that the statement of this theorem in \cite{DBLP:conf/icalp/FuYZ24} is for fixed dimension $m$ rather than fixed component-dimension $m$. The version we use here can be easily obtained from the ranking function defined in \cite{DBLP:conf/icalp/FuYZ24}.

\paragraph*{Fast-Growing Complexity Hierarchy}

The fast-growing functions $(F_k)_{k \in \mathbb{N}}$ are defined as follows:
$F_0(x) := x + 1, \quad F_{k+1}(x) := F_k^{x+1}(x)$,
where $F_k^{x+1}$ means the $(x+1)$-fold iteration of $F_k$. For example, $F_2(x) = 2^{x+1}(x+1)-x$ is exponential in $x$, and $F_3(x)$ grows even faster than a tower of exponents. We say a function is \emph{elementary} if it is bounded by $F_2^c$ for some constant $c$. Based on these functions we may define a hierarchy $(\mathscr{F}_k)_{k\in \mathbb{N}}$ of computable functions: $\mathscr{F}_k := \bigcup_{c \in \mathbb{N}} \textsf{FDTIME}(F_k^c(n))$.
The fast-growing complexity hierarchy is defined as follows:
$\textsf{F}_k := \bigcup_{p \in \mathscr{F}_{k-1}} \textsf{DTIME}(F_k(p(n))).$
For more details of this complexity hierarchy---which we will not use for this paper---we refer the readers to \cite{DBLP:journals/toct/Schmitz16}.

\paragraph*{Wideness, Pumpability and Fixed Coordinates}

We have introduced the cycle space of a VASS, which is the vector space spanned by the effects of simple cycles. However, it is the \emph{cone} generated by the effects of simple cycles that better describes the behavior of the VASS. Let $V$ be a VASS and $E = \{\Vec{e}_1, \ldots, \Vec{e}_k\}$ be the set of effects of simple cycles in $V$. We define the cone of $V$ as $\Cone(V) := \Cone(E)$.
And $V$ is said to be \emph{wide} if $\Cone(V) = \CycleSpace(V)$, so any combination of cycle effects can be expressed using non-negative coefficients. We remark that the notion of wideness was proposed in \cite{DBLP:conf/icalp/CzerwinskiJ0O25} in a slightly different manner.

Pumpability is another important notion that we will investigate. Let $p(\Vec{x})$ be a configuration in $V$. We say $p(\Vec{x})$ is \emph{forward pumpable} (or simply \emph{pumpable}) in $V$ if there is a run $p(\Vec{x}) \xrightarrow{*} p(\Vec{x}')$ with $\Vec{x}' \ge \Vec{x} + \Vec{1}$. We say $p(\Vec{x})$ is \emph{backward pumpable} in $V$ if it is forward pumpable in $V^{\text{rev}}$, or equivalently, if there is a run $p(\Vec{x}') \xrightarrow{*} p(\Vec{x})$ in $V$ with $\Vec{x}' \ge \Vec{x} + \Vec{1}$. 

Let $S=V[\Vec{a}_1,\dots,\Vec{a}_k]$ be a sequential $d$-VAS. We say that a coordinate $i \in [d]$ is \emph{fixed} in $S$ if $\Vec{u}(i)=0$ holds for every transition $\Vec{u}$ in $V$. Fixed coordinates are annoying from time to time. We may use the following easy lemma to get rid of them.

\begin{restatable}{lemma}{existsSequentialVasWithoutFixedCoordinate}\label{lem:exists-sequential-vas-without-fixed-coordinate}
    Let $(S,s,t)$ be a sequential $d$-VAS. Then there exists a sequential $d'$-VAS $(S',s',t')$ containing no fixed coordinates such that $d'\le d$, $\Size(S',s',t')\le \Size(S,s,t)$, and $\Len(S,s,t)=\Len(S',s',t')$.
\end{restatable}

Unless otherwise specified, we may assume without loss of generality that the sequential $d$-VAS in the subsequent discussion contains no fixed coordinates, i.e., there exists at least one transition $\Vec{u}_i$ with $\Vec{u}_i(i)\ne 0$ in $V$ for each $i \in [d]$.

\section{Pumpability from \texorpdfstring{$d - 1$}{d-1} Unbounded Counters}
\label{sec:pumpability}

\emph{Rackoff's extraction} \cite[Lemma A.1]{DBLP:conf/lics/LerouxS19} (see also \cite{DBLP:journals/tcs/Rackoff78}, \cite[Lemma 5]{DBLP:conf/icalp/CzerwinskiJ0O25}, and \cite{DBLP:journals/jacm/KunnemannMSSW25}) is an important technique to extract pumpability or coverability from an ``unbounded run''. In the scenario of $d$-VASS runs, this technique states, intuitively, that if there is a run where each of the $d$ coordinates has ever been raised above some sufficiently large value, then one can construct a run from the same source whose target has large values in all $d$ coordinates. A consequence of this fact is that, within one strongly connected component of VASS, at least one coordinate on a run from some unpumpable configuration must remain bounded. Dimension reduction can then be conducted for reachability instances starting from an unpumpable configuration.

In this section, we improve Rackoff's extraction for wide sequential $d$-VASes. We show that coverability or pumpability is guaranteed if \emph{all but one} coordinates can be raised above some large value in a run. This will give us at least two bounded coordinates in the unpumpable case. Let us now state the result formally. Let $\pi$ be a run in some $d$-VASS $V$ and $B \in \mathbb{N}$, and we introduce the notation $\UBCounters(\pi, B)$ for the set of indices $i \in [d]$ such that $\pi$ contains a configuration  $q_i(\Vec{x}_i)$ with $\Vec{x}_i(i) \ge B$. That is, $\UBCounters(\pi, B)$ are those coordinates that ever reached a value above $B$ along $\pi$. Our improved Rackoff's extraction is stated as follows.

\begin{theorem}
    \label{thm:reach-pumpable-of-ever-unbounded-d-1}
    There is a polynomial $P$ such that, given a wide sequential $d$-VAS $S = V[\Vec{a}_1, \ldots, \Vec{a}_k]$ with source {configuration $q_0(\Vec{x}_0)$}, if there is a run $\pi$ from $q_0(\Vec{x}_0)$ such that $|\UBCounters(\pi, U)| \ge d - 1$ where $U := P(\norm{\Vec{x}_0}, \Size(S))^{(d+1)!}$, then there is a run $\rho$ with $|\rho| \le U$ such that $q_0(\Vec{x}_0) \xrightarrow{\rho} q(\Vec{x})$ for some configuration $q(\Vec{x})$ pumpable in $S$. Moreover, $\rho$ can be obtained from $\pi$ by removing some self-loops.
\end{theorem}

We remark that the ``moreover'' part of this theorem guarantees that the target of $\pi$ is reachable from the pumpable configuration $q(\Vec{x})$ by a $\mathbb{Z}$-run consisting of transitions in $V$.

First, we show that the pumpability of $q(\Vec{x})$ is guaranteed if $d-1$ coordinates in $\Vec{x}$ are larger than some polynomial bound. The trick is that wideness ensures the existence of a cycle with a positive effect in the remaining coordinate, and this cycle is actually a self-loop in the sequential VAS, which can be fired whenever other coordinates have large values.

\begin{lemma}
    \label{lem:tgt-pumpable-of-unbounded-d-1}
    Let $S = V[\Vec{a}_1, \ldots, \Vec{a}_k]$ be a wide sequential $d$-VAS with source {configuration $q_0(\Vec{x}_0)$}. If there is a run $q_0(\Vec{x}_0) \xrightarrow{*} q(\Vec{x})$ for some state $q$ in $S$ such that $\Vec{x}(i) \ge B(\norm{\Vec{x}_0}, \Size(S))$ for all $i \in [d - 1]$, then $q(\Vec{x})$ is pumpable in $S$, where 
    \begin{equation}
        B(\norm{\Vec{x}_0}, \Size(S)) := (\norm{\Vec{x}_0} + 1 + \Size(S)) \cdot (1+\Size(S)).
    \end{equation}
\end{lemma}

\begin{proof}
    Let the states of $S$ be $q_0, q_1, \ldots, q_k$. We expand the run $q_0(\Vec{x}_0) \xrightarrow{*} q(\Vec{x})$ as 
    \begin{equation}
        \begin{aligned}
            \pi: q_0(\Vec{x}_0) \xrightarrow{\pi_0} q_0(\Vec{y}_0) \xrightarrow{\Vec{a}_1} 
        q_1(\Vec{x}_1) \xrightarrow{\pi_1} & q_1(\Vec{y}_1) \xrightarrow{\Vec{a}_2}
        \cdots \xrightarrow{\Vec{a}_\ell} 
        q_\ell(\Vec{x}_\ell) \xrightarrow{\pi_\ell} & q_\ell(\Vec{y}_\ell) = q(\Vec{x})
        \end{aligned}        
    \end{equation}
    for some $\ell \in [0, k]$, where each $\pi_j$ is a path in the VAS $V$.
    
    We have assumed that $S$ is wide and has no fixed coordinates (recall \autoref{lem:exists-sequential-vas-without-fixed-coordinate}), and so is $V$. Therefore, there must be a transition (i.e.\ a self-loop) $\Vec{t}$ in $V$ such that $\Vec{t}(d) > 0$. 
    Let $H := \norm{\Vec{x}_0} + 1 + \Size(S)$. We claim that $q(\Vec{x}) \xrightarrow{\Vec{t}^{H}} q(\Vec{x} + H \cdot \Vec{t})$ is a legal run. To see this, one just needs to verify the target $\Vec{x} + H \cdot \Vec{t}$ is non-negative, as all the intermediate vectors on this path are a convex combination of $\Vec{x}$ and $\Vec{x}+H\cdot\Vec{t}$. For $i \in [d - 1]$, from $\Vec{x}(i) \ge B(\norm{\Vec{x}_0}, \Size(S))$ we have 
    \begin{equation}
        \Vec{x}(i) + H \cdot \Vec{t}(i) \ge H \cdot (1 + \Size(S)) - H \cdot \Size(S) \ge H \ge 0.
    \end{equation}
    And for $i = d$, recall that $\Vec{t}(d) > 0$, so $\Vec{x}(d) + H \cdot \Vec{t}(d) \ge H \ge 0.$ Now we have justified that $q(\Vec{x}) \xrightarrow{\Vec{t}^{H}} q(\Vec{x} + H \cdot \Vec{t})$ is a legal run, and moreover, the target $\Vec{x}' := \Vec{x} + H \cdot \Vec{t}$ satisfies $\Vec{x}'(i) \ge H$ for all $i \in [d]$. Regarding the effect of $q_0(\Vec{x}_0) \xrightarrow{\pi} q(\Vec{x}) \xrightarrow{\Vec{t}^H} q(\Vec{x}')$, we have 
    \begin{equation}
        \Delta(\pi_0 \ldots \pi_\ell \Vec{t}^H) + \sum_{i \in [\ell]}\Vec{a}_i = \Vec{x}' - \Vec{x}_0 \ge (H - \norm{\Vec{x}_0}) \cdot \Vec
        1 \ge (\Size(S) + 1) \cdot \Vec{1}. 
    \end{equation}
    Here we slightly abuse the notation, viewing $\pi_1, \ldots, \pi_\ell$ as paths in the VAS $V$, so they can be concatenated.
    Notice that $\norm{\Vec{a}_1 + \cdots + \Vec{a}_\ell} \le \Size(S)$. Let $\rho := \Vec{t}^H \pi_0 \pi_1 \ldots \pi_\ell$ be a cycle in $V$, we then have $\Delta(\rho) \ge \Vec{1}$. Next, we show that $\rho$ is fireable from $q(\Vec{x})$.
    
    Indeed, we have shown that the prefix $\Vec{t}^H$ induces a legal run $q(\Vec{x}) \xrightarrow{\Vec{t}^{H}} q(\Vec{x}')$. It remains to show that $\rho' := \pi_0 \pi_1 \ldots \pi_\ell$ is fireable from $\Vec{x}'$ in $V$.
    Notice that $\rho'$ is obtained from $\pi$ by removing all the bridges $\Vec{a}_1, \ldots, \Vec{a}_\ell$. Since $\pi$ is fireable from $q_0(\Vec{x}_0)$, we know that $\Drop(\pi) \le \norm{\Vec{x}_0} \cdot \Vec{1}$. Using the fact that $\norm{\Vec{a}_1 + \cdots + \Vec{a}_j} \le \Size(S)$ for any $j$, we deduce $\Drop(\rho') \le \Drop(\pi) + \Size(S) \cdot \Vec{1} \le (\norm{\Vec{x}_0} + \Size(S)) \cdot \Vec{1} \le H \cdot \Vec{1} \le \Vec{x}'$. This shows that $\rho'$ is fireable from $\Vec{x}'$.
\end{proof}

\autoref{lem:tgt-pumpable-of-unbounded-d-1} reduces pumpability to a special kind of coverability, in the sense that only $d-1$ coordinates need to cover some large value simultaneously. The following lemma, which is similar to the standard Rackoff's extraction, further reduces the coverability condition to covering a larger value in each coordinate separately. We remark that the proof, which is deferred to the appendix, has a similar fashion to \cite[Lemma 5]{DBLP:conf/icalp/CzerwinskiJ0O25}.

\begin{restatable}{lemma}{lemReachUnboundedofEverUnbounded}
    \label{lem:reach-unbounded-of-ever-unbounded}
    There exists a polynomial $R$ such that, given a $d$-VASS $V$ and $U \in \mathbb{N}$, for any run $\pi$ in $V$, there is a run $\rho$ starting from the same source as $\pi$ with $|\rho| \le R(U, \Size(V))^{(d+1)!}$ such that its target $q(\Vec{x})$ satisfies $\Vec{x}(i) \ge U$ for all $i \in \UBCounters(\pi, R(U, \Size(V))^{(d+1)!})$. Moreover, $\rho$ can be obtained from $\pi$ by removing some cycles.
\end{restatable}

Now we can combine these two lemmas to establish \autoref{thm:reach-pumpable-of-ever-unbounded-d-1}.

\begin{proof}[Proof of \autoref{thm:reach-pumpable-of-ever-unbounded-d-1}]
    Let $B$ be the polynomial in \autoref{lem:tgt-pumpable-of-unbounded-d-1} and $R$ be the polynomial defined in \autoref{lem:reach-unbounded-of-ever-unbounded}. We may take the polynomial $P$ defined by 
    \begin{equation}
        P(\norm{\Vec{x}_0}, \Size(S)) := R(B(\norm{\Vec{x}_0}, \Size(S)), \Size(S)).
    \end{equation}
    Let $\pi$ be a run from $q_0(\Vec{x}_0)$ in $S$ such that $|\UBCounters(\pi, U)| \ge d - 1$. By \autoref{lem:reach-unbounded-of-ever-unbounded} there is a run $\rho$ with $|\rho| \le P(\norm{\Vec{x}_0}, \Size(S))^{(d+1)!}$ such that $q_0(\Vec{x}_0) \xrightarrow{\rho} q(\Vec{x})$ for some $q(\Vec{x})$ with $\Vec{x}(i) \ge B(\norm{\Vec{x}_0}, \Size(S))$ for all $i \in \UBCounters(\pi, U)$. We may assume that $[d-1] \subseteq \UBCounters(\pi, U)$ by renaming the coordinates. Therefore by \autoref{lem:tgt-pumpable-of-unbounded-d-1} the configuration $q(\Vec{x})$ is pumpable in $S$. Finally, since each cycle in $S$ consists only of self-loops, $\rho$ can be obtained from $\pi$ by removing self-loops.
\end{proof}

\section{Bounding Shortest Runs of \texorpdfstring{$d$}{d}-VAS}
\label{sec:d-vas}

In this section, we prove our first main result, that the length of the shortest runs in a $(d+2)$-VAS can be bounded by that in a VASS with component-dimension at most $d$.

\begin{theorem}
    \label{thm:len-d-plus-2-vas-le-len-vass-max-dim-d}
    For every $d \in \mathbb{N}$, there is a polynomial $P_d$ such that, for every $(d + 2)$-VAS $(V,\Vec{s},\Vec{t})$, there is a VASS $(V',s',t')$ with $\dimcom(V') \le d$ such that $\Size(V', s', t') \le P_d(\Size(V, \Vec{s}, \Vec{t}))$, and $\min \Len(V, \Vec{s}, \Vec{t}) \le \min \Len(V', s', t')$.
\end{theorem}

If $d\ge 3$, by \autoref{thm:d-vass-in-fd}, we can further deduce that the length of shortest reachability witnesses for $(V,\Vec{s},\Vec{t})$ must be bounded by $F_d(g(\Size(V,\Vec{s},\Vec{t})))$, where $g$ is an elementary function. All such runs can be enumerated and checked in $F_d(g(n))$-space. The following improved complexity upper bound follows immediately.

\begin{restatable}{theorem}{complexityDplusTwo}
    \label{thm:reach-d-plus-2-in-Fd}
    Reachability in $(d+2)$-VAS is in $\mathsf{F}_d$ for every $d\ge 3$.
\end{restatable}

The pattern of statement in \autoref{thm:len-d-plus-2-vas-le-len-vass-max-dim-d} will appear frequently in the following sections. For ease of narration, we shall introduce the following definition.

\begin{definition}
    Let $f : \mathbb{N} \to \mathbb{N}$ be a function. 
    \begin{itemize}
        \item A VASS $(V, s, t)$ is said to be \emph{length-bounded} by a VASS $(V', s', t')$ \emph{with $f$-amplification}, if $\min\Len(V, s, t) \le \min\Len(V', s', t')$ and $\Size(V', s', t') \le f(\Size(V, s, t))$, and $s \xrightarrow{*} t$ in $V$ if and only if $s' \xrightarrow{*} t'$ in V'.
        \item A VASS $(V, s, t)$ is said to be \emph{length-captured} by a family $\mathcal{V}$ of VASSes $(V', s', t')$ \emph{with $f$-amplification}, if $\Len(V, s, t) = \bigcup_{(V', s', t') \in \mathcal{V}} \Len(V', s', t')$ and $\Size(V', s', t') \le f(\Size(V, s, t))$ for all $(V', s', t') \in \mathcal{V}$.
    \end{itemize}
\end{definition}

Hence, \autoref{thm:len-d-plus-2-vas-le-len-vass-max-dim-d} could be reformulated as follows:

\begin{theorem}[\autoref{thm:len-d-plus-2-vas-le-len-vass-max-dim-d}]
    For every $d \in \mathbb{N}$, there is a polynomial $P_d$ such that, every $(d+2)$-VAS $(V, \Vec{s}, \Vec{t})$ is length-boounded by a VASS $(V', s', t')$ with $P_d$-amplification, where $\dimcom(V') \le d$.
\end{theorem}

Notice that we may always assume a run $s \xrightarrow{\pi} t$ exists in the VASS $(V, s, t)$, as otherwise we can take $(V', s', t')$ to be an unreachable instance with constant size. In the following, we tacitly adopt this assumption.

\subsection{Non-Wide VAS to Wide Sequential VAS}

We will actually prove \autoref{thm:len-d-plus-2-vas-le-len-vass-max-dim-d} for wide sequential VASes. Here, we show how to convert a non-wide VAS to a family of wide sequential VASes.

\begin{lemma}
    \label{lem:non-wide-to-family-of-wide-seq}
    For every $d \in \mathbb{N}$, there is a polynomial $W_d$ such that, every $d$-VAS $(V, \Vec{s}, \Vec{t})$ is length-captured by a family $\mathcal{S}$ of wide sequential $d$-VASes $(S, s', t')$ with $W_d$-amplification. Moreover, if $V$ is non-wide, then $\dimcyc(S) \le \dimcyc(V) - 1$.
\end{lemma}

The proof idea is similar to \cite[Sect.\ 4, Case 2]{DBLP:conf/icalp/CzerwinskiJ0O25}. If $V$ is non-wide, then there exists a vector $\Vec{n}$ such that every transition belongs to the half-space determined by $\Vec{n}$. In particular, the inner product of $\Vec{n}$ with the configuration cannot be decreased by any transition. Hence, every transition whose effect has a non-zero inner product with $\Vec{n}$ can only be fired for a bounded number of times determined by $\Inner{\Vec{t} - \Vec{s}, \Vec{n}}$. We then enumerate every combination of these transitions and expand $V$ as a sequential VAS correspondingly. What is crucial here is that we need to find a good $\Vec{n}$ such that the VAS with transitions orthogonal to $\Vec{n}$ is wide. For this purpose, we introduce the following lemma.

\begin{restatable}{lemma}{existsNormalWideOrtho}
    \label{lem:exists-normal-vector-with-wide-ortho}
    Let $X \subseteq \mathbb{Z}^d$ be a finite set of vectors {such that $\Cone(X) \ne \Span(X)$}. There exists a {non-zero} vector $\Vec{n} \in \mathbb{Z}^d$ such that 
    \begin{enumerate}
        \item \label{item:normal-pos} $\Inner{\Vec{n}, \Vec{x}} \ge 0$ for all $\Vec{x} \in X$;
        \item \label{item:normal-wide} let $X_0 := \{ \Vec{x} \in X \mid \Inner{\Vec{n}, \Vec{x}} = 0 \}$; then $\Cone(X_0) = \Span(X_0)$;
        \item \label{item:normal-bound} $\norm{\Vec{n}} \le (d + 1)\cdot (r\norm{X})^r$, where $r := \dim(\Span(X))$.
    \end{enumerate}
\end{restatable}

With the help of \autoref{lem:exists-normal-vector-with-wide-ortho}, we are able to identify the set of bounded transitions. A VAS can then be turned into a wide sequential VAS by making these transitions as bridges.

\begin{proof}[Proof of \autoref{lem:non-wide-to-family-of-wide-seq}]
    Let $V$ be a $d$-VAS with transitions $T \subseteq \mathbb{Z}^d$. 
    {If $V$ is wide, then we are already done. Otherwise,}
    applying \autoref{lem:exists-normal-vector-with-wide-ortho} on $T$, we get a {non-zero} vector $\Vec{n} \in \mathbb{Z}^d$ with $\norm{\Vec{n}} \le D := (d + 1) \cdot (d\norm{T})^{d}$. Moreover, we may partition $T$ into $T_1\ne \emptyset$ and $T_0$ where 
    \begin{align}
        T_1 := \{ \Vec{t} \in T \mid \Inner{\Vec{n}, \Vec{t}} > 0 \}, \quad
        T_0 := \{ \Vec{t} \in T \mid \Inner{\Vec{n}, \Vec{t}} = 0 \}.
    \end{align}
    Any transition in $T_1$ increases the inner product with $\Vec{n}$ while no transition can decrease this inner product. Fix configurations $\Vec{s}, \Vec{t} \in \mathbb{N}^d$. Consider a run $\Vec{s} \xrightarrow{\pi} \Vec{t}$. The number of occurences of transitions from $T_1$ on $\pi$ is bounded by 
    \begin{align}
        K := \Inner{\Vec{n}, \Vec{t}} - \Inner{\Vec{n}, \Vec{s}} \le 2 D \cdot \Size(V, \Vec{s}, \Vec{t}).
    \end{align}
    We take $\mathcal{S}$ to be the family of sequential $d$-VAS $S := V_0[\Vec{a}_1, \ldots, \Vec{a}_k]$ satisfying $k \le K$ and $\Vec{a}_j \in T_1$ for all $j \in [k]$, where $V_0$ is the VAS consists of transitions in $T_0$. Clearly $\mathcal{S}$ captures every run from $\Vec{s}$ to $\Vec{t}$. Every run in some $S \in \mathcal{S}$ also induces a run in $V$ as $S$ contains only transitions from $V$. Moreover, \autoref{lem:exists-normal-vector-with-wide-ortho} guarantees that $V_0$ is wide, and so is every $S \in \mathcal{S}$.
    Notice that $\CycleSpace(S) = \Span(T_0) \subsetneq \Span(T) = \CycleSpace(V)$, hence $\dimcyc(S) \le \dimcyc(V) - 1$.
    Considering the size amplification, we have $\Size(S) \le (K + 1) \cdot \Size(V) \le 3(d+1) \cdot \Size(V, \Vec{s}, \Vec{t})^{d + 1}$. We are done by taking $W_d(x) := x+3(d+1)\cdot x^{d+1}$.
\end{proof}

We remark that after the decomposition presented in the proof of \autoref{lem:non-wide-to-family-of-wide-seq}, we essentially obtain a family of wide sequential $d$-VASes, possibly with several fixed coordinates. However, by applying \autoref{lem:exists-sequential-vas-without-fixed-coordinate}, we shall assume that these new VASSes contain no fixed coordinates, and the dimension might be less than $d$.

\subsection{Boundedness of Runs}
\label{subsec:boundedness}

Now we focus on runs in wide sequential VAS. We shall identify two cases depending on whether \autoref{thm:reach-pumpable-of-ever-unbounded-d-1} can be applied to extract forward and backward pumpable configurations.

Let $\pi$ be a run in some $d$-VASS $V$ and $B \in \mathbb{N}$. Recall that $\UBCounters(\pi, B)$ is the set of indices $i \in [d]$ such that $\pi$ contains a configuration $q_i(\Vec{x}_i)$ with $\Vec{x}_i(i) \ge B$. We say $\pi$ is \emph{$k$-bounded} by $B$ if $|\UBCounters(\pi_1, B)|\le d-k$. In the literature of VASS, $1$-boundedness is often used to perform dimension reduction. For wide sequential VAS, \autoref{thm:reach-pumpable-of-ever-unbounded-d-1} allows us to shift the focus to $2$-boundedness. For technical reasons, we shall make a slight generalization here. We say that $\pi$ is \emph{almost ($2$-)bounded by $B$} if there exists a partition $\pi=\pi_1\pi_2$ such that both $\pi_1$ and $\pi_2$ are $2$-bounded by $B$. On the other hand, $\pi$ is \emph{almost unbounded by $B$} if there exists a partition $\pi=\pi_1\pi_2$ such that $|\UBCounters(\pi_1, B)|, |\UBCounters(\pi_2, B)| \ge d-1$. 
Intuitively, almost bounded runs allow us to perform dimension reduction, while almost unbounded runs permit us to apply \autoref{thm:reach-pumpable-of-ever-unbounded-d-1} to extract forward and backward pumpable configurations.
Although these two conditions do not appear to be complementary, a dichotomy between them still persists.

\begin{restatable}{lemma}{almostBoundedOrUnbounded}
    \label{lem:almost-bounded-or-unbounded}
    Suppose that $(V,s,t)$ is a $d$-VASS with $s\xrightarrow{\pi}t$ in $V$ and $B \in \mathbb{N}$. Then either $\pi$ is almost unbounded by $B$, or $\pi$ is almost bounded by $B+\Size(V)$.
\end{restatable}

In the following, we handle almost unbounded runs and almost bounded runs in wide sequential VAS separately.

\subsubsection{Almost Unbounded Runs}

In this part, we address almost unbounded runs in wide sequential VASes. We show that in this case, reachability is relatively easy to verify in the sense that a polynomial upper bound can be established on the length of shortest witnesses for such runs.

\begin{lemma}
    \label{lem:poly-bound-for-almost-unbounded}
    For every $d \in \mathbb{N}$, there exists a polynomial $U_d$ such that, given a wide sequential $d$-VAS $(S,s,t)$ with $s \xrightarrow{\pi}t$ in $S$, if $\pi$ is almost unbounded by $U_d(\Size(S, s, t))$, then there exists a run $s \xrightarrow{\rho} t$ in $S$ of length $|\rho|\le U_d(\Size(S, s, t))$.
\end{lemma}

The proof idea is summarized as follows. If a run $\pi$ is almost unbounded by a sufficiently large value, then $\pi$ can be decomposed into two segments $\pi_1\pi_2$ such that, in both segments, at least $d-1$ coordinates go beyond the given threshold. Recall that the pumping technique presented in~\autoref{sec:pumpability} allows us to construct a $\mathbb{Z}$-run of the form 
\begin{equation}
    \label{eq:lifting}
    s\xrightarrow{\rho_1}q_{1}\left(\boldsymbol{x}_1\right)\xrightarrow{}_\mathbb{Z} q_{2}\left(\boldsymbol{x}_2\right) \xrightarrow{\rho_2} t,
\end{equation}
where $q_1\left(\boldsymbol{x}_1\right)$ and $q_2\left(\boldsymbol{x}_2\right)$ are forward pumpable and backward pumpable, respectively. Furthermore, the lengths of $\rho_1$ and $\rho_2$ are bounded by a polynomial of $\Size(S,s,t)$. It suffices to show that the second $\mathbb{Z}$-run can be replaced by a short run. A similar result has been established for strongly connected $3$-VASS \cite[Sect.\ 4, Case 1]{DBLP:conf/icalp/CzerwinskiJ0O25}. 
The proof consists of two steps: (i) replacing an arbitrary $\mathbb{Z}$-run by a short one coming from an integer linear program, and (ii) lifting the resulting $\mathbb{Z}$-run to sufficiently high via pumping cycles so that it becomes a run. Although the configurations $q_1(\Vec{x}_1)$ and $q_2(\Vec{x}_2)$ may be pumped in different directions, the wideness of $S$ ensures that an appropriate combination of transitions can eventually compensate for this discrepancy. These two steps are realized in the following lemma, whose proof can be found in the appendix.

\begin{restatable}{lemma}{liftZRunByPumpable}
    \label{lem:lift-Z-run-by-pumpable}
    Let $S=V[\Vec{a}_1,\dots,\Vec{a}_k]$ be a wide sequential $d$-VAS with source state $p$ and target state $q$. Let $p(\Vec{x})\xrightarrow{*}_\mathbb{Z} q(\Vec{y})$ be a $\mathbb{Z}$-run in $S$, where $p(\Vec{x})$ is forward pumpable and $q(\Vec{y})$ is backward pumpable. Then $p(\Vec{x})\xrightarrow{\rho}q(\Vec{y})$ for some $\rho$ in $S$ with $|\rho|\le L_d(\Size(V, p(\Vec{x}), q(\Vec{y})))$,  where $L_d$ is a polynomial independent of $V$, $p(\Vec{x})$, and $q(\Vec{y})$.
\end{restatable}

We remark that the wideness is essential in the proof of \autoref{lem:lift-Z-run-by-pumpable}. Next, we proceed to the proof of \autoref{lem:poly-bound-for-almost-unbounded}.

\begin{proof}[Proof of \autoref{lem:poly-bound-for-almost-unbounded}]
    Let $P$ be the polynomial defined in \autoref{thm:reach-pumpable-of-ever-unbounded-d-1} and $L_d$ be the polynomial defined in \autoref{lem:lift-Z-run-by-pumpable}. We take $U_d$ as the following polynomial
    \begin{equation}
        U_d(x):=2P(x, x)^{(d+1)!}+L_d\left(2x\cdot P(x, x)^{(d+1)!}+x\right).
    \end{equation}
    Let $M:=\Size(S,s,t)$, and $N:=P(M,M)^{(d+1)!}$. For any run $\pi$ almost unbounded by $U_d(M)$, by definition, there exists a configuration $q(\Vec{w})$ such that $\pi$ can be splited into $s \xrightarrow{\pi_1}q(\Vec{w}) \xrightarrow{\pi_2} t$, where $|\UBCounters(\pi_i, U_d(M))|\ge d- 1$ holds for $i=1,2$. Note that $U_d(M)\ge N \ge P(\norm{s}, \Size(S))^{(d+1)!}$. Applying \autoref{thm:reach-pumpable-of-ever-unbounded-d-1} to $\pi_1$, one can obtain a run $\rho_1$ with $|\rho_1|\le P(\norm{s}, \Size(S))^{(d+1)!}\le N$ such that $s \xrightarrow{\rho_1} q_1(\Vec{x}_1)$ for some forward pumpable configuration $q_1(\Vec{x})$ in $S$. Moreover, since $\rho$ can be obtained from $\pi$ by removing some self-loops, we have $q_1(\Vec{x}_1) \xrightarrow{*}_{\mathbb{Z}} q(\Vec{w})$. A symmetric argument to $\pi^{\text{rev}}$ yields that $q(\Vec{w})\xrightarrow{*}_{\mathbb{Z}} q_2(\Vec{x}_2)\xrightarrow{\rho_2}t$, where $|\rho_2|\le N$ and $q_2(\Vec{x}_2)$ is backward pumpable in $S$. Then we can merge these two $\mathbb{Z}$-runs as
    \begin{equation}
        s\xrightarrow{\rho_1} q_1(\Vec{x}_1)\xrightarrow{*}_{\mathbb{Z}} q_2(\Vec{x}_2) \xrightarrow{\rho_2} t.
    \end{equation}
    The sequential VAS $S'$ induced by $q_1(\Vec{x}_1)\xrightarrow{*}_{\mathbb{Z}}q_2(\Vec{x}_2)$ has size at most $M + 2NM$. The configuration $p(\Vec{x})$ is still forward pumpable in $S'$, and $q(\Vec{y})$ is still backward pumpable in $S'$. Applying \autoref{lem:lift-Z-run-by-pumpable} to $S'$ and this $\mathbb{Z}$-run, we have $q_1(\Vec{x}_1)\xrightarrow{\rho'}q_2(\Vec{x}_2)$ with $|\rho'|\le L_d(M + 2NM)$. Let $\rho:=\rho_1\rho'\rho_2$. Then $s\xrightarrow{\rho}t$ with $|\rho|\le 2N+L_d(M + 2NM)= U_d(M)$.
\end{proof}

\subsubsection{Almost Bounded Runs}
For $k$-bounded runs, it is standard to encode the $k$ bounded coordinates into states so that they can be captured by a new VASS with lower (geometric) dimension.

\begin{restatable}{lemma}{reductionForBounded}
    \label{lem:reduction-for-bounded}
    Let $B\in\mathbb{N}$ and $(V,s,t)$ be a $d$-VASS with $s\xrightarrow{\pi}t$ in $V$, which is $k$-bounded by $B$. Then $(V, s, t)$ is length-bounded by a $d$-VASS $(V',s',t')$ with $f$-amplication, where $\dimcom(V')\le d - k$ and $f(x) = B^k \cdot x$.
\end{restatable}

However, for an almost bounded run $\pi$ considered here, this dimension reduction technique cannot be applied directly, since $\pi$ consists of two bounded segments $\pi_1$ and $\pi_2$. Instead, we construct two new VASSes according to $\pi_1$ and $\pi_2$, respectively, and then concatenate them. The resulting VASS captures $\pi$ with a strictly smaller component-dimension.

\begin{restatable}{lemma}{dimensionReductionForAlmostBoundedRuns}
    \label{lem:dimension-reduction-for-almost-bounded-runs}
    For each $d \in \mathbb{N}$, there is a polynomial $A_d$ such that, for every $B\in\mathbb{N}$ and $(d+2)$-VASS $(V,s,t)$ with $s\xrightarrow{\pi}t$ in $V$, if $\pi$ is almost bounded by $B$, then $(V, s, t)$ is length-bounded by a $(d+2)$-VASS $(V', s', t')$ with $A_d$-amplification, where $\dimcom (V')\le d$.
\end{restatable}

The proofs of both lemmas are quite standard. We defer them to the appendix for space limitations.
So far, we have analyzed both cases of (almost) boundedness. We now return to the proof of \autoref{thm:len-d-plus-2-vas-le-len-vass-max-dim-d}.

\begin{proof}[Proof of \autoref{thm:len-d-plus-2-vas-le-len-vass-max-dim-d}]
    Let $W_d$, $U_d$ and $A_d$ be the polynomials defined in \autoref{lem:non-wide-to-family-of-wide-seq}, \autoref{lem:poly-bound-for-almost-unbounded}, and \autoref{lem:dimension-reduction-for-almost-bounded-runs}, respectively.
    Note that these polynomials are also nondecreasing with respect to $d$. We can always adjust $A_d$ such that $A_d(x,y)\ge x+1$ and then define 
    \begin{equation}
        P_d(x):=A_{d}(U_{d+2}(W_{d+2}(x))+W_{d+2}(x), W_{d+2}(x)).
    \end{equation}
    Let {$(V,\Vec{s}, \Vec{t})$ be a $(d+2)$-VASS} with $\Vec{s}\xrightarrow{*}\Vec{t}$ and $M:=\Size(V,\Vec{s},\Vec{t})$. By \autoref{lem:non-wide-to-family-of-wide-seq}, {$(V, \Vec{s}, \Vec{t})$ is length-captured by a family $\mathcal{S}$ of wide sequential VASes with $W_d$-amplification. Since $\Len(V,\Vec{s},\Vec{t}) \neq \emptyset$,} there exists a wide sequential $d_S$-VAS {$(S,s,t) \in \mathcal{S}$} with $d_S\le d+2$ such that $\emptyset\ne \Len(S,s, t)\subseteq \Len(V,\Vec{s},\Vec{t})$, and $\Size(S,s, t)\le W:=W_{d+2}(M)$. Assume that $s\xrightarrow{*}t$ in $S$ is witnessed by $\pi$. Let $U:=U_{d_S+2}(W)$. Applying \autoref{lem:almost-bounded-or-unbounded} to $\pi$, there are two cases.
    
    \textbf{The run $\pi$ is almost unbounded by $U$.} Since $S$ is wide and sequential, by \autoref{lem:poly-bound-for-almost-unbounded}, there exists a run $s\xrightarrow{\rho}t$ of length $|\rho|\le U$. We can always construct a $0$-VASS $V'$ by concatenating $U+1$ states with source $s'$ and target $t'$. Then $\Size(V',s',t')\le U+1\le A_{d}(U+\Size(S), W)$ and $\min \Len(V',s',t')=U\ge \min \Len(S,s,t) \ge \min \Len(V,\Vec{s},\Vec{t})$.
        
    \textbf{The run $\pi$ is almost bounded by} $U+\Size(S)$. By \autoref{lem:dimension-reduction-for-almost-bounded-runs}, there is a VASS {$(V',s',t')$} with $\dimcom(V')\le d$ such that {$s'\xrightarrow{*}t'$ in $V'$,} $\min\Len(V,\Vec{s},\Vec{t})\le \min \Len(S,s,t)\le \min\Len(V's',t')$ and $\Size(V',s',t')\le A_{d}(U+\Size(S), W)$.
        
    In both cases, the size does not exceed $P_d(M)$, which justifies that $(V, \Vec{s},\Vec{t})$ is length-bounded by a VASS of component dimension at most $d$ with $P_d$-amplification.
\end{proof}

\paragraph*{Notes on Geometrically \texorpdfstring{$d$}{d}-Dimensional VASes}
So far, we have established \autoref{thm:len-d-plus-2-vas-le-len-vass-max-dim-d}. A natural extension is to ask what the corresponding result would be if we fix the geometric dimension $g$ of the original VAS rather than its (standard) dimension $d$. In fact, by taking a closer look at the above proofs, we can easily derive the following conclusion.

\begin{restatable}{lemma}{lenGPlusOneVasLeLenVassMaxDimG}
    \label{lem:len-g-plus-1-vas-le-len-vass-max-dim-g}
    For every $g \in \mathbb{N}$, there is a polynomial $P_g$ such that, every $d$-VAS $(V,\Vec{s},\Vec{t})$ with $\dimcyc(V) = g + 1$ is length-bounded by a $d$-VASS $(V',s',t')$ with $f$-amplification, where $\dimcom(V') \le g$ and $f(x) = P_g(x)^{(d+1)!}$.
\end{restatable}

\begin{proof}[Proof Sketch]
    The analysis is analogous to that in \autoref{thm:len-d-plus-2-vas-le-len-vass-max-dim-d}. For the size bound of the new VASS, the $(d+1)!$ term in the exponent originates from Rackoff’s extraction (\autoref{lem:reach-unbounded-of-ever-unbounded}). For the component-dimension of the new VASS, it suffices to show that the dimension reduction in \autoref{lem:reduction-for-bounded} decreases the cycle-dimension of the sequential VAS $S$ by at least one. Note that after encoding any $i \in [d]$ into the state, every cycle $\theta$ in the resulting VASS, denoted by $V'$, satisfies $\Delta(\theta)(i)=0$. In contrast, the original sequential VAS $S$ contains no fixed coordinates, i.e., there exists a self-loop $u$ with $\Delta(u)(i)\ne 0$. Consequently, we have $\CycleSpace(V') \subsetneq \CycleSpace(S)$, which implies that $\dimcyc(V')\le \dimcyc(S) -1 \le  g$.
\end{proof}

The existence of a $2$-bounded run allows us to remove two coordinates of the VAS, which, however, might only decrease the geometric dimension by one. As an example, consider a $3$-VAS $(V,\Vec{s},\Vec{t})$ for which $\Vec{u}(1)+\Vec{u}(2)=0$ holds for every transition vector $\Vec{u}$ in $V$. It is straightforward that these two coordinates are always bounded by $\Size(V,\Vec{s},\Vec{t})^{O(1)}$ in every run in $V$. However, removing these two coordinates can only reduce the component-dimension by one, as they are collinear.

Finally, we can combine \autoref{thm:d-vass-in-fd} and \autoref{lem:len-g-plus-1-vas-le-len-vass-max-dim-g} to derive the following corollary.

\begin{corollary}
    Reachability in geometrically $(g+1)$-dimensional VAS is in $\mathsf{F}_g$ for $g\ge 3$.
\end{corollary}

\section{Fine-Grained Characterization of VAS Runs}
\label{sec:low-dim-vas-boundedness}

\autoref{thm:len-d-plus-2-vas-le-len-vass-max-dim-d} gives a nice characterization of shortest runs in $d$-dimensional VAS using VAS of component-dimension $d-2$. This suffices to prove the $\mathsf{F}_{d-2}$ upper-bound for fixed dimension $d \ge 5$. However, some good properties of this characterization are hidden beneath the technical proofs. In this section, we give a finer characterization of runs in $d$-VAS, relating them to $(d-2)$-VASS or VASS of component-dimension $d - 3$, which is formulated as follows.

\begin{proposition}
    \label{lem:len-d-vas-le-cases}
    For each $d \in \mathbb{N}$ where $d \ge 2$, there is a polynomial $P_d$ such that
    every $d$-VAS $(V,\Vec{s},\Vec{t})$ is length-bounded by a VASS $(V',s',t')$ with $P_d$-amplification. Moreover,
    \begin{itemize}
        \item either $V'$ is a $(d-2)$-VASS with $\dimcyc(V') \le \dimcyc(V) - 1$,
        \item or $V'$ is comprised of two $d$-VASSes $V_1, V_2$ connected by a single transition, where 
        \begin{equation}
            \dimcyc(V_1), \dimcyc(V_2) \le \max\{0, \dimcyc(V) - 3\}.
        \end{equation}
        In particular, it holds that $\dimcom(V') \le \max\{0, \dimcyc(V) - 3\}$.
    \end{itemize}
\end{proposition}

This proposition is particularly favorable for low-dimensional VAS compared to \autoref{thm:len-d-plus-2-vas-le-len-vass-max-dim-d}. For example, \autoref{thm:len-d-plus-2-vas-le-len-vass-max-dim-d} bounds the length of shortest runs in 5-VAS by VASS of component-dimension $3$, for which we know nothing better than \autoref{thm:d-vass-in-fd}. On the other hand, \autoref{lem:len-d-vas-le-cases} allows us to rely on sharper results for 3-VASS \cite{DBLP:conf/icalp/CzerwinskiJ0O25} and VASS of component-dimension 2 \cite{DBLP:conf/icalp/FuYZ24}.

Before proving \autoref{lem:len-d-vas-le-cases}, we introduce the notion of \emph{collinear coordinates} in VASes. Let $V = (Q, T)$ be a $d$-VAS and $i, j \in [d]$. We say $i$ and $j$ are \emph{collinear (in $V$)} if there exists $\alpha \in \mathbb{Q}$ such that for every transition $\Vec{u} \in T$, we have $\Vec{u}(i) = \alpha \cdot\Vec{u}(j)$. Furthermore, $i$ and $j$ are \emph{positively collinear} if $\alpha > 0$ and \emph{negatively collinear} if $\alpha < 0$. We remark that by \autoref{lem:exists-sequential-vas-without-fixed-coordinate} we may assume all coordinates are not fixed, so we cannot have $\alpha = 0$ in the case where $i$ and $j$ are collinear. For a sequential VAS $S = V[\Vec{a}_1, \ldots, \Vec{a}_k]$, collinearity is defined with respect to $V$. We observe that negatively collinear coordinates are always bounded, and can be removed by encoding into states. Formally, we have the following lemma, whose proof is provided in the appendix.

\begin{restatable}{lemma}{lenOfNegColVAS}
    \label{lem:len-of-neg-col-vas-le-d-minus-2-vass}
    There is a polynomial $P$ such that every sequential $d$-VAS $(S, s, t)$, which contains a pair of negatively collinear coordinates, is length-bounded by a $(d-2)$-VASS $(V',s',t')$ with $P$-amplification, where $\dimcyc(V') \le \dimcyc(S)-1$.
\end{restatable}

The core idea of proving \autoref{lem:len-d-vas-le-cases} is to improve the pumping condition \autoref{lem:tgt-pumpable-of-unbounded-d-1}. It states that in wide sequential VASes, pumpability always holds in a target configuration where at most one coordinate is bounded. This condition can be weakened to the point that all bounded coordinates are pairwise positively collinear.

\begin{restatable}{lemma}{cfgPumpableOfBoundedPosCol}
    \label{lem:cfg-pumpable-of-bounded-pos-col}
    Let $B$ be the polynomial given by \autoref{lem:tgt-pumpable-of-unbounded-d-1}, $S = V[\Vec{a}_1, \ldots, \Vec{a}_k]$ be a wide sequential $d$-VAS (containing no fixed coordinate) with source configuration $q_0(\Vec{x}_0)$. If there is a run $q_0(\Vec{x}_0) \xrightarrow{*} q(\Vec{x})$ for some state $q$ such that the coordinates in the set $\{i \in [d] \mid \Vec{x}(i) < B(\norm{\Vec{x}_0}, \Size(S))\}$ are pairwise positively collinear, then $q(\Vec{x})$ is pumpable in $S$.
\end{restatable}

\begin{proof}[Proof Sketch]
    Let $K := \{i \in [d] \mid \Vec{x}(i) < B(\norm{\Vec{x}_0}, \Size(S))\}$. The proof is almost the same as that of \autoref{lem:tgt-pumpable-of-unbounded-d-1} with one difference: we need to find a transition $\Vec{t}$ in $V$ such that $\Vec{t}(i) > 0$ for all $i \in K$ in order to pump up coordinates in $K$. Since $V$ is wide, there must be a transition $\Vec{t}$ with $\Vec{t}(i) > 0$ for some $i \in K$. As every pair of coordinates in $K$ is positively collinear, the effect of $\Vec{t}$ on these coordinates must be positive simultaneously.
\end{proof}

In the following, we prove \autoref{lem:len-d-vas-le-cases} depending on whether the VAS is wide or not.

\subsection{Wide VAS}

First, we deal with wide VASes. Recall that \autoref{thm:len-d-plus-2-vas-le-len-vass-max-dim-d} was established for wide sequential VASes. There, we introduced the notions of almost boundedness and almost unboundedness to handle non-trivial bridges between VAS components. Since a wide VAS consists of only a single state, it suffices to distinguish runs based on whether they are $2$-bounded. This simplification yields the result that the shortest run length in a wide $d$-VAS can be bounded by that of a $(d-2)$-VASS, rather than a VASS of component-dimension $d-2$.

\begin{restatable}{lemma}{lenWideVAS}
    \label{lem:len-wide-vas}
    For each $d\in\mathbb{N}$ where $d\ge2$, there is a polynomial $M_d$ such that every wide $d$-VAS $(V,\Vec{s}, \Vec{t})$ with $\dimcyc(V)\ge1$ is length-bounded by a $(d-2)$-VASS $(V', s',t')$ with $M_d$-amplification, where $\dimcyc(V')\le \dimcyc(V)-1$.
\end{restatable}

\begin{proof}[Proof Sketch]
    Let $P$ be the polynomial in \autoref{thm:reach-pumpable-of-ever-unbounded-d-1}, $M := \Size(V, \Vec{s}, \Vec{t})$. Pumpability is guaranteed if there is at most one coordinate bounded by $U := P(M, M)^{(d+1)!}$ along $\pi$. In this case, \autoref{lem:lift-Z-run-by-pumpable} bounds the length of shortest runs by a polynomial depending on $d$. On the other hand, if there are at least two coordinates bounded by $U$, we may encode these coordinates into states to obtain a $(d-2)$-VASS capturing this run. 
\end{proof}

\subsection{Non-Wide VAS}
In the case where a given $d$-VAS $V$ is non-wide, we still apply \autoref{lem:non-wide-to-family-of-wide-seq} to obtain a family of wide sequential VASes. We emphasize again the fact that the cycle-dimension of newly obtained sequential VASes is at most $\dimcyc(V)-1$. Compared with the result we want for \autoref{lem:len-d-vas-le-cases}, we still need to eliminate two coordinates or reduce the component-dimension by at least $2$. Built upon \autoref{lem:len-of-neg-col-vas-le-d-minus-2-vass}, we focus on those sequential VASes containing no negatively collinear pair of coordinates here, which leads to the following lemma.

\begin{restatable}{lemma}{lenWideSequentialVASFiner}
    \label{lem:len-wide-sequential-vas-finer}
    For each $d \in \mathbb{N}$, there is a polynomial $N_d$ such that
    every wide sequential $d$-VAS $(S, s, t)$ containing no negatively collinear pairs is length-bounded by a $d$-VASS $(V',s',t')$ with $N_d$-amplification, where $V'$ is comprised of two $d$-VASSes $V_1, V_2$ connected by a single transition, where $\dimcyc(V_1), \dimcyc(V_2) \le \max\{0, \dimcyc(S) - 2\}$.
\end{restatable}

The proof resembles that of \autoref{thm:len-d-plus-2-vas-le-len-vass-max-dim-d} with one major difference: we need to make sure that the bounded coordinates in a bounded run are not collinear so that eliminating them can reduce the component-dimension by at least $2$. To begin with, we introduce two new notions. Let $B\in \mathbb{N}$. A run $\pi$ is said to be \emph{quasi-unbounded by $B$} if it can be splited into $\pi=\pi_1\pi_2$, where the coordinates in $[d]\setminus\UBCounters(\pi_i,B)$ are pairwise positively collinear for each $i=1,2$, and $\pi$ is \emph{quasi-bounded by $B$} if $\pi=\pi_1\pi_2$, where $[d]\setminus\UBCounters(\pi_i,B)$ contains at least one pair of non-collinear coordinates for each $i=1,2$. We have the following dichotomy lemma.

\begin{restatable}{lemma}{quasiBoundedOrUnbounded}
    \label{lem:quasi-bounded-or-unbounded}
    Let $(V,s,t)$ be a VASS with $s\xrightarrow{\pi}t$ and $B \in \mathbb{N}$. If $V$ contains no negatively collinear pairs, then $\pi$ is either quasi-unbounded by $B$, or quasi-bounded by $B+\Size(V)$.
\end{restatable}

With the help of \autoref{lem:cfg-pumpable-of-bounded-pos-col}, we can extract pumpability from quasi-unbounded runs. Hence, the following polynomial length bound on quasi-unbounded runs can be obtained. Its proof is just a replay of that of \autoref{lem:poly-bound-for-almost-unbounded}.

\begin{restatable}{lemma}{polyBoundForQuasiUnbounded}
    \label{lem:poly-bound-for-quasi-unbounded}
    For every $d \in \mathbb{N}$, there exists a polynomial $U_d$ such that, given a wide sequential $d$-VAS $(S,s,t)$ with $s \xrightarrow{\pi}t$ in $S$, if $\pi$ is quasi-unbounded by $U_d(\Size(S, s, t))$, then there exists a run $s \xrightarrow{\rho} t$ of length $|\rho|\le U_d(\Size(S, s, t))$.
\end{restatable}

As for the quasi-bounded case, one can repeat the dimension reduction technique in \autoref{lem:dimension-reduction-for-almost-bounded-runs}, which captures the two parts $\pi = \pi_1\pi_2$ separately by encoding into states the two bounded coordinates. Quasi-boundedness ensures the existence of two bounded coordinates that are not collinear. Therefore, the cycle-dimension drops by $2$ after the encoding. 

Now we can establish the finer characterization in \autoref{lem:len-d-vas-le-cases}.
\begin{proof}[Proof of \autoref{lem:len-d-vas-le-cases}]
    Consider a $d$-VAS $(V,\Vec{s},\Vec{t})$ with $\Vec{s}\xrightarrow{*}\Vec{t}$ and let $M:=\Size(V,\Vec{s},\Vec{t})$. 
    We may abbreviate the two cases in the lemma as the ``first characterization'' ($(V, \Vec{s}, \Vec{t})$ is length-bounded by a $(d-2)$-VASS $(V', s', t')$ with $\dimcyc(V') \le \dimcyc(V) - 1$) and the ``second characterization'' ($(V, \Vec{s}, \Vec{t})$ is length-bounded by a $d$-VASS $(V', s', t')$ such that $V'$ is comprised of two $d$-VASSes $V_1, V_2$ connected by a single transition, where $\dimcyc(V_1), \dimcyc(V_2) \le \max\{0, \dimcyc(V) - 3\}$), respectively. 
    
    The case where $\dimcyc(V)=0$ falls naturally into the second characterization. Now assume that $\dimcyc(V)\ge 1$. If $V$ is wide itself, then by \autoref{lem:len-wide-vas}, we have the first characterization. Otherwise, applying \autoref{lem:non-wide-to-family-of-wide-seq}, we obtain a wide sequential $d'$-VAS $(S,s,t)$ where $d'\le d$ and $\dimcyc(S)\le \dimcyc(V)-1$, such that $\Size(S,s,t)=W_d(M)$, $s\xrightarrow{*}t$ in $S$, and $\min\Len(V,\Vec{s},\Vec{t})\le \min\Len(S,s,t)$. If $S$ contains a pair of negatively collinear coordinates, the first characterization follows from \autoref{lem:len-of-neg-col-vas-le-d-minus-2-vass}. Now we may assume that $S$ contains no negatively collinear pairs. Then by \autoref{lem:len-wide-sequential-vas-finer}, there exists a $d'$-VASS $(V',s',t')$ with $\Size(V', s', t') \le N_d(W_d(M))$ such that $s'\xrightarrow{*}t'$ and $\min \Len(V, \Vec{s}, \Vec{t}) \le \min \Len(V, \Vec{s}, \Vec{t})$, where $N_d$ is given in \autoref{lem:len-wide-sequential-vas-finer}. Moreover, 
    $V'$ is comprised of two $d$-VASSes $V_1, V_2$ connected by a single transition, where $\dimcyc(V_1), \dimcyc(V_2) \le \max\{0, \dimcyc(S) - 2\} \le \max\{0, \dimcyc(V) - 3\}$
    . Thus, the second characterization holds. The polynomial $P_d$ is determined implicitly in the above argument.
\end{proof}

\section{Low Dimensional VAS}
\label{sec:low-dim-vas}

In this section, we apply the fine-grained characterization of VAS runs (\autoref{lem:len-d-vas-le-cases}) to prove better complexity bounds for 4- and 5-VAS:

\begin{theorem}\label{thm:5-vas-4-vas-upper-bound}
    Reachability in 5-VAS is in $\ELEM$. Reachability in 4-VAS is in $\PSPACE$ under binary encoding and is $\NL$-complete under unary encoding.
\end{theorem}

The $\ELEM$ upper bound of 5-VAS is almost immediate. With the help of \autoref{lem:len-d-vas-le-cases}, the length of the shortest runs in a 5-VAS is bounded by that in a 3-VASS or a sequence of two geometrically 2-dimensional VASSes of polynomial size amplification. For the 3-VASS, we have a recent triply-exponential bound \cite[Lemma 2]{DBLP:conf/icalp/CzerwinskiJ0O25}; While for the sequence of two geometrically 2-dimensional VASSes, we can convert each one into linear-path schemes \cite[Theorem 3.4]{DBLP:conf/icalp/FuYZ24}. The shortest runs in linear-path schemes are characterized by systems of linear inequalities and thus enjoy an exponential length bound (in terms of unary-encoded size). These analyses sum up to the following lemma, whose proof is deferred to the appendix.

\begin{restatable}{lemma}{lenFourVASLeThreeEXP}
    \label{lem:len-4-vas-le-3-exp}
    In a 5-VAS $(V,\Vec{s},\Vec{t})$, if $\Vec{s}\xrightarrow{*}\Vec{t}$, then $\Vec{s}\xrightarrow{\rho}\Vec{t}$ for some $\rho$ with $|\rho|\le 2^{2^{2^{\Size(V,\Vec{s},\Vec{t})^{O(1)}}}}$.
\end{restatable}

\begin{proof}[Proof of \autoref{thm:5-vas-4-vas-upper-bound}, for 5-VAS]
    Enumerating all paths of triply-exponential length (in terms of unary-encoded size) can be done in $3$-\textsf{EXPSPACE} (in terms of binary-encoded size). Therefore, we obtain the \ELEM{} upper bound for 5-VAS.
\end{proof}

As for the reachability problem in $4$-VAS, we can also apply \autoref{lem:len-wide-sequential-vas-finer} to capture the shortest runs by either a $2$-VAS or a sequence of two geometrically 1-dimensional VASSes. 
On the one hand, 2-VAS inherits the $\PSPACE$ upper bound from 2-VASS \cite{DBLP:journals/jacm/BlondinEFGHLMT21}. On the other hand, we are not able to apply the linear-path schemes characterization to the sequence of geometrically 1-dimensional VASSes, as it yields merely a doubly-exponential bound in terms of binary-encoded size. Fortunately, the sequence of two geometrically 1-dimensional VASSes form a geometrically $2$-dimensional one as a whole. In \autoref{sec:low-dim-vas-geo-2d-vass}, we will show the following conversion from a geometrically $2$-dimensional VASS to a real 2-VASS:

\begin{restatable}{lemma}{lemGeoTwoDToTwoDPoly}
    \label{lem:geo-2d-to-2d-poly}
    There is a polynomial $H$ such that, for every geometrically 2-dimensional $d$-VASS $(V, s, t)$, there exists a family $\mathcal{V}$ of 2-VASSes $(V', s', t')$ such that $(2d+1) \cdot \Len(V, s, t) = \bigcup_{(V', s', t') \in \mathcal{V}} \Len(V', s', t')$, and $\Size(V', s', t') \le H(\Size(V, s, t))$.
\end{restatable}

Hence, for both the case of a 2-VAS and the case of a geometrically 2-VASS, we are able to bound the length of their shortest runs using known results for 2-VASS \cite{DBLP:journals/jacm/BlondinEFGHLMT21}. We then establish the length bound for 4-VAS in the following lemma.

\begin{lemma}
    \label{thm:len-4-vas-poly}
    For any $4$-VAS $(V,\Vec{s}, \Vec{t})$ with $\Vec{s} \xrightarrow{*} \Vec{t}$, it holds that $\min \Len(V, \Vec{s}, \Vec{t}) \le \Size(V,\Vec{s},\Vec{t})^{O(1)}$.
\end{lemma}

\begin{proof}
    By \autoref{lem:len-d-vas-le-cases}, the 4-VAS $(V, \Vec{s}, \Vec{t})$ is length-bounded by a VASS $(V',s',t')$ with polynomial amplification, where, $V'$ is either a $2$-VASS or a $4$-VASS of component-dimension at most $1$. For the latter case, it should be noted that $V'$ consists of two geometrically $1$-dimensional sub-VASSes, and hence, it is geometrically $2$-dimensional as a whole. Applying \cref{lem:geo-2d-to-2d-poly}, $V'$ can be converted into a 2-VASS of polynomial size. We conclude the proof with \autoref{cor:min-len-le-poly-of-geo-2d}.
\end{proof}

For any $4$-VAS $(V,\Vec{s}, \Vec{t})$ with $\Vec{s} \xrightarrow{*} \Vec{t}$ in $V$, we deduce that there is a short witness whose length is bounded by $\Size(V,\Vec{s},\Vec{t})^{O(1)}$. A brute-force algorithm searching all such runs can be completed in polynomial space when the input is encoded in binary, and in logarithmic space when the input is encoded in unary. Finally, we are able to finish the proof of \autoref{thm:5-vas-4-vas-upper-bound}.

\begin{proof}[Proof of \autoref{thm:5-vas-4-vas-upper-bound}, for 4-VAS]
    The upper bounds follow from \autoref{thm:len-4-vas-poly}. It is shown in \cite{DBLP:journals/tcs/HopcroftP79} that a $0$-VASS can be simulated by a $3$-VAS with a polynomial overhead. Moreover, the conversion from $0$-VASS to unary $3$-VAS can be done in logarithmic space. Since the reachability in directed graphs is \NL-complete, we conclude that reachability in unary $3$- and $4$-VAS is \NL-complete.
\end{proof}

In the sequel, we primarily focus on proving \autoref{lem:geo-2d-to-2d-poly} for geometrically 2-dimensional VASS.
We also study the reachability problem of geometrically 2-dimensional VASS on its own interest and prove the following characterization:

\begin{restatable}{theorem}{complexityGeoTwoVASS}\label{thm:geo-2-vass}
    Reachability in geometrically 2-dimensional $d$-VASS is:
    \begin{itemize}
        \item \PSPACE{}-complete under binary encoding;
        \item \NP{}-complete under unary encoding, with $d$ not fixed;
        \item \NL{}-complete under unary encoding, with $d$ fixed.
    \end{itemize}
\end{restatable}

We remark that the \PSPACE{}-completeness under binary encoding was already proved in \cite{DBLP:conf/concur/Zheng25}, while results for unary encoding are novel and cannot be derived by techniques in \cite{DBLP:conf/concur/Zheng25}. 
Notice that for geometrically 2-dimensional VASS under unary encoding with dimension not fixed, we do not have an \NL{} upper bound as in the case for 2-VASS. Indeed, in this case, enumerating configurations cannot fit in logarithmic space (there is an $O(d)$ factor). On the other hand, if the VASS is a VAS, we can actually retain the \NL{} upper bound, whose proof is deferred to \autoref{sec:geo-2d-vas}.

\begin{restatable}{theorem}{complexityGeo2VAS}\label{thm:geo-2d-vas}
    Reachability in geometrically 2-dimensional VAS is in $\NL$ under unary encoding.
\end{restatable}

\subsection{Geometrically 2-Dimensional VASS}
\label{sec:low-dim-vas-geo-2d-vass}

Several results are known for geometrically 2-dimensional VASS. In \cite[Theorem 3.4]{DBLP:conf/icalp/FuYZ24} a linear path scheme characterization for geometrically 2-dimensional VASS was proved, which implies the reachability relation can be modeled by systems of integer linear programming. But the number of variables in the linear programming can merely be bounded by an exponential\footnote{in terms of binary-encoded size} function. It is also proved in \cite{DBLP:conf/concur/Zheng25} that reachability in geometrically 2-dimensional VASS is in \PSPACE{} under binary encoding. However, the length bound of the shortest runs still has exponential dependence on the number of states. In this section, we give a much simplified conversion from geometrically 2-dimensional VASS to 2-VASS of polynomial size, from which a polynomial length bound on the shortest runs can be derived.

\lemGeoTwoDToTwoDPoly*

Notice that for a geometrically 2-dimensional VASS $V$ we may always assume that $\dimcyc(V) = 2$ as we can add dummy cycles to enlarge the cycle space if $\dimcyc(V) < 2$. In the following we adopt this assumption tacitly.

\subsubsection{Sign Reflecting Projection}
We will show that a geometrically 2-VASS can be ``projected'' onto two coordinates to get an equivalent 2-VASS. Given a subset $I \subseteq [d]$ of indices, for a vector $\Vec{x} \in \mathbb{Q}^d$ we write $\Vec{x}|_I$ for the vector (or function) in $\mathbb{Q}^I$ such that $\Vec{x}|_I(i) = \Vec{x}(i)$ for all $i \in I$. We shall recognize functions in $\mathbb{Q}^I$ as vectors in $\mathbb{Q}^{|I|}$.

\begin{definition}
    Let $X \subseteq \mathbb{Q}^d$ be a vector space. We say $I \subseteq [d]$ is a \emph{sign-reflecting projection} of $X$ if $|I| = \dim(X)$ and for all $\Vec{x} \in X$, $\Vec{x}|_I \ge \Vec{0}$ implies $\Vec{x} \ge \Vec{0}$.
\end{definition}

In \cite[Theorem 3.4]{DBLP:conf/icalp/FuYZ24}, it is proved that sign-reflecting projections exist for 2-dimensional spaces satisfying certain conditions. Here, we introduce a much weaker condition.

\begin{definition}
    \label{def:2d-space-projective-wrt}
    Let $X = \Span\{\Vec{u}, \Vec{v}\} \subseteq \mathbb{Q}^d$ be a 2-dimensional vector space. Denote $\Vec{r}_\ell := (\Vec{u}(\ell), \Vec{v}(\ell))$ for $\ell \in [d]$. We say $X$ is \emph{projective} (w.r.t.\ the base $\Vec{u}, \Vec{v}$) if $\Cone(\{\Vec{r}_1, \ldots, \Vec{r}_d\}) = \Cone(\{\Vec{r}_i, \Vec{r}_j\})$ for some $i, j \in [d]$.
\end{definition}

Firstly we verify that \autoref{def:2d-space-projective-wrt} does not depend on the choice of base vectors. The following lemma gives an equivalent definition to projectiveness.

\begin{restatable}{lemma}{twoDSpaceProjectiveIff}
    \label{lem:2d-space-projective-iff}
    Let $X = \Span\{\Vec{u}, \Vec{v}\} \subseteq \mathbb{Q}^d$ be a 2-dimensional vector space. Denote $\Vec{r}_\ell := (\Vec{u}(\ell), \Vec{v}(\ell))$ for $\ell \in [d]$. Then $X$ is projective w.r.t.\ $\Vec{u}, \Vec{v}$ if and only if there exists a vector $\Vec{n} \in \mathbb{Q}^2$ such that $\Inner{\Vec{n}, \Vec{r}_\ell} > 0$ for all $\ell \in [d]$ where $\Vec{r}_\ell \ne \Vec{0}$.
\end{restatable}

Based on this equivalent definition, it is easy to show that \autoref{def:2d-space-projective-wrt} is independent of base vectors.

\begin{restatable}{lemma}{twoDSpaceProjectiveIndependentBase}
    \label{lem:lem:2d-space-projective-independent-base}
    Let $X = \Span\{\Vec{u}, \Vec{v}\} = \Span\{\Vec{x}, \Vec{y}\} \subseteq \mathbb{Q}^d$ be a 2-dimensional vector space. Then $X$ is projective w.r.t.\ $\Vec{u}, \Vec{v}$ if and only if $X$ is projective w.r.t.\ $\Vec{x}, \Vec{y}$.
\end{restatable}

Now we can say $X$ is projective or not without referring to some base of $X$. The following lemma justifies the existence of a sign-reflecting projection for projective spaces.

\begin{restatable}{lemma}{srpOfProjective}
    \label{lem:srp-of-projective}
    Let $X \subseteq \mathbb{Q}^d$ be a projective 2-dimensional vector space. Then there is a sign-reflecting projection of $X$.
\end{restatable}

In the following, we prove \autoref{lem:geo-2d-to-2d-poly} based on whether the cycle space is projective or not.

\subsubsection{Projective Case}

\begin{restatable}{lemma}{geoTwoDtoTwoDProjective}
    \label{lem:geo-2d-to-2d-projective}
    There is a polynomial $H_1$ such that, for every geometrically 2-dimensional $d$-VASS $(V, s, t)$, where $\CycleSpace(V)$ is projective, there exists a family $\mathcal{V}$ of 2-VASSes $(V', s', t')$ such that $(2d+1) \cdot \Len(V, s, t) = \bigcup_{(V', s', t') \in \mathcal{V}} \Len(V', s', t')$, and $\Size(V', s', t') \le H_1(\Size(V, s, t))$.
\end{restatable}

Intuitively, we would like to project $V$ onto the two coordinates which form the sign-reflecting projection of $\CycleSpace(V)$. To make sure the resulting VASS faithfully reflects the behavior of $V$, we need to design a gadget to recover the omitted coordinates and test if they have non-negative values. From now on, let's fix a geometrically 2-dimensional $d$-VASS $(V, s, t)$, such that $\CycleSpace(V)$ is projective. Let $I = \{i, j\} \subseteq [d]$ be a sign-reflecting projection of $\CycleSpace(V)$. Let $\Vec{u}, \Vec{v} \in \mathbb{Z}^d$ be two vectors that span $\CycleSpace(V)$. We may assume that $\Vec{u}, \Vec{v}$ are effects of two simple cycles in $V$, hence $\norm{\Vec{u}}, \norm{\Vec{v}} \le \Size(V)$.

\begin{restatable}{proposition}{canonicalAxesOfProjective}
    \label{prop:canonical-axes-of-projective}
    There exists $\overline{\Vec{u}}, \overline{\Vec{v}} \in \CycleSpace(V) \cap \mathbb{N}^d$ such that $\overline{\Vec{u}}|_I = (c, 0)$ and $\overline{\Vec{v}}|_I = (0, c)$ for some $c > 0$, and $\norm{\overline{\Vec{u}}}, \norm{\overline{\Vec{v}}} \le 2\Size(V)^2$.
\end{restatable}

The two vectors $\overline{\Vec{u}}, \overline{\Vec{v}}$ given by \autoref{prop:canonical-axes-of-projective} will be called the \emph{canonical axes} of the cycle space of $V$, with respect to the base $\Vec{u}$ and $\Vec{v}$. Since configurations in a VASS run might not always stay in the cycle space, we introduce the notion of \emph{shifts} of vectors to measure their deviation from the cycle space.

\begin{definition}
    Given $\Vec{x} \in \mathbb{Q}^d$, the \emph{shift} of $\Vec{x}$ from $\CycleSpace(V)$ (w.r.t.\ the canonical axes $\overline{\Vec{u}}, \overline{\Vec{v}}$) is defined as 
    \begin{align}
        \delta(\Vec{x}) := c \cdot \left( \Vec{x} - \frac{\Vec{x}(i) \cdot \overline{\Vec{u}} +\Vec{x}(j) \cdot \overline{\Vec{v}} }{c} \right),
        \label{eqn:def-of-shift}
    \end{align}
    where we recall that $i, j$ are the indices from the sign-reflecting projection $I$ of $\CycleSpace(V)$, and $c = \overline{\Vec{u}}(i) = \overline{\Vec{v}}(j)$.
\end{definition}

Note that $\delta(\cdot)$ is a linear mapping from $\mathbb{Q}^d$ to $\mathbb{Q}^d$. As $\overline{\Vec{u}}, \overline{\Vec{v}}$ spans $\CycleSpace(V)$, using the fact that $\overline{\Vec{u}}|_I = (c, 0)$ and $\overline{\Vec{v}}|_I = (0, c)$, it is routine to verify that $\delta(\Vec{x}) = \Vec{0}$ if and only if $\Vec{x} \in \CycleSpace(V)$. The definition also ensures that $\delta(\Vec{x})$ is integral whenever $\Vec{x}$ is integral. 

We define the shift of a configuration $c = p(\Vec{x})$ as $\delta(c) := \delta(\Vec{x})$. The following proposition implies that starting from a fixed source, configurations on any run have bounded shifts. 

\begin{restatable}{proposition}{shiftSubShiftSrcNormLE}
    \label{prop:shift-sub-shift-src-norm-le}
    For any configurations $e, e'$ in $V$ with $e \xrightarrow{*} e'$, we have $\norm{\delta(e') - \delta(e)} \le 6 \Size(V)^3$.
\end{restatable}

The idea of converting $(V, s, t)$ to a 2-VASS is to drop all coordinates other than $i, j$ and keep record of the shifts using the states. However, using \autoref{prop:shift-sub-shift-src-norm-le} alone, we have $O(6 \Size(V)^3)^d$ possible shifts that could occur on a run from the source $s$, which is exponential in $d$. The next proposition says that in a fixed run, the shift of a configuration is determined by its state.

\begin{restatable}{proposition}{shiftFromStateOfFixedRun}
    \label{prop:shift-from-state-of-fixed-run}
    Let $e \xrightarrow{\pi} e'$ be a run in $V$, let $Q_\pi$ be the set of states occurring on $\pi$. There exists a function $f_{\pi} : Q_\pi \to \mathbb{Z}^d$ such that for any configuration $q(\Vec{x})$ on $\pi$, we have $\delta(\Vec{x}) = f_\pi(q)$.
\end{restatable}

\paragraph*{Nonnegativity Testing}

Using \autoref{eqn:def-of-shift}, a $\mathbb{Z}$-configuration $q(\Vec{x})$ can be recovered from $q, \Vec{x}|_I$ and $\delta(\Vec{x})$. We will design a 2-VASS gedget to test whether $\Vec{x} \ge \Vec{0}$ holds at the projected configuration $q(\Vec{x}|_I)$ with $\delta(\Vec{x})$ in hand. By \autoref{eqn:def-of-shift}, for each $\ell \in [d]$ we have
\begin{align}
    c \cdot \Vec{x}(\ell) = \overline{\Vec{u}}(\ell) \cdot \Vec{x}(i) + \overline{\Vec{v}}(\ell) \cdot \Vec{x}(j) + \delta(\Vec{x})(\ell).
    \label{eqn:cx-eq-xij-plus-delta}
\end{align}
Recall that $\overline{\Vec{u}}$, $\overline{\Vec{v}}$ are non-negative. The following lemma shows that a 2-VASS can test if the right-hand side of \autoref{eqn:cx-eq-xij-plus-delta} is non-negative.

\begin{restatable}{lemma}{existsTwoVASSTestingInequ}
    \label{lem:exists-2-vass-testing-inequ}
    Let $m, n \in \mathbb{N}$ and $b \in \mathbb{Z}$. There is a 2-VASS $V_{m,n;b}$ with two states $p, q$ such that $\Size(V_{m,n;b}) \le O(|b|^2)$ and for all $x, y \in \mathbb{N}$, we have $p(x, y) \xrightarrow{\pi} q(x', y')$ for some run $\pi$ and $x', y' \in \mathbb{N}$ if and only if $mx + ny + b \ge 0$ and $(x', y') = (x, y)$ and $|\pi| = 2$.
\end{restatable}

\begin{proof}[Proof Sketch]
    As $m, n \ge 0$, the set $\{(x, y) \in \mathbb{N}^2 \mid mx + ny + b\}$ can be described as the region above finitely many minimal points. A VASS may test non-deterministically whether the current configuration is above one of those minimal points by subtracting it and then adding it back.
\end{proof}

Now we have enough tools to convert a geometrically 2-dimensional VASS whose cycle space is projective to a family of 2-VASSes. We sketch the proof of \autoref{lem:geo-2d-to-2d-projective} below.

\begin{proof}[Proof sketch of \autoref{lem:geo-2d-to-2d-projective}]
    Let $(V, s, t)$ be a geometrically 2-dimensional VASS where $\CycleSpace(V)$ is projective. Let $I = \{i, j\}$ be a sign-reflecting projection of $\CycleSpace(V)$. By \autoref{prop:shift-from-state-of-fixed-run}, a function $f: Q \to [-6\Size(V)^3, 6\Size(V)]$ will assign the shift $f(q) + \delta(s)$ on a run from $s$ to $t$ when encountering the state $q$. Fixing such a function $f$, we drop all coordinates except $i$ and $j$ from the set of transitions. Moreover, from states $p$ to $q$ we shall keep only those transitions $u = (p, \Vec{a}, q)$ such that $f(q) = f(p) + \delta(a)$, i.e.\ the shift is updated correctly. After each transition, we add $d$ gadgets given by \autoref{lem:exists-2-vass-testing-inequ} to test whether the $d$ coordinates of the recovered configuration are non-negative. We repeat the above construction for each function $f: Q \to [-6\Size(V)^3, 6\Size(V)]$, and collect all of the resulting 2-VASSes to form the desired family $\mathcal{V}$.
\end{proof}

\subsubsection{Non-Projective Case}\label{subsec:non-proj}

Now we move the focus to geometrically 2-dimensional VASS whose cycle space is not projective. We will show that such a VASS can be converted into a geometrically 1-dimensional VASS, which can be reduced to a projective geometrically 2-dimensional VASS by simply adding an isolated self-loop whose effect is properly chosen. 

\begin{restatable}{lemma}{projectiveGeoTwoDFromNonProjGeoTwod}
    \label{lem:projective-geo-2d-from-non-projective-geo-2d}
    There is a polynomial $H_2$ such that, for any geometrically 2-dimensional VASS $(V, s, t)$ where $\CycleSpace(V)$ is not projective, there exists a geometrically 2-dimensional VASS $(V', s', t')$ such that $\CycleSpace(V')$ is projective, and that $\Len(V, s, t) = \Len(V', s', t')$, $\Size(V', s', t') \le H_2(\Size(V, s, t))$.
\end{restatable}

From now on, we fix a geometrically 2-dimensional VASS $(V,s,t)$ as in the assumption of \autoref{lem:projective-geo-2d-from-non-projective-geo-2d}, so $\CycleSpace(V)$ is not projective. For a vector $\Vec{x} \in \mathbb{Q}^d$ we write $\Supp(\Vec{x}) = \{i \in [d] \mid \Vec{x}(i) \ne 0\}$ for the set of non-zero coordinates of $\Vec{x}$. The notation is extended to sets $X \subseteq \mathbb{Q}^d$ as $\Supp(X) = \bigcup_{\Vec{x} \in X}\Supp(\Vec{x})$. For a non-projective 2-dimensional space, we can find a non-negative normal vector with small support.

\begin{restatable}{lemma}{normalVecOfNonProjective}
    \label{lem:normal-vec-of-non-projective}
    Let $X = \Span(\Vec{u}, \Vec{v}) \subseteq \mathbb{Q}^d$ be a 2-dimensional vector space that is not projective, where $\Vec{u}, \Vec{v} \in \mathbb{Z}^d$. Let $M := \max\{\norm{\Vec{u}}, \norm{\Vec{v}}\}$. Then there exists a non-zero vector $\Vec{n} \in \mathbb{N}^d$ such that $\Inner{\Vec{x}, \Vec{n}} = 0$ for all $\Vec{x} \in X$, and $\norm{\Vec{n}} \le 2M^2$, $|\Supp(\Vec{n})| \le 3$ and $\Supp(\Vec{n}) \subseteq \Supp(X)$.
\end{restatable}

We remark that this lemma is essentially \cite[Lemma 5.4]{DBLP:journals/jacm/BlondinEFGHLMT21}.
Suppose $\CycleSpace(V) = \Span(\Vec{u}, \Vec{v})$ where $\Vec{u}, \Vec{v}$ are effects of two simple cycles. Applying \autoref{lem:normal-vec-of-non-projective} we get a normal vector $\Vec{n} \ge \Vec{0}$ of $\CycleSpace(V)$ with $\norm{\Vec{n}} \le 2\Size(V)^2$. Let $K := \Supp(\Vec{n})$, it can be shown that on a run in $V$ starting from a fixed source, configurations have bounded values in coordinates in $K$.

\begin{restatable}{proposition}{boundedCountersOfNonProjective}
    \label{prop:bounded-counters-of-non-projective}
    Suppose $p(\Vec{x}) \xrightarrow{*} q(\Vec{y})$ in $V$. Then $\norm{ \Vec{y}|_K } \le 6 \Size(V, p(\Vec{x}), q(\Vec{y}))^3$.
\end{restatable}

The idea of \autoref{lem:projective-geo-2d-from-non-projective-geo-2d} is to encode into states all possible values in the coordinates in $K$ that could occur on a run from the source $s$. As $K = \Supp(\Vec{n}) \subseteq \Supp(\CycleSpace(V))$, after this encoding, the dimension of the cycle space must drop. We may add some isolated cycles to make its cycle space projective without affecting its reachability relation. Complete proof of \autoref{lem:projective-geo-2d-from-non-projective-geo-2d} can be found in the appendix.

Finally, we are able to prove \autoref{lem:geo-2d-to-2d-poly}.

\begin{proof}[Proof of \autoref{lem:geo-2d-to-2d-poly}]
    Let $(V, s, t)$ be a geometrically 2-dimensional VASS. We first apply \autoref{lem:projective-geo-2d-from-non-projective-geo-2d} if $\CycleSpace(V)$ is not projective. Applying \autoref{lem:geo-2d-to-2d-projective} then finishes the proof.
\end{proof}

\subsubsection{Complexity Bounds for Geometrically 2-Dimensional VASS}

In this section, we discuss the reachability problem in geometrically 2-dimensional VASS on its own interest. Since \autoref{lem:geo-2d-to-2d-poly} converts a geometrically 2-dimensional VASS to a 2-VASS of polynomial size, we are able to apply the well-known length bounds for shortest runs in 2-VASS.

\begin{theorem}[{\cite[Theorem 3.2]{DBLP:journals/jacm/BlondinEFGHLMT21}}]
    \label{thm:min-len-le-poly-of-2d}
    Let $(V, s, t)$ be a 2-VASS. If $s \xrightarrow{*} t$ then $\min \Len(V, s, t) \le \Size(V, s, t)^{O(1)}$.
\end{theorem}

\begin{restatable}{corollary}{minLenLEPolyOfGeoTwoD}
    \label{cor:min-len-le-poly-of-geo-2d}
    Let $(V, s, t)$ be a geometrically 2-dimensional VASS. If $s \xrightarrow{*} t$ then $\min \Len(V, s, t) \le \Size(V, s, t)^{O(1)}$.
\end{restatable}

From this polynomial length bound, we deduce the complexity result of reachability problems in geometrically 2-dimensional VASS as stated in \autoref{thm:geo-2-vass}.

\begin{proof}[Proof of \autoref{thm:geo-2-vass}]
    Consider a reachability instance $(V, s, t)$. Under binary encoding, the length bound in \autoref{cor:min-len-le-poly-of-geo-2d} becomes exponential in the input size. But polynomial space is enough to enumerate every run of exponential length. This confirms the \PSPACE{} membership, while hardness is inherited from that of 2-VASS \cite{DBLP:journals/jacm/BlondinEFGHLMT21}.

    Under unary encoding, we may decide reachability by searching for a run of polynomial length, which can be done in \NP{}. When the actual dimension $d$ is a fixed constant in the problem, the searching can be improved to \NL{} since in the searching we may only keep record of the number of steps and the current configuration whose size is bounded by $\log|Q| + d\cdot \log(\Size(V, s, t)^{O(1)}) \le d \cdot O(\log\Size(V, s, t))$. Notice, however, that when $d$ is not fixed, storing a single configuration requires polynomial space.

    The \NL{}-hardness follows from the trivial reduction from the graph reachability problem. For \NP{}-hardness, we reduce from the \textsc{Subset-Sum} problem, viewing the $d$-coordinates as a $d$-bit binary representation. Details are deferred to the appendix. We remark that a similar reduction can be found in \cite[Lemma 2]{DBLP:conf/apn/Leroux21}.
\end{proof}

\subsection{Unary Geometrically 2-Dimensional VAS}
\label{sec:geo-2d-vas}

From the perspective of algorithms, a major reason why \autoref{lem:geo-2d-to-2d-poly} does not yield an \NL{} upper bound for geometrically 2-dimensional VASS under unary encoding is that it only establishes an equivalence between a geometrically 2-dimensional VASS and a family of polynomial-sized 2-VASSes. The family might contain exponentially many VASSes, however. There exists no trivial nondeterministic algorithm capable of checking the reachability of exponentially many unary 2-VASSes in logarithmic space. However, geometrically 2-dimensional VASes, being a special class of VASSes with only a single state, are expected to admit more space-efficient algorithms when using the unary encoding scheme. Indeed, we give the following reduction for unary geometrically 2-dimensional VASes.

\begin{lemma}\label{lem:geo-2-vas-to-4-vass}
    Reachability in unary geometrically $2$-dimensional VAS can be reduced in logarithmic space to reachability in unary geometrically $2$-dimensional $4$-VASS.
\end{lemma}

This is the key to \autoref{thm:geo-2d-vas}.

\begin{proof}[Proof of \autoref{thm:geo-2d-vas}]
    \autoref{thm:geo-2-vass} gives the \NL{} upper bound for the reachability problem in geometrically $2$-dimensional $4$-VASS. Combined with the logarithmic space reduction given in \autoref{lem:geo-2-vas-to-4-vass}, the \NL{} membership of reachability in unary geometrically $2$-dimensional VAS follows immediately.
\end{proof}

Then we proceed to prove \autoref{lem:geo-2-vas-to-4-vass}. Intuitively, in the case where the cycle space is projective, having a single state ensures that the shift of any configuration remains constant on a run. The resulting family of $2$-VASSes obtained from \autoref{lem:geo-2d-to-2d-projective} should contain only one VASS.

The non-projective case of geometrically $2$-dimensional VAS is more challenging. In \autoref{sec:low-dim-vas-geo-2d-vass}, our strategy for dealing with geometrically $2$-dimensional VASS whose cycle space is non-projective was to transform them into a projective VASS. However, this approach increases the number of states, and we meet the problem of getting an exponentially sized family again. This calls for a new projection scheme that bypasses the projective/non-projective classification altogether. Now we fix a geometrically $2$-dimensional $d$-VAS $(V,\Vec{s},\Vec{t})$. Let $\Vec{u},\Vec{v}\in \mathbb{Z}^d$ be two transitions in $V$ such that $\CycleSpace(V)=\Span\{\Vec{u},\Vec{v}\}$. Then $\norm{\Vec{u}}, \norm{\Vec{v}} \le \Size(V)$. Note that we do not assume $\dimcyc(V) = 2$ exactly, and $\Vec{u}$, $\Vec{v}$ need not be linearly independent. Denote $\Vec{r}_\ell := (\Vec{u}(\ell), \Vec{v}(\ell))$ for $\ell \in [d]$. The following proposition shows that any vector in the set $\{\Vec{r}_\ell \mid \ell \in [d]\}$ can be generated by a nonnegative linear combination of at most four fixed vectors among them.

\begin{restatable}{proposition}{coneGeneratedByFourVec}
    \label{prop:cone-generated-by-4-vec}
    There exists a subset $I\subseteq[d]$ such that $|I|\le 4$ and $\Cone(\{\Vec{r}_1,\dots,\Vec{r}_d\})=\Cone(\{\Vec{r}_i\mid r\in I\})$. 
\end{restatable}

We need an effective version of this proposition in order to construct the logspace reduction.

\begin{restatable}{lemma}{existsSubsetCardFour}
    \label{cor:exists-subset-card-4}
    Given as input $(V,\Vec{s},\Vec{t})$, one can compute in logarithmic space an index set $I$ with $|I|\le 4$ and a sequence of tuples $(c_\ell, \alpha_\ell, \beta_\ell, i_\ell, j_\ell)_{\ell\in [d]}$ where $c_\ell \in \mathbb{N}_{>0}$, $\alpha_\ell, \beta_\ell \in \mathbb{N}$, $i_\ell, j_\ell \in I$ such that $c_\ell\cdot \Vec{x}(\ell)=\alpha_\ell \cdot \Vec{x}(i_\ell)+\beta_\ell\cdot\Vec{x}(j_\ell)$ holds for every $\Vec{x}\in \CycleSpace(V)$.
\end{restatable}

We will project transitions of $V$ onto $\mathbb{Z}^{I}$ to obtain a $4$-VASS, where the index set $I$ is determined by \autoref{cor:exists-subset-card-4}. We reuse the nonnegativity testing gadget in \autoref{lem:exists-2-vass-testing-inequ} to ensure that the new VASS captures the reachability of $(V,\Vec{s},\Vec{t})$.

\begin{proof}[Proof of \autoref{lem:geo-2-vas-to-4-vass}]
    Let $(V,\Vec{s},\Vec{t})$ be a geometrically $2$-dimensional VAS with transition set $T\subseteq \mathbb{Z}^d$. By \autoref{cor:exists-subset-card-4}, one can compute a subset $I\subseteq[d]$ with $|I|\le 4$ and a linear map $f:\mathbb{Q}^{I} \to \mathbb{Q}^{d}$, where $f(\Vec{x})(\ell):=c_\ell^{-1}\cdot(\alpha_\ell \cdot \Vec{x}(i_\ell)+\beta_\ell\cdot\Vec{x}(j_\ell))$ for each $\Vec{x}\in \mathbb{Q}^I$ and $\ell \in [d]$ ensures that $f(\Vec{x})$ is the unique vector in $\CycleSpace(V)$ such that $f(\Vec{x})|_I = \Vec{x}$. Therefore, the following claim holds.
    \begin{claim}
        If $\Vec{s}\xrightarrow{*}\Vec{w}$ in $V$, then $\Vec{w}=\Vec{s}-f(\Vec{s}|_I)+f(\Vec{w}|_I)$.
    \end{claim}
    \begin{claimproof}
        Just observe that $\Vec{w} - \Vec{s} \in \CycleSpace(V)$.
    \end{claimproof}
    Let $\Vec{\delta}:=\Vec{s}-f(\Vec{s}|_I)$. We construct a $|I|$-VASS $(V', s', t')$ in a manner similar to that in the proof of \autoref{lem:geo-2d-to-2d-projective}:
    \begin{itemize}
        \item We add two states $\overline{r},\underline{r}$ to $V'$. Intuitively, both $\overline{r}(\Vec{x})$ and $\underline{r}(\Vec{x})$ encodes the $\mathbb{Z}$-configuration $\Vec{\delta}+f(\Vec{x})$ in $V$, while we shall guarantee that $\underline{r}(\Vec{x})$ is reachable only if $\Vec{\delta}+f(\Vec{x})\ge \Vec{0}$.
        \item For each transition $\Vec{a}\in T$, we add $(\underline{r}, \Vec{a}|_I, \overline{r})$ into $V'$.
        \item Then we add non-negativity testing gadgets to connect $\overline{r}$ back to $\underline{r}$. For each $\ell \in [d]$, we can compute the number $\delta_\ell:=c_\ell\cdot\Vec{\delta}(\ell)\in \mathbb{Z}$ and construct the $2$-VASS gadget $V_{\alpha_\ell,\beta_\ell;\delta_\ell}$ given by \autoref{lem:exists-2-vass-testing-inequ}. The transitions in $V_{\alpha_\ell,\beta_\ell;\delta_\ell}$ are given explicitly in the proof of \autoref{lem:exists-2-vass-testing-inequ}, and one can verify that this gadget can be constructed in logarithmic space. We then extend this VASS to be an $|I|$-VASS such that the original two coordinates align with $i_\ell, j_\ell \in I$. Finally, we concatenate the $d$ gadgets and put them between $\overline{r}$ and $\underline{r}$.
    \end{itemize}
    Let $s':=\underline{r}(\Vec{s}|_I)$ and $t':=\underline{r}(\Vec{t}|_I)$. Suppose $s'\xrightarrow{*}\underline{r}(\Vec{x})$ in $V'$. Then, by our construction, it is routine to check that 
    \begin{equation}
        (\Vec{\delta}+f(\Vec{x}))(\ell)=c_\ell^{-1}\cdot(\delta_i+\alpha_\ell \cdot \Vec{x}(i_\ell)+\beta_\ell\cdot\Vec{x}(j_\ell)) \in \mathbb{N}
    \end{equation}
    holds for each $\ell \in [d]$, and hence $\Vec{s}\xrightarrow{*}\Vec{\delta}+f(\Vec{x})$ in $V$. Since $\Vec{t}=\Vec{\delta}+f(\Vec{t}|_I)$, it follows that $\Vec{s}\xrightarrow{*}\Vec{t}$ in $V$ if and only if $s'\xrightarrow{*}t'$ in $V'$. 
    
    Considering the geometric dimension, note that each simple cycle $\theta$ in $V'$ contains one transition in the form of $(\underline{r},\Vec{a}|_I,\overline{r})$ for some $\Vec{a}\in T$, and several transitions from the nonnegativity gadges whose effects sum to $\Vec{0}$. We obtain that $\Delta(\theta)=\Vec{a}|_I$. Therefore, $\dimcyc(V')\le \dimcyc(V)\le 2$.
\end{proof}

\section{Future Work}
\label{sec:conclusion}

Several intriguing questions remain at the boundary of our current understanding. A fundamental open problem is whether the length of runs in $d$-VAS can be effectively captured by those in $(d-3)$-VASS. Regarding VAS in fixed dimensions, numerous challenges remain unresolved. For binary 2-VAS and 3-VASS, a complexity gap exists as the problems currently reside between an \NP{}-hard lower bound and a \PSPACE{} upper bound. Furthermore, the lower bound for the unary reachability problem in 2-VAS, and more generally, in geometrically 2-dimensional systems where $d$ is not fixed, remains elusive; specifically, it is yet to be determined whether these problems are \NL{}-hard. The case of 5-VAS presents another significant challenge. While the reachability sets of 5-VAS are known to be effectively semilinear, the current 3-\textsf{EXPSPACE} upper bound, derived from the length-boundedness of 3-VASS, appears to be loose. We suggest that a more refined analysis of the geometric constraints in 5-dimensional systems might yield tighter complexity bounds, making the reachability problem in 5-VAS a compelling subject for further exploration.

\bibliography{ref}

\appendix

\section{Proof of Lemma \ref{lem:exists-sequential-vas-without-fixed-coordinate}}

\existsSequentialVasWithoutFixedCoordinate*

\begin{proof}
    Assume that $S=V[\Vec{a}_1,\dots,\Vec{a}_k]$ and $s=p(\Vec{x})$. If $t$ is not reachable from $s$, we are done by taking $(V',s',t')$ to be a trivial unreachable VAS. Otherwise, it is straightforward that $\Vec{x}(j)+\sum_{\iota\in [i]}\Vec{a}_\iota(j)\ge 0$ holds for every $i \in [k]$ and $j$ fixed in $S$. In this case, we can drop all the fixed coordinates to construct a new sequential $d'$-VAS $(V',s',t')$ without increasing the size and dimension. Moreover, the winesses of $(V,s,t)$ one-to-one correspond to those of $(V',s',t')$ with the same length, which implies that $\Len(V',s',t')=\Len(V,s,t)$.
\end{proof}

\section{Proof of Lemma~\ref{lem:reach-unbounded-of-ever-unbounded}}

\lemReachUnboundedofEverUnbounded*

\begin{proof}
    Let $U, M \in \mathbb{N}$, we define two sequences of numbers $L_i, H_i$ depending on $U$ and $M$ for $i = 0, 1, \ldots$ as follows:
    \begin{align}
        H_0 &= 0, \quad &L_0 &= M,\\
        H_{i + 1} &= U + L_i \cdot M, \quad &L_{i + 1} &= M \cdot (H_{i+1})^{i+1} + L_i.
    \end{align}
    We shall prove the following claim by induction on $d$:
    \begin{claim}
        Let $V$ be $d$-VASS with $\Size(V) \le M$. For any run $\pi$ in $V$, there is a run $\rho$ starting from the same source as $\pi$ with $|\rho| \le L_d$ such that its target $q(\Vec{x})$ satisfies $\Vec{x}(i) \ge U$ for all $i \in \UBCounters(\pi, H_d)$. Moreover, $\rho$ can be obtained from $\pi$ by removing some cycles.
    \end{claim}
    \begin{claimproof}
        The case $d = 0$ is easy: we take $\rho$ to be obtained from $\pi$ by removing all cycles exhaustively. Now, suppose the claim holds for all $(d-1)$-VASS, we shall show the claim for $d$-VASS $V$. Consider a run $\pi$ in $V$ with source $q_0(\Vec{x}_0)$. Notice that if $\UBCounters(\pi, H_d) = \emptyset$ then every coordinate is less than $H_d$ along $\pi$. In this case, we simply let $\rho$ be obtained from $\pi$ by removing the cycles exhaustively connecting repeated configurations, whose length is bounded by $\Size(V) \cdot (H_d)^d \le L_d$. Next we assume $\UBCounters(\pi, H_d) \ne \emptyset$.
        On the run $\pi$ we pick the first configuration $q_1(\Vec{x}_1)$ with $\Vec{x}_1(j) \ge H_d$ for some $j \in \UBCounters(\pi, H_d)$. We factor $\pi$ as 
        \begin{equation}
            q_0(\Vec{x}_0) \xrightarrow{\pi_1} q_1(\Vec{x}_1) \xrightarrow{\pi_2} q_2(\Vec{x}_2).
        \end{equation}
        Notice that for any configuration $p(\Vec{y})$ visited on $\pi_1$ before $q_1(\Vec{x}_1)$, we have $\norm{\Vec{y}} < H_d$. Let $\rho_1$ be obtained from $\pi_1$ by removing cycles connecting repeated configurations, we have $|\rho_1| \le \Size(V) \cdot (H_d)^d \le M \cdot (H_d)^d$. Let $V^{-j}$ be the $(d-1)$-VASS obtained from $V$ by dropping the $j$-th coordinate from every transition. Then $\pi_2$ can be mapped to a run $\pi_2^{-j}$ in $V^{-j}$ from $q_1(\Vec{x}_1^{-j})$ to $q_2(\Vec{x}_2^{-j})$, where we write $\Vec{x}^{-j}$ for the vector obtained from $\Vec{x}$ by dropping the $j$-th coordinate. Observe that $\UBCounters(\pi_2, H_d) \setminus\{j\} \subseteq \UBCounters(\pi_2^{-j}, H_{d-1})$. By induction hypothesis there is a run $\rho_2^{-j}$ with $|\rho_2^{-j}| \le L_{d-1}$ such that $q_1(\Vec{x}_1^{-j}) \xrightarrow{\rho_2^{-j}} q(\Vec{x}')$ where $\Vec{x}'(i) \ge U$ for all $i \in \UBCounters(\pi_2^{-j}, H_{d-1})$. Moreover, we may assume $\rho_2^{-j}$ is obtained from $\pi_2^{-j}$ by removing cycles. Then $\rho_2^{-j}$ is a subsequence of $\pi_2^{-j}$. With the $j$-th coordinate added back, we take $\rho_2$ to be the same subsequence of $\pi_2$. Then $q_0(\Vec{x}_0)\xrightarrow{\rho_1} q_1(\Vec{x}_1) \xrightarrow{\rho_2}_{\mathbb{Z}} q(\Vec{x})$ for some $\Vec{x}$ with $\Vec{x}^{-j} = \Vec{x}'$. In particular, for all $i \in \UBCounters(\pi_2, H_d) \setminus\{j\}$, we have $\Vec{x}(i) \ge U$. For the $j$-th coordinate, notice that $\Vec{x}_1(j) \ge H_d \ge U + L_{d-1}\cdot \Size(V)$ and $\Delta(\rho_2)(j) \le L_{d-1}\cdot \Size(V)$. Therefore $\Vec{x}(j) \ge U$ as well. Also verify that $\rho_2$ is indeed fireable from $q_1(\Vec{x}_1)$ with $j$-th coordinate included, as $\Drop(\rho_2)(j) \le L_{d-1}\cdot \Size(V) \le \Vec{x}_1(j)$. We are done by taking $\rho := \rho_1\rho_2$.
    \end{claimproof}
    It remains to verify that $H_d, L_d \le R(U, M)^{(d+1)!}$ for some polynomial $R$. It suffices to establish the following claim for $U,M \ge 1$, whose proof is delegated to the appendix.
    \begin{restatable}{claim}{HdLdLeqPolyPowFac}
        \label{claim:h_d-l_d-le-poly_pow_fac}
        $H_d \le (4UM^2)^{d!}$ and $L_d\le (4UM^2)^{(d+1)!}$.
    \end{restatable}
    \begin{claimproof}
        The case $d=0$ is trivial.
        Now assume that $H_{d-1} \le (4 U M^2)^{(d-1)!}$ for $d \ge 1$. Note that $L_{d-1} \le M\cdot (H_{d-1})^{d-1}+H_{d-1} \le 2 M \cdot (H_{d-1})^{d-1}$.
        By definition, we have
        \begin{align}
        \begin{split}
            H_d &\le U + (2M\cdot (H_{d-1})^{d-1})\cdot M \\
            &\le (4UM^2)\cdot (H_{d-1})^{d-1} \\
            &\le (4UM^2)^{(d-1)!\cdot (d-1)+1} \\
            &\le (4UM^2)^{d!}
        \end{split}
        \end{align}
        Finally, $L_d\le 2M \cdot (H_d)^d\le (4UM^2)^{1 + d!\cdot d}\le (4UM^2)^{(d+1)!}$.
    \end{claimproof}
    Let $R(U,M):=4(U+1)\cdot(M+1)^2$, and we are done.
\end{proof}

\section{Proofs Omitted from Section \ref{sec:d-vas}}

\subsection{Proof of Lemma \ref{lem:exists-normal-vector-with-wide-ortho}}

\existsNormalWideOrtho* 

\begin{proof}
    According to the Farkas-Minkowski-Weyl Theorem (see e.g.\ \cite[Corollary 7.1a]{DBLP:books/daglib/0090562}), there is a matrix $A \in \mathbb{Z}^{k \times d}$ for some $k$ such that $\Cone(X) = \{ \Vec{x} \in \mathbb{Q}^d \mid A\Vec{x} \ge \Vec{0} \}$. Let $\Vec{a}_1, \ldots, \Vec{a}_k \in \mathbb{Z}^d$ be the rows of $A$, we may take $\Vec{n}_0 := \sum_{i = 1}^k \Vec{a}_i$ as a candidate of $\Vec{n}$. 
    \begin{enumerate}
        \item Indeed, $\Inner{\Vec{a}_i, \Vec{x}} \ge 0$ for any $\Vec{x} \in \Cone(X)$, so $\Inner{\Vec{n}_0, \Vec{x}} \ge 0$ as well. 
        \item Now let $X_0 := \{ \Vec{x} \in X \mid \Inner{\Vec{n}_0, \Vec{x}} = 0 \}$. For all $\Vec{x} \in X_0$, we must have $\Inner{\Vec{a}_i, \Vec{x}} = 0$ for each $i \in [k]$ as $0 = \Inner{\Vec{n}_0, \Vec{x}} = \sum_{i = 1}^k \Inner{\Vec{a}_i, \Vec{x}}$ is a sum of non-negative values. Take $\Vec{y} \in \Span(X_0)$, then $\Inner{\Vec{a}_i, \Vec{y}} = 0$ for each $i \in [k]$. Therefore, $A\Vec{y} = 0$ which implies $\Vec{y} \in \Cone(X)$. Notice that in the expression of $\Vec{y}$ as a non-negative linear combination of vectors of $X$ we cannot have non-zero coefficients for any $\Vec{x}_+ \in X \setminus X_0$ as they have positive inner products with $\Vec{n}_0$. Thus $\Vec{y} \in \Cone(X_0)$. This justifies $\Cone(X_0) = \Span(X_0)$.
    \end{enumerate}
    
    We cannot derive a nice bound on $\norm{\Vec{n}_0}$ directly from the Farkas-Minkowski-Weyl Theorem. Fortunately, we can express candidates of the vector $\Vec{n}$ by the following system:
    \begin{align}
        \begin{split}
            R_0 \cdot \Vec{n} &= \Vec{0},\\
            R_1 \cdot \Vec{n} & \ge \Vec{1},
        \end{split}
        \label{eqn:normal-vector-ilp}
    \end{align}
    where $R_0$ is the matrix whose rows are exactly the vectors in $X_0$, and $R_1$ is the matrix whose rows are exactly the vectors in $X \setminus X_0$. Observe that any integral solution $\Vec{n}$ to (\ref{eqn:normal-vector-ilp}) satisfies conditions (1) and (2), even if $\Vec{n}$ is not parallel to $\Vec{n}_0$. Indeed, (\ref{eqn:normal-vector-ilp}) ensures that $\{ \Vec{x} \in X \mid \Inner{\Vec{n}, \Vec{x}} = 0 \}$ is exactly $X_0$, for which we have justified $\Cone(X_0) = \Span(X_0)$.
    
    We can now bound the norm of a minimal solution to (\ref{eqn:normal-vector-ilp}) using standard results in integer programming. Let $r := \dim(\Span(X))$, then any $(r\times r)$-subdeterminant of the matrix $\begin{pmatrix} R_0 & \Vec{0} \\ R_1 & \Vec{1} \end{pmatrix}$ is bounded in absolute value by $\norm{X}^r \cdot r! \le (r\norm{X})^r$. By \cite[Theorem 1]{vonzurGathen1978} there exists an integer solution $\Vec{n}$ to (\ref{eqn:normal-vector-ilp}) such that $\norm{\Vec{n}} \le (d + 1)\cdot (r\norm{X})^r$. Since $\Cone(X_0) = \Span(X_0) \subseteq \Span(X) \ne \Cone(X)$, we have $X_0 \subsetneq X$, which impplies that $\Vec{n} \ne \Vec{0}$.
\end{proof}

\subsection{Proof of Lemma \ref{lem:almost-bounded-or-unbounded}}

\almostBoundedOrUnbounded*

\begin{proof}
    We begin with a run $\pi$, which is not almost unbounded by $B$ in $V$. We claim that the empty run $s\xrightarrow{\varepsilon}s$ is a prefix of $\pi$ that satisfies $|\UBCounters(\varepsilon,B)| \le d-2$, that is, $s$ contains at least 2 coordinates bounded by $B$. Otherwise, $d-1$ coordinates of the initial configuration $s$ are at least $B$. Splitting $\pi$ as $\varepsilon\pi$, it then follows that $\pi$ is almost unbounded by $B$ in $V$, which yields a contradiction. Now we take the longest prefix $\pi_1$ of $\pi$ satisfying $|\UBCounters(\pi_1,B)| \le d-2$. If $\pi_1=\pi$, then take $\pi_2$ to be the empty string, and we show that $\pi$ is almost bounded by $B$. Otherwise, one may split $\pi$ into $\pi_1 t \pi_2$, where $t$ is a single transition of $V$. By the maximality of $\pi_1$, it must holds that $|\UBCounters(\pi_1t, B)|\ge d-1$. However, since $\pi$ is not almost unbounded by $B$, one has that $|\UBCounters(\pi_2
    , B+\Size(V))|\le |\UBCounters(\pi_2, B)|\le d-2$. Note that $t$ can raise the value of each coordinate by at most $\Size(S)$. We have $|\UBCounters(\pi_1t, B+\Size(V))|\le |\UBCounters(\pi_1, B)|\le d-2$. Therefore, $\pi$ is almost bounded by $B+\Size(V)$.
\end{proof}

\subsection{Proof of Lemma \ref{lem:lift-Z-run-by-pumpable}}

\liftZRunByPumpable*

\begin{proof}
    Assume that $V$ has transition set $T \subseteq \mathbb{Z}^d$. Consider any $\mathbb{Z}$-run $p(\Vec{x})\xrightarrow{*}_{\mathbb{Z}}q(\Vec{y})$, where $p(\Vec{x})$ is forward pumpable and $q(\Vec{y})$ is backward pumpable. Let $M:=\Size(V, p(\Vec{x}), q(\Vec{y}))$ and $r:=\dimcyc(S)\le d$. First, we claim that this $\mathbb{Z}$-run can be replaced by a short path.
    \begin{claim}\label{claim:short-Z-run}
        $p(\Vec{x})\xrightarrow{\varsigma}_{\mathbb{Z}} q(\Vec{y})$ for some path $\varsigma$ with $|\varsigma|\le2\cdot (M+1)^{r+2}$.
    \end{claim}
    \begin{claimproof}
        Since $V$ is sequential, for each $i \in [k]$, the transition $\Vec{a}_i$ appears exactly once in any $\mathbb{Z}$-run from $p(\Vec{x})$ to $q(\Vec{y})$. Let $\Vec{a}:=\sum_{i \in [k]}\Vec{a}_k$, $A \in  \mathbb{Z}^{d\times |T|}$ be the matrix whose column vectors are exactly those in $T$, and $\Vec{\lambda}\in \mathbb{N}^{|T|}$ be a variable vector. We may formalize the $\mathbb{Z}$-reachability relation as the following integer programming equation
        \begin{equation}\label{eqn:z-reach}
            A\Vec{\lambda}=\Vec{y}-\Vec{x}-\Vec{a},
        \end{equation}
        which has at least one nonnegative integer solution by assumption. Note that $r=\dim(\mathrm{rank}(A))$ and that the absolution value of any $(r\times r)$-subdeterminant of the matrix $\begin{pmatrix} A &\Vec{y}-\Vec{x}-\Vec{a}\end{pmatrix}$ is bounded by $r!\cdot(\norm{T}+\norm{\Vec x}+\norm{\Vec y}+\norm{\Vec a})^r\le M^r$. Again by \cite[Corollary 1]{vonzurGathen1978}, there exists a nonnegative integer solution $\Vec{\lambda}_0$ to \autoref{eqn:z-reach} such that $\norm{\Vec{\lambda}_0}\le (|T|+1)\cdot M^r \le (M+1)^{r+1}$.
        Moreover, $\Vec{\lambda}_0$ induces a cycle $\theta$ in the last SCC of $S$. Let $\varsigma:=\Vec{a}_1\dots\Vec{a}_k\theta$. Then 
        \begin{equation}
            |\varsigma|\le k+|T|\cdot \norm{\Vec{\lambda}_0}\le 2 \cdot (M+1)^{r+2}.
        \end{equation}
        Note that $p(\Vec{x})\xrightarrow{\varsigma}q(\Vec{y})$ is a $\mathbb{Z}$-run, which completes the proof.
    \end{claimproof}
    Since $p(\Vec{x})$ is pumpable in $S$, there exists a a cycle $\theta_1$ such that $\Delta(\theta_1)\ge \Vec{1}$ and $p(\Vec{x})\xrightarrow{\theta_1^n} p(\Vec{x}+n\cdot \Delta(\theta_1))$ holds for every $n \in \mathbb{N}$. Similarly, there exists a cycle $\theta_2$ such that $-\Delta(\theta_2)\ge \Vec{1}$ and $q(\Vec{y}+n\cdot \Delta(\theta_2))\xrightarrow{\theta_1^n} q(\Vec{y})$ holds for every $n \in \mathbb{N}$. 
    Moreover, by the upper bound of Rackoff~\cite[Lemma 3.4]{DBLP:journals/tcs/Rackoff78}, there exists a polynomial $P_d$ such that the length of such cycles in $S$ is at most $P_d(M)$. Hence, we may assume that $|\theta_1|, |\theta_2|\le P_d(M)$. 
    
    Since $\Delta(\theta_1), \Delta(\theta_2) \in \Cone(V)$, by the wideness of $S$, $-\Delta(\theta_2)-\Delta(\theta_1)\in\Cone(V)$. By the fundamental theorem of linear inequalities~\cite[Theorem 7.1]{DBLP:books/daglib/0090562}, it can be written as a non-negative linear combination of linearly independent vectors $\Vec{e}_1,\dots,\Vec{e}_r$ from $T$. Hence, there exist non-negative integers $\{\lambda_i\}_{i \in [0,r]}$ such that
    \begin{equation}
        \lambda_0 \cdot (-\Delta(\theta_2)-\Delta(\theta_1))=\sum_{i \in [r]} \lambda_i \cdot \Vec{e}_i,
    \end{equation}
    where $\lambda_0>0$ holds in particular. By Pottier's Theorem~\cite[Theorem 1]{pottier1991minimal}, we may further assume that $|\lambda_i|\le \Lambda:= (1+r\norm{T}+2P_d(M)\norm{T})^r$ for each $i \in [0,r]$. Let $o$ be the cycle induced by $\Vec{e}_1^{\lambda_1}\dots \Vec{e}_r^{\lambda_r}$ in the first SCC of $S$ and let $n_1:=\norm{\Drop(o)}\le r\norm{T}\cdot \Lambda$.
    \begin{claim}
        The following is a legal run:
        \begin{align} 
        p(\Vec{x})\xrightarrow{\theta_1^{n_1\lambda_0}} p(\Vec{x}+n_1\lambda_0\cdot \Delta(\theta_1)) \xrightarrow{o^{n_1}} p(\Vec{x}-n_1\lambda_0 \cdot \Delta(\theta_2)).
        \end{align}
    \end{claim}
    \begin{claimproof}
        Note that $\Delta(o^{n_1})=n_1\sum_{i \in [r]}\lambda_i\cdot\Vec{e}_i=-n_1\lambda_0\cdot\Delta(\theta_2)-n_1\lambda_0\cdot\Delta(\theta_1)$, which implies that this is a $\mathbb{Z}$-run. It is straightforward that $\theta_1^{n_1\lambda_0}$ is a run. Since $\Drop(o)\le \norm{\Drop(o)}\cdot \Vec{1} = n_1 \cdot \Vec{1} \le \Vec{x}+n_1\lambda_0 \cdot \Delta(\theta_1)$, $o$ is fireable from $p(\Vec{x}+n_1\lambda_0\cdot \Delta(\theta_1))$. Similarly, $o^{\text{rev}}$ is also fireable from $p(\Vec{x}-n_1\lambda_0 \cdot \Delta(\theta_2))$ in $S^{\text{rev}}$. Therefore, both the first and the last iteration of $o^{n_1}$ are runs, which implies that the entire segment is a run.
    \end{claimproof}
    Let $\theta_1':=\theta_1^{n_1\lambda_0}o^{n_1}$. Note that $\Delta(\theta_1') = -n_1\lambda_0 \cdot \Delta(\theta_2) \ge \Vec{1}$. Hence, we can replace $\theta_1$ by $\theta_1'$ as a new pumping cycle. Let $\varsigma$ be the short path obtained by \autoref{claim:short-Z-run} and let $n_2:=|\varsigma|\cdot \norm{T}$. Construct a new path $\rho:=(\theta_1')^{n_2}\varsigma\theta_2^{n_1n_2\lambda_0}$. We claim that $p(\Vec{x})\xrightarrow{\rho} q(\Vec{y})$ is a run. It suffices to show that the infix $\varsigma$ is a legal run:
    \begin{equation}\label{eqn:pumped-z-run}
        p(\Vec{x}+n_2\cdot\Delta(\theta_1')) = p(\Vec{x}-n_2n_1\lambda_0\cdot \Delta(\theta_2)) \xrightarrow{\varsigma} q(\Vec{y}-n_1n_2\lambda_0\cdot\Delta(\theta_2)).
    \end{equation}
    Since $\Vec{x}-n_2n_1\lambda_0\cdot \Delta(\theta_2) \ge n_2\cdot \Vec{1} \ge \Drop(\varsigma)$, \autoref{eqn:pumped-z-run} follows immediately. Finally, we bound the length of $\rho$ as follows
    \begin{equation}
        \begin{aligned}
            |\rho|&\le n_1n_2\cdot\Lambda\cdot(|\theta_1|+|\theta_2|+r)+|\varsigma|\\
            &\le 4(r\norm{T})^2\cdot\Lambda^2\cdot|\varsigma|\cdot P_d(M)\\
            &\le 3^{2r+2}\cdot P_d(M) ^{2r+1} \cdot (M+1)^{3r+4}.
        \end{aligned}
    \end{equation}
    We are done by taking $L_d(x):=3^{2d+2}\cdot P_d(x)^{2d+1}\cdot (x+1)^{3d+4}$.
\end{proof}

\subsection{Proof of Lemma~\ref{lem:reduction-for-bounded}}

\reductionForBounded*

\begin{proof}
    Assume that $V=(Q,T)$. Let $I$ be the set of $k$ coordinates bounded by $B$ along $\pi$. For any vector $\Vec{z} \in \mathbb{Q}^d$, we write $\Vec{z}|_I$ for the projected vector in $\mathbb{Q}^I$ such that $\Vec{z}|_I(i) = \Vec{z}(i)$ for all $i \in I$. We construct a VASS $V':=(Q',T')$, where $Q':=Q\times [0,B-1]^{I}$ and
    \begin{align}
    \begin{split}
        T':=\bigl\{ ((p,\Vec{u}), \Vec{a}, (q, \Vec{v}))\mid{} & \Vec{u}, \Vec{v}\in [0,B-1]^{I}, (p,\Vec{a},q)\in T, \text{and }\Vec{u}+\Vec{a}|_I=\Vec{v} \bigr\}.
    \end{split}
    \end{align}
    For any configuration $p(\Vec{u})$ in $V$, let $\llbracket p(\Vec{u})\rrbracket :=(p,\Vec{u}|_I)(\Vec{u})$. Clearly, if $\Vec{u}(i)<B$ for all $i \in I$, then $\llbracket p(\Vec{u})\rrbracket $ is a configuration of $V'$. Let $s':=\llbracket s\rrbracket $ and $t':=\llbracket t\rrbracket $. It is straightforward that $\pi$ induces a run $s'\xrightarrow{\pi'}t'$ with the same length by applying $\llbracket\cdot\rrbracket $ pointwise to each configuration. Note that the $i$-th coordinates of the configurations in $\pi$ are less than $B$ for all $i\in I$. We deduce that $\pi'$ is indeed a run in $V'$. Hence, $|\pi'|=|\pi|\in \Len(V',s',t')$. By induction on the length of runs, one can easily verify the following claim.
    \begin{claim}\label{claim:exists-original-run}
        If $s'\xrightarrow{\rho'}c'$ for some $\rho'$ and $c'$ in $V'$, then there exists a run $s\xrightarrow{\rho}c$ in $V$ with $|\rho|=|\rho'|$, where $c$ is the unique configuration satisfying $\llbracket c\rrbracket =c'$.
    \end{claim}
    Therefore, any run $s'\xrightarrow{\rho'}t'$ in $V'$ induces a run $s\xrightarrow{\rho}t$ in $V$ by dropping the second components of all the states for each configuration, which implies $\Len(V',s',t') \subseteq \Len(V,s,t)$. 
    Considering the geometric dimension, any cycle $\theta$ in $V'$ must satisfy $\Delta(\theta)|_I=\Vec{0}$. Hence, it holds that $\dimcom(V')\le d - |I| = d - k$. Finally, since $|Q'|\le B^{|I|} \cdot |Q|$, $\norm{T'}\le \norm{T}$, and $|T'| \le B^{|I|}\cdot |T|$, we have $\Size(V',s',t')\le B^{k}\cdot \Size(V,s,t)$.
\end{proof}

\subsection{Proof of Lemma~\ref{lem:dimension-reduction-for-almost-bounded-runs}}

\dimensionReductionForAlmostBoundedRuns*

\begin{proof}
    By assumption, we can partition $\pi$ into two segments $s\xrightarrow{\pi_1}q(\Vec{w})\xrightarrow{\pi_2}t$ such that $|\UBCounters(\pi_1, B)|\le d$ and $|\UBCounters(\pi_2, B)|\le d$. Let $I_1:=[d+2]\setminus \UBCounters(\pi_1,B)$ and $I_2:=[d+2]\setminus \UBCounters(\pi_2, B)$. By \autoref{lem:reduction-for-bounded}, one can construct a $(d+2)$-VASS $V_1=(Q_1,T_1)$ with configurations $s'$ such that $s'\xrightarrow{*}c_1:=(q,\Vec{w}|_{I_1})(\Vec{w})$, as well as a $(d+2)$-VASS $V_2=(Q_2,T_2)$ with configurations $t'$ such that $c_2:=(q,\Vec{w}|_{I_2})(\Vec{w})\xrightarrow{*}t'$. For clarity, for $i = 1, 2$, we rename every state $(p,\Vec{u})$ in $V'$ as $(p^i,\Vec{u})$, so that $Q_1\cap Q_2=\emptyset$. Define $V':=(Q',T')$, where $Q':=Q_1\cup Q_2$ and 
    \begin{align}
    \begin{split}
        T' := T_1\cup T_2 \cup \{ & ((p^1,\Vec{u}_1), \Vec{0},(p^2,\Vec{u}_2)) \mid p \in Q, \\
        & \Vec{u}_1\in [0,B-1]^{I_1}, \Vec{u}_2\in[0,B-1]^{I_2} \text{ and } \\ 
        & \exists \Vec{u} \in \mathbb{N}^d, \Vec{u}_1 = \Vec{u}|_{I_1}, \Vec{u}_2 = \Vec{u}|_{I_2}\}.
    \end{split}
    \end{align}
    Intuitively, we can view $V_1$ and $V_2$ as two sub-VASS sequentially connected by some dummy transitions. Note that $s'$ and $t'$ are still configurations in $V'$. Moreover, $c_1\xrightarrow{t}c_2$ holds, where $t:=((q^1, \Vec{w}|_{I_1}), \Vec{0}, (q^2, \Vec{w}|_{I_2}) \in T'$. Hence, $s'\xrightarrow{*}t'$ in $V'$.
    \begin{claim}
        $\min \Len(V,s,t)+1\le \min\Len(V',s',t')$.
    \end{claim}
    \begin{claimproof}
        Since $s'\xrightarrow{*}t'$ holds in $V'$, let $\rho$ be the shortest witness with $|\rho|=\min \Len(V',s',t')$. Since $\rho$ can only go from $V_1$ to $V_2$ through a dummy transition $t:=((p^1,\Vec{u}_1), \Vec{0},(p^2,\Vec{u}_2))$, where $p \in Q$, we can divide $\rho$ into $s'\xrightarrow{\rho_1}(p^1,\Vec{u}_1)(\Vec{v})\xrightarrow{t}(p^2,\Vec{u}_2)(\Vec{v})\xrightarrow{\rho_2}t'$ for some $\Vec{v}\in \mathbb{N}^{d+2}$. Since $\rho_1$ is a run in $V_1$, by \autoref{claim:exists-original-run}, there exists a run $s\xrightarrow{\rho_1'} p(\Vec{v})$ in $V$ with $|\rho_1'|=|\rho_1|$. Similarly, there exists a run $p(\Vec{v})\xrightarrow{\rho_2'}t$ in $V$ with $|\rho_2'|=|\rho_2|$. Therefore, $s\xrightarrow{\rho_1'\rho_2'}t$ in $V$ with length $|\rho|-1$, which implies $\min \Len(V,s,t)+1\le |\rho|=\min\Len(V',s',t')$.
    \end{claimproof} 
    Note that $\dimcom(V')\le \max\left(\dimcom(V_1), \dimcom(V_2)\right)\le d$. Also, $\Size(V',s',t')\le 4B^{4}\cdot\Size(V,s,t)^2$. We are done by taking $A_d(x,y):=4x^{4}y^2$.
\end{proof}

\lenGPlusOneVasLeLenVassMaxDimG*

\begin{proof}
    We repeat the analysis in \autoref{thm:len-d-plus-2-vas-le-len-vass-max-dim-d}. Observe that the polynomial bounds obtained in this section are essentially double exponential in the dimension $d$, where the $(d+1)!$ term in the exponent originates from Rackoff’s extraction (\autoref{lem:reach-unbounded-of-ever-unbounded}). Hence, the size amplification is controlled by $(P_g(\Size(V, \Vec{s}, \Vec{t})))^{(d+1)!}$ for some polynomial $P_g$. Then we bound the component-dimension of the new VASS. In the almost unbounded case, the component-dimension of the resulting VASS is exactly $0$. In the almost bounded case, we encode two coordinates of the original VASS into states with the help of \autoref{lem:dimension-reduction-for-almost-bounded-runs}, where we concatenate two VASSes obtained by \autoref{lem:reduction-for-bounded}. Therefore, it suffices to show that the dimension reduction in \autoref{lem:reduction-for-bounded} decreases the cycle-dimension of the sequential VAS $S$ by at least one. Note that after encoding any $i \in [d]$ into the state, every cycle $\theta$ in the resulting VASS, denoted by $V'$, satisfies $\Delta(\theta)(i)=0$. In contrast, the original sequential VAS $S$ contains no fixed coordinates, i.e., there exists a self-loop $u$ with $\Delta(u)(i)\ne 0$. Consequently, we have $\CycleSpace(V') \subsetneq \CycleSpace(S)$, which implies that $\dimcyc(V')\le \dimcyc(S) -1 \le  g$.
\end{proof}

\section{Proofs Omitted from Section \ref{sec:low-dim-vas-boundedness}}

\subsection{Proof of Lemma \ref{lem:len-of-neg-col-vas-le-d-minus-2-vass}}
\lenOfNegColVAS*

\begin{proof}
    Assume that $S=V[\Vec{a}_1,\ldots,\Vec{a}_k]$ contains two negatively collinear coordinates $i, j$ and that $s=p_0(\Vec{x})$. Let $M:=\Size(S,s,t)$. Since $i$ and $j$ are negatively collinear, there exists $\alpha<0$ such that $\Vec{u}(i)=\alpha\cdot \Vec{u}(j)$ for every $\Vec{u}$ in $V$. We first bound the absolute value of $\alpha$: take any $\Vec{u}$ in $V$ with $\Vec{u}(i)\ne 0$, and it follows that $|\alpha|=\frac{|\Vec{u}(j)|}{|\Vec{u}(i)|}\ge \frac{1}{M}$.
    We prove the following claim.
    \begin{claim}\label{claim:bounded-on-i-j}
        Let $s\xrightarrow{\pi}t$ be a run in $S$. For any configuration $p(\Vec{w})$ along $\pi$, it holds that $\Vec{w}(i),\Vec{w}(j)< B$, where $B:=3M^2+2M$.
    \end{claim}
    \begin{claimproof}
        To reach a contradiction, we assume that there exists a configuration $p(\Vec{w})$ in $\pi$ such that, w.l.o.g., $\Vec{w}(j)\ge B$. Since all the bridges occur exactly once in $\pi$, one obtain that $\Vec{\delta}:=\Vec{w}-\Vec{x}-\Vec{a}\in \Cone(V)$, where $\Vec{a}:=\sum_{\ell \in [k]}\Vec{a}_\ell$. By the collinearity of $i$ and $j$, we further deduce that
    \begin{equation}
        \Vec{\delta}(i)=\alpha\cdot (\Vec{w}-\Vec{x}-\Vec{a})(j)\le -\frac{1}{M}\cdot (3M^2+2M-M-M)=-3M.
    \end{equation}
    Now we have $\Vec{w}(i)=\Vec{x}(i)+\Vec{a}(i)+\Vec{\delta}(i)\le 2M-3M<0$, which is a contradiction. 
    \end{claimproof}
    We encode the $i$- and $j$-th coordinates into states and obtain a new $(d-2)$-VASS $(V',s',t')$, whose size are bounded by $B^{2}\cdot M=M^{O(1)}$. Since $S$ contains no fixed coordinates, we have $\CycleSpace(V')\subsetneq \CycleSpace(S)$ and therefore $\dimcyc(V') \le \dimcyc(S)-1$. For any run $s\xrightarrow{\pi}t$, by \autoref{claim:bounded-on-i-j}, $\pi$ is bounded by $B$ on both $i$-th and $j$-th coordinates. This run must be captured by a run in $V'$ with the same length. On the other hand, any run in $V'$ also induces a run in $V$ with the same length. Therefore, we have $\Len(S,s,t)=\Len(V',s',t')$.
\end{proof}

\subsection{Proof of Lemma \ref{lem:len-wide-vas}}

\lenWideVAS*

\begin{proof}
    Let $L_d$ and $P$ be the polynomials define in \autoref{lem:lift-Z-run-by-pumpable} and \autoref{thm:reach-pumpable-of-ever-unbounded-d-1}, respectively. W.l.o.g., we assume that they are nondecreasing with respect to $d$ and each parameter. It suffices to consider a wide $d$-VAS $(V,\Vec{s},\Vec{t})$ containing no fixed coordinates with $\Vec{s}\xrightarrow{\pi} \Vec{t}$ in $V$. Let $M:=\Size(V, \Vec{s},\Vec{t})$ and $U:=P(M,M)^{(d+1)!}$. 
    
    In case $\pi$ is $2$-bounded by $U$, i.e.\ $|\UBCounters(\pi, U)| \le d - 2$. Suppose $i, j \in [d] \setminus \UBCounters(\pi, U)$ are two coordinates bounded by $U$ along $\pi$. We may encode into states all possible values in $[0, B]$ on these two coordinates and obtain the $(d-2)$-VASS $V'$. Clearly $\Size(V') \le U^2 \cdot\Size(V)$. Also observe that $\CycleSpace(V') = \{\Vec{c} \in \CycleSpace(V) \mid \Vec{c}|_{\{i, j\}} = \Vec{0}\} \subsetneq \CycleSpace(V)$, thus $\dimcyc(V')\le \dimcyc(V)-1$.
    
    In case $\pi$ is not $2$-bounded by $U$, we have $|\UBCounters(\pi, U)|\ge d-1$. In this case, \autoref{lem:poly-bound-for-almost-unbounded} does not apply, but similarly, we can prove the following claim.
    \begin{claim}\label{claim:exists-poly-run-of-2-unbounded}
        There exists a run $\Vec{s}\xrightarrow{\rho}\Vec{t}$ of length $|\rho|\le L_d(M)$.
    \end{claim}
    \begin{claimproof}
        Note that $V$ is wide. Since $|\UBCounters(\pi, U)|\ge d-1$, by \autoref{thm:reach-pumpable-of-ever-unbounded-d-1}, there exists a run $\pi'$ such that $\Vec{s}\xrightarrow{\pi'}\Vec{w}$ for some configuration $\Vec{w}$, which is forward pumpable in $V$. Note that there is only one state in $V$. We deduce that $\Vec{s}$ is forward pumpable. Similarly, the target configuration $\Vec{t}$ is backward pumpable. Applying \autoref{lem:lift-Z-run-by-pumpable} to $\pi$, we obtain that $\Vec{s}\xrightarrow{\rho}\Vec{t}$ for some $\rho$ with $|\rho|\le L_d(M)$. 
    \end{claimproof}
    It is trivial that we can construct a geometrically $0$-dimensional $0$-VASS $(V',s',t')$ of size bounded by $L_d(M)+1$ to capture the short run $\rho$ obtained by \autoref{claim:exists-poly-run-of-2-unbounded}. 
    
    Combining the above analysis, we conclude that $(V,\Vec{s},\Vec{t})$ is length-bounded by a $(d-2)$-VASS $(V', s',t')$ with $M_d$-amplification, where $\dimcyc(V')\le \dimcyc(V)-1$ and the polynomial $M_d$ is determined implicitly in this proof.
\end{proof}

\subsection{Proof of Lemma \ref{lem:quasi-bounded-or-unbounded}}

\quasiBoundedOrUnbounded*

\begin{proof}
    Consider a run $\pi$ not quasi-unbounded by $B$ in $V$. Assume that $s=p(\Vec{x})$. We claim that the empty run $s\xrightarrow{\varepsilon}s$ is a prefix of $\pi$ that satisfies $[d]\setminus\UBCounters(\varepsilon, B)$ contains at least one non-collinear pair, that is, $s$ contains at least two coordinates bounded by $B$ which are non-collinear. Otherwise, all coordinates in $\{i \in [d]\mid \Vec{x}(i)<B\}$ are collinear. Since $V$ contains no negatively collinear coordinates, we further deduce that they are positively collinear. Splitting $\pi$ as $\varepsilon\pi$, it then follows that $\pi$ is quasi-unbounded by $B$ in $V$, which yields a contradiction. Now we take the longest prefix $\pi_1$ of $\pi$ satisfying $[d]\setminus\UBCounters(\pi_1, B)$ contains at least one non-collinear pair. If $\pi_1=\pi$, then take $\pi_2$ to be the empty string, and we show that $\pi$ is quasi-bounded by $B$. Otherwise, one may split $\pi$ into $\pi_1 t \pi_2$, where $t$ is a single transition of $V$. By the maximality of $\pi_1$, it must hold that all coordinates in $[d]\setminus\UBCounters(\pi_1t, B)$ are pairwise positively collinear. However, since $\pi$ is not quasi-unbounded by $B$, one has that $[d]\setminus\UBCounters(\pi_2
    , B+\Size(V))\supseteq [d]\setminus\UBCounters(\pi_2, B)$ contains at least two non-collinear coordinates. Note that $t$ can raise the value of each coordinate by at most $\Size(S)$. We have $[d]\setminus\UBCounters(\pi_1t, B+\Size(V))\supseteq [d]\setminus\UBCounters(\pi_1, B)$. Therefore, $\pi$ is quasi-bounded by $B+\Size(V)$.
\end{proof}

\subsection{Proof of Lemma \ref{lem:poly-bound-for-quasi-unbounded}}

\polyBoundForQuasiUnbounded*

\begin{proof}
    Take $U_d$, $P$ be the polynomials given in \autoref{lem:poly-bound-for-almost-unbounded} and \autoref{lem:tgt-pumpable-of-unbounded-d-1}. Let $M:=\Size(S,s,t)$ and $N:=P(M,M)^{(d+1)!}$. By quasi-unboundedness of $\pi$, $s\xrightarrow{\pi_1}q(\Vec{w})\xrightarrow{\pi_2}t$ for some configuration $q(\Vec{w})$, where the coordinates in $[d]\setminus\UBCounters(\pi_i,U_d(M))$ are pairwise positively collinear for $i=1,2$. We have the following claim.
    \begin{claim}
        There exists a run $\rho_1$ with $|\rho_1|\le N$ such that $s\xrightarrow{\rho_1} q_1(\Vec{x}_1)$ for some forward pumpable configuration $q_1(\Vec{x}_1)$, and moreover, $q_1(\Vec{x}_1) \xrightarrow{*}_{\mathbb{Z}} q(\Vec{w})$.
    \end{claim}
    \begin{claimproof}
        This claim is essentially \autoref{thm:reach-pumpable-of-ever-unbounded-d-1}. First, by \autoref{lem:reach-unbounded-of-ever-unbounded}, we can obtain a run $\rho_1$ with $|\rho_1|\le N$ such that $s\xrightarrow{\rho_1} q_1(\Vec{x}_1)\xrightarrow{*}_{\mathbb{Z}}q(\Vec{w})$ for some $q_1(\Vec{x}_1)$ with $\Vec{x}_1(i)\ge B(M,M)$ for each $i \in \UBCounters(\pi, U_d(M))$, where $B$ is the polynomial given in \autoref{lem:reach-unbounded-of-ever-unbounded}. Observe that coordinates in the set $\{i \in [d] \mid \Vec{x}_1(i) < B(\norm{s}, \Size(S))\}$ is contained in $[d]\setminus \UBCounters(\pi_1,U_d(M))$, which are pairwise positively collinear. Applying \autoref{lem:cfg-pumpable-of-bounded-pos-col}, we obtain that $q_1(\Vec{x}_1)$ is forward pumpable.
    \end{claimproof}
    Following a symmetric argument to $\pi_2^{\mathrm{rev}}$ and $S^{\mathrm{rev}}$, we also have $q(\Vec{w})\xrightarrow{*}_{\mathbb{Z}}q_2(\Vec{x}_2)\xrightarrow{\rho_2}t$, where $|\rho_2|\le N$ and $q_2(\Vec{x}_2)$ is backward pumpable. Since $q_1(\Vec{x}_1)\xrightarrow{*}_{\mathbb{Z}}q_2(\Vec{x}_2)$, by \autoref{lem:lift-Z-run-by-pumpable}, there exits a run $q_1(\Vec{x}_1)\xrightarrow{\rho'}q_2(\Vec{x}_2)$ with $|\rho'|\le L_d(M+2NM)$ for the polynomial $L_d$ from \autoref{lem:lift-Z-run-by-pumpable}. Let $\rho:=\rho_1\rho'\rho_2$, and then $s\xrightarrow{\rho}t$ with $|\rho|\le 2N+L_d(N+2NM)=U_d(M)$.
\end{proof}

\subsection{Proof of Lemma~\ref{lem:len-wide-sequential-vas-finer}}

\lenWideSequentialVASFiner*

\begin{proof}
    Let $(S, s, t)$ be a wide sequential $d$-VAS with $s \xrightarrow{\pi}t$ for some $\pi$. Let $M:=\Size(S,s,t)$ and $U:=U_d(M)$, where $U_d$ is the polynomials defined in \autoref{lem:poly-bound-for-quasi-unbounded}. Applying \autoref{lem:quasi-bounded-or-unbounded} to $\pi$, there are still two possible cases. In the case where $\pi$ is quasi-unbounded by $U$, the result follows from \autoref{lem:poly-bound-for-quasi-unbounded}. We now focus on the case where $\pi$ is quasi-bounded by $U+\Size(S)$. By \autoref{lem:dimension-reduction-for-almost-bounded-runs}, there is a VASS $(V',s',t')$ with $s'\xrightarrow{*}t'$ such that $\min\Len(V,\Vec{s},\Vec{t})\le \min \Len(S,s,t)\le \min\Len(V's',t')$ and $\Size(V',s',t')\le A_d(U+M,M)$, where $A_d$ is the polynomial in \autoref{lem:dimension-reduction-for-almost-bounded-runs}. Looking at the construction in the proof of \autoref{lem:dimension-reduction-for-almost-bounded-runs}, $V'$ is obtained by concatenating two sub-VASSes $V_1$ and $V_2$ with a transition of zero effect, and both are obtained by encoding two coordinates into states (but not deleting these coordinates). We are left to bound their cycle dimensions.
    \begin{claim}
        $\dimcyc(V_1)\le \dimcyc(S) - 2$.
    \end{claim}
    
    \begin{claimproof}
        Choose a basis $\Vec{u}_1, \ldots, \Vec{u}_r$ of $\CycleSpace(S)$ to form a matrix $A\in\mathbb{Q}^{d\times r}$, where $r:=\dimcyc(S)$. Note that $A$ gives the isomorphsm from $\mathbb{Q}^r$ to $\CycleSpace(S)$. Assume that $V_1$ is obtained by encoding the $i$- and $j$-th coordinates into states for some $i,j\in [d]$. Let $R\in\mathbb{Q}^{2\times r}$ be the matrix formed by the $i$- and $j$-th rows of $A$. We emphasize that the two bounded coordinates are not collinear and hence, $\text{rank}(R)=2$. Then $R$ determines a subspace $X:=\{A\Vec{x}\mid \Vec{x} \in \mathbb{Q}^r\text{ and } R\Vec{x}=\Vec{0}\}$ of $\CycleSpace(S)$. Consider any cycle in $V_1$ with effect $\Vec{u}$, and we have $\Vec{u}\in \CycleSpace(V_1)$, which implies that there exists a unique vector $\Vec{x}\in\mathbb{Q}^r$ such that $\Vec{u}=A\Vec{x}$. By construction, we have $\Vec{u}(i)=\Vec{u}(j)=0$, or equivalently, $R\Vec{x}=\Vec{0}$. Therefore, $\Vec{u}\in X$, which implies that $\CycleSpace(V_1)\subseteq X$. Now, since $A$ have full column rank, we have $\dimcyc(V_1)\le \dim(X)=\dim(\mathrm{null}(R))=r-2$.
    \end{claimproof}
    
    Similarly, we have $\dimcyc(V_2)\le \dimcyc(S)-2$, which completes the proof.
\end{proof}

\section{Proofs Omitted from Section \ref{sec:low-dim-vas}}

\subsection{Proof of Lemma \ref{lem:len-4-vas-le-3-exp}}

\lenFourVASLeThreeEXP*

\begin{proof}
    We take $\rho$ to be the shortest run witnessing $s\xrightarrow{*}t$. Let $M:=\Size(V,\Vec{s},\Vec{t})$. By \autoref{lem:len-d-vas-le-cases}, there exists a VASS $(V',s',t')$ with $\Size(V',s',t')=M^{O(1)}$ such that $s'\xrightarrow{*}t'$ and $|\rho|\le \min\Len(V',s',t')$. Moreover, $V'$ is either a $3$-VASS or a $5$-VASS with component-dimension at most $2$. 
    
    In the case that $V'$ is a $3$-VASS, we can apply \cite[Lemma 2]{DBLP:conf/icalp/CzerwinskiJ0O25} to obtain the bound $|\rho| \le \min\Len(V',s',t')\le M^{2^{2^{M^{O(1)}}}}$. 
    
    The case where $V'$ is a $5$-VASS with component-dimension at most $2$ requires more careful investigation of previous results in \cite{DBLP:conf/icalp/FuYZ24}. Let $\rho'$ be a run with minimum length in $V'$ from $s'$ to $t'$. For the purpose of bounding the length of $|\rho|$, we may remove the states in $V'$ not visited on $\rho'$ and assume that the form of $V'$ is a sequence of strongly connected geometrically 2-dimensional sub-VASSes. According to \cite[Theorem 3.4, Lemma 4.2]{DBLP:conf/icalp/FuYZ24}, the reachability relation of each geometrically 2-dimensional sub-VASSes $V_s$ is characterized by a system of integer linear programming $A_s \Vec{\lambda}_s \ge \Vec{b}_s$ where the numbers of variables and inequalities are bounded by $\Size(V_s)^{O(1)}$ and $\norm{A_s}, \norm{\Vec{b}_s} \le \Size(V_s)^{O(1)}$.
    Moreover, any solution $\Vec{\lambda}_s$ induces a run of length $\norm{\Vec{\lambda}_s} \cdot \Size(V_s)^{O(1)}$.
    The number of sub-VASSes in $V'$ is at most $\Size(V')$, and we take the conjunction $A \Vec{\lambda} \ge \Vec{b}$ of all their characterization systems of integer linear programming. Using the bound due to Pottier \cite{pottier1991minimal}, there is a solution $\Vec{\lambda}$ with $\norm{\Vec{\lambda}} \le \Size(V')^{\Size(V')^{O(1)}}$. Therefore, the length of the shortest runs from $s', t'$ is bounded by $\Size(V')^{\Size(V')^{O(1)}} \cdot \Size(V')^{O(1)} \le 2^{M^{O(1)}}$.
\end{proof}

\subsection{Proof of Lemma \ref{lem:2d-space-projective-iff}}

\twoDSpaceProjectiveIff*

\begin{proof}
    ($\implies$) Suppose $X$ is projective w.r.t.\ $\Vec{u}, \Vec{v}$. Then there exist $i, j \in [d]$ such that $\Cone(\{\Vec{r}_1, \ldots, \Vec{r}_d\}) = \Cone(\{\Vec{r}_i, \Vec{r}_j\})$. Applying \autoref{lem:exists-normal-vector-with-wide-ortho} to the set $X_{ij} := \{\Vec{r}_i, \Vec{r}_j\}$, we obtain a vector $\Vec{n}$ which has non-negative inner product with $\Vec{r}_i$ and $\Vec{r}_j$. Let $X_0 := \{\Vec{r} \in X_{ij} \mid \Inner{\Vec{n}, \Vec{r}} = 0\}$, we also have $\Cone(X_0) = \Span(X_0)$. We argue that $X_0 = \emptyset$. Clearly, $X_0$ cannot be a singleton set. Notice that $\Vec{r}_i$ and $\Vec{r}_j$ must be linearly independent as $X$ is 2-dimensional. If $X_0 = X_{ij}$ then $-\Vec{r}_i \in \Span(X_0) = \Cone(X_0)$ so $-\Vec{r}_i = \alpha \Vec{r}_i + \beta \Vec{r}_j$ for some $\alpha, \beta \ge 0$. But then $(\alpha + 1)\Vec{r}_i + \beta \Vec{r}_j = \Vec{0}$ which contradicts to the linear independence of $\Vec{r}_i$ and $\Vec{r}_j$. Therefore $X_0 = \emptyset$, which means that $\Inner{\Vec{n}, \Vec{r}_i} > 0$ and $\Inner{\Vec{n}, \Vec{r}_j} > 0$. As every non-zero $\Vec{r}_\ell$ is a non-negative linear combination of $\Vec{r}_i$ and $\Vec{r}_j$, we have $\Inner{\Vec{n}, \Vec{r}_\ell} > 0$ as well.
    
    ($\impliedby$) Suppose $\Vec{n} \in \mathbb{Q}^2$ is a vector such that $\Inner{\Vec{n}, \Vec{r}_\ell} > 0$ for all $\ell \in [d]$ where $\Vec{r}_\ell \ne \Vec{0}$. Let $R \subseteq \{\Vec{r}_1, \ldots, \Vec{r}_d\}$ be minimal such that $\Cone(\{\Vec{r}_1, \ldots, \Vec{r}_d\}) = \Cone(R)$, we need to show that $|R| \le 2$. Suppose otherwise $R$ contains 3 (non-zero) vectors $\Vec{r}_i, \Vec{r}_j, \Vec{r}_k$, they must be linearly dependent. We may write $\Vec{r}_i = \lambda \Vec{r}_j + \mu \Vec{r}_k$ for some $\lambda, \mu \in \mathbb{Q}$. By renaming these vectors we may further assume $\lambda$ and $\mu$ have the same sign. If $\lambda, \mu \ge 0$ then $\Vec{r}_i \in \Cone(\{\Vec{r}_j, \Vec{r}_k\})$ so $\Cone(R) = \Cone(R \setminus \{\Vec{r}_i\})$, which contradicts the minimality of $R$. If $\lambda, \mu \le 0$ then $\Inner{\Vec{n}, \Vec{r}_i} = \lambda \Inner{\Vec{n}, \Vec{r}_i} + \mu \Inner{\Vec{n}, \Vec{r}_j} < 0$, which contradicts the property of $\Vec{n}$. Thus, $R$ cannot contain more than 2 vectors.
\end{proof}

\subsection{Proof of Lemma \ref{lem:lem:2d-space-projective-independent-base}}

\twoDSpaceProjectiveIndependentBase*

\begin{proof}
    By scaling we may assume $\Vec{u}, \Vec{v}, \Vec{x}, \Vec{y}$ are integer vectors. Let $R_{uv} := \begin{pmatrix} \Vec{u} & \Vec{v} \end{pmatrix}$ be the $d \times 2$ matrix whose colums are $\Vec{u}$ and $\Vec{v}$. Similarly let $R_{xy} := \begin{pmatrix} \Vec{x} & \Vec{y} \end{pmatrix}$. Then there is an invertible matrix $M \in \mathbb{Q}^{2 \times 2}$ such that $R_{uv} = R_{xy} \cdot M$. Let $\Vec{b} \in \mathbb{Q}^d$ be such that  
    \begin{align}
        \Vec{b}(\ell) = \begin{cases}
            1 & (\Vec{u}(\ell), \Vec{v}(\ell)) \ne \Vec{0},\\
            0 & (\Vec{u}(\ell), \Vec{v}(\ell)) = \Vec{0}.
        \end{cases} \quad \forall \ell \in [d]
    \end{align}
    Suppose $X$ is projective w.r.t.\ $\Vec{u}, \Vec{v}$. By \autoref{lem:2d-space-projective-iff} there exists a vector $\Vec{n} \in \mathbb{Z}^2$ such that $R_{uv} \cdot \Vec{n} \ge \Vec{b}$. Thus $R_{xy} \cdot M\Vec{n} = (R_{xy} \cdot M) \cdot \Vec{n} = R_{uv} \cdot \Vec{n} \ge \Vec{b}$. Observe that as $M$ is invertible, $\Vec{b}(\ell) = 0$ if and only if $\Vec{x}(\ell) = \Vec{y}(\ell) = 0$. Therefore by \autoref{lem:2d-space-projective-iff}, $X$ is projective w.r.t.\ $\Vec{x},\Vec{y}$ as well.
\end{proof}

\subsection{Proof of Lemma \ref{lem:srp-of-projective}}

\srpOfProjective*

\begin{proof}
    Suppose $X = \Span\{\Vec{u}, \Vec{v}\}$. Denote $\Vec{r}_\ell := (\Vec{u}(\ell), \Vec{v}(\ell))$ for $\ell \in [d]$. By definition there exist $i, j \in [d]$ such that $\Vec{r}_\ell \in \Cone(\{\Vec{r}_i, \Vec{r}_j\})$ for all $\ell \in [d]$. We claim that $I := \{i, j\}$ is a sign-reflecting projection for $X$. Let $\Vec{x} \in X$ and assume $\Vec{x}(i), \Vec{x}(j) \ge 0$. We need to verify that $\Vec{x}(\ell) \ge 0$ for all $\ell \in [d]$. Suppose $\Vec{x} = \alpha \Vec{u} + \beta \Vec{v}$ for some $\alpha, \beta \in \mathbb{Q}$. Consider any $\ell \in [d]$, we have $\Vec{r}(\ell) = \lambda \Vec{r}_i + \mu \Vec{r}_j$ for some $\lambda, \mu \ge 0$. Thus 
    \begin{align}
    \begin{split}
        \Vec{x}(\ell) &= \alpha \Vec{u}(\ell) + \beta \Vec{v}(\ell) \\
        &= \alpha (\lambda \Vec{u}(i) + \mu \Vec{u}(j)) + \beta (\lambda \Vec{v}(i) + \mu \Vec{v}(j))\\
        &= \lambda (\alpha\Vec{u}(i) + \beta \Vec{v}(i)) + \mu (\alpha\Vec{u}(j) + \beta \Vec{v}(j))\\
        &= \lambda \Vec{x}(i) + \mu \Vec{x}(j)\\ 
        & \ge 0. \qedhere
    \end{split}
    \end{align}
\end{proof}

\subsection{Proof of Proposition \ref{prop:canonical-axes-of-projective}}

\canonicalAxesOfProjective*

\begin{proof}
    Let $A := \begin{pmatrix}
        \Vec{u}(i) & \Vec{v}(i) \\
        \Vec{u}(j) & \Vec{v}(j)
    \end{pmatrix}$. We claim that $A$ is invertible. Otherwise there exists $(\alpha, \beta) \ne \Vec{0}$ such that $A \cdot (\alpha, \beta)^T = \Vec{0}$, which means that $\Vec{z} := \alpha \Vec{u} + \beta \Vec{v}$ satisfies $\Vec{z}|_I = \Vec{0}$. As $I$ is a sign-reflecting projection for $\CycleSpace(V)$, from $\Vec{z}|_I \ge \Vec{0}$ and $-\Vec{z}|_I \ge \Vec{0}$ we deduce $\Vec{z} \ge \Vec{0}$ and $-\Vec{z} \ge 0$. Thus $\Vec{z} = 0$, but this is impossible as $\Vec{u}, \Vec{v}$ are linearly independent. Now we can express the inverse of $A$ as 
    \begin{align}
        A^{-1} = \frac{1}{\Vec{u}(i) \Vec{v}(j) - \Vec{u}(j)\Vec{v}(i)} \begin{pmatrix}
            \Vec{v}(j) & -\Vec{v}(i) \\
            -\Vec{u}(j) & \Vec{u}(i)
        \end{pmatrix}.
    \end{align}
    Let $c := \Vec{u}(i) \Vec{v}(j) - \Vec{u}(j)\Vec{v}(i)$, we may assume w.l.o.g.\ that $c > 0$. Notice that $c \le 2\Size(V)^2$. We put
    \begin{align}
        \begin{pmatrix}\overline{\Vec{u}} & \overline{\Vec{v}}\end{pmatrix} := 
        \begin{pmatrix}{\Vec{u}} & {\Vec{v}}\end{pmatrix} \cdot (cA^{-1}).
    \end{align}
    Then we have $\overline{\Vec{u}}|_I = (c, 0)$ and $\overline{\Vec{v}}|_I = (0, c)$ and $\norm{\overline{\Vec{u}}}, \norm{\overline{\Vec{v}}} \le 2\Size(V)^2$. Using the fact that $I$ is sign-reflecting, we finally deduce that $\overline{\Vec{u}}, \overline{\Vec{v}} \ge \Vec{0}$.
\end{proof}

\subsection{Proof of Proposition \ref{prop:shift-sub-shift-src-norm-le}}

\shiftSubShiftSrcNormLE*

We will need to following bound on shifts.

\begin{proposition}
    \label{prop:norm-delta-le}
    Let $\Vec{x} \in \mathbb{Z}^d$. Then $\norm{\delta(\Vec{x})} \le 6 \Size(V)^2 \cdot \norm{\Vec{x}}$.
\end{proposition}

\begin{proof}
    Notice $\norm{\delta(\Vec{x})} \le c\norm{\Vec{x}} + (\norm{\overline{\Vec{v}}} + \norm{\overline{\Vec{v}}}) \cdot \norm{\Vec{x}}$. From \autoref{prop:canonical-axes-of-projective} we have $c, \norm{\overline{\Vec{u}}}, \norm{\overline{\Vec{v}}} \le 2\Size(V)^2$. Hence $\norm{\delta(\Vec{x})} \le 6 \Size(V)^2 \cdot \norm{\Vec{x}}$.
\end{proof}

\begin{proof}[Proof of \autoref{prop:shift-sub-shift-src-norm-le}]
    Let $\pi$ be a path such that $e \xrightarrow{\pi} e'$. We have $\delta(e') - \delta(e) = \delta(\Delta(\pi))$. If we remove cycles exhaustively from $\pi$, we will get a simple path eventually. Therefore, we can write $\Delta(\pi) = \Vec{c} + \Vec{z}$ where $\Vec{c} \in \CycleSpace(V)$ is the sum of effects of cycles, and $\Vec{z}$ is the effect of a simple path in $V$. Now $\delta(\Delta(\pi)) = \delta(\Vec{c}) + \delta(\Vec{z}) = \delta(\Vec{z})$. Notice that $\norm{\Vec{z}} \le \Size(V)$ as a simple path has at most $|Q|$ transitions. So by \autoref{prop:norm-delta-le} we have $\norm{\delta(\Vec{z})} \le 6 \Size(V)^2 \cdot \Size(V) \le 6 \Size(V)^3$.
\end{proof}

\subsection{Proof of Proposition~\ref{prop:shift-from-state-of-fixed-run}}

\shiftFromStateOfFixedRun*

\begin{proof}
    For each $q \in Q_\pi$, let $q(\Vec{x}_q)$ be the first configuration on $\pi$ whose state is $q$. We put $f_\pi(q) := \delta(\Vec{x}_q)$. Then for any $q(\Vec{x})$ on $\pi$, notice that $\Vec{c} := \Vec{x} - \Vec{x}_q \in \CycleSpace(V)$. Thus $\delta(\Vec{x}) = \delta(\Vec{x}_q) + \delta(\Vec{c}) = \delta(\Vec{x}_q) = f_\pi(q)$.
\end{proof}

\subsection{Proof of Lemma~\ref{lem:exists-2-vass-testing-inequ}}

\existsTwoVASSTestingInequ*

\begin{proof}
    As $m, n \ge 0$, the set $\{(x, y) \in \mathbb{N}^2 \mid mx + ny + b\}$ can be described as the region above finitely many minimal points. To be explicit, we define $M \subseteq \mathbb{N}^2$ as the following set:
    \begin{align}\label{eq:nonneg-testing}
        M := \begin{cases}
            \{(0, 0)\} & b \ge 0 \\
            \left\{ \left( \lceil - b/m \rceil, 0 \right) \right\} & b < 0, n = 0, m > 0\\ 
            \left\{ \left( 0, \lceil - b/n \rceil \right) \right\} & b < 0, n > 0, m = 0\\
            \left\{ \left( i, \lceil - (b + mi)/n \rceil \right) \mid 0 \le i \le \lceil - b/m \rceil \right\} & b < 0, n > 0, m > 0\\
            \emptyset & b < 0, n = 0, m = 0
        \end{cases}
    \end{align}
    It is routine to verify that for all $x, y \in \mathbb{N}$, $mx + ny + b \ge 0$ if and only if $(x, y) \ge \Vec{m}$ for some $\Vec{m} \in M$. Thus we will construct $V_{m,n;b}$ to contain states $p, q$ and $s_{\Vec{m}}$ for each $\Vec{m} \in M$, and transitions $(p, -\Vec{m}, s_{\Vec{m}})$ and $(s_{\Vec{m}}, +\Vec{m}, q)$ for each $\Vec{m} \in M$. It is straightforward to see that $p(x, y) \xrightarrow{*} q(x', y')$ if and only if $mx + ny + b \ge 0$ and $(x', y') = (x, y)$, and the path contains exactly 2 transitions. Considering the size of $V_{m,n;b}$, notice that $\norm{M} \le |b|$ and $|M| \le |b| + 1$. Thus $\Size(V_{m,n;b}) = (2 + |M|) + 2  \cdot (2|M|) \cdot (\norm{M} + 1) \le O(|b|^2)$.
\end{proof}

\subsection{Proof of Lemma \ref{lem:geo-2d-to-2d-projective}}

\geoTwoDtoTwoDProjective*

\begin{proof}
    Let $(V, s, t)$ be a geometrically 2-dimensional VASS where $\CycleSpace(V)$ is projective. Suppose $s = p(\Vec{x})$ and $t = q(\Vec{y})$. Let $I = \{i, j\}$ be a sign-reflecting projection of $\CycleSpace(V)$ and $\overline{\Vec{u}}, \overline{\Vec{v}}$ be the canonical axes given by \autoref{prop:canonical-axes-of-projective} where $\overline{\Vec{u}}(i) = \overline{\Vec{v}}(j) = c$, based on which the shift function $\delta$ is defined. Let $Q$ be the states of $V$ and $T$ be the transitions of $V$. Define $\mathcal{F}$ to be the set of functions $f : Q \to [-6\Size(V)^3, 6\Size(V)^3]^d$ satisfying $f(p) = \Vec{0}$ and $f(q) = \delta(\Vec{y}) - \delta(\Vec{x})$. Intuitively, $f$ assigns each state $r \in Q$ the shift $\delta_r^f := f(r) + \delta(\Vec{x})$. For each function $f \in \mathcal{F}$ we will construct a 2-VASS $(V_f, s_f, t_f)$ and $V_f = (Q_f, T_f)$ is described as follows.
    \begin{itemize}
        \item For each state $r \in Q$ we add two states $\overline{r}$ and $\underline{r}$ into $Q_f$. Intuitively $\underline{r}(\Vec{m})$ will represent a configuration $r(\Vec{m}')$ of $V$ where $\Vec{m}' := (\overline{\Vec{u}} \cdot \Vec{m}(1) + \overline{\Vec{v}} \cdot \Vec{m}(2) + \delta_r^f) / c$ for which we are sure that $\Vec{m}' \ge \Vec{0}$, while $\overline{r}(\Vec{m})$ means we are not sure if $\Vec{m}' \ge \Vec{0}$.
        \item For all transition $u = (r, \Vec{a}, r') \in T$ such that $\delta(\Vec{a}) = f(r') - f(r)$, we add the transition $(\underline{r}, \Vec{a}|_I, \overline{r'})$ into $T_f$. We verify that after each transition the recovered configuration is always integral.
        \begin{claim}
            Suppose $\Vec{m} \in \mathbb{Z}^2$ renders true that $\Vec{m}' := (\overline{\Vec{u}} \cdot \Vec{m}(1) + \overline{\Vec{v}} \cdot \Vec{m}(2) + \delta_r^f) / c$ is an integer vector. If $\underline{r}(\Vec{m}) \xrightarrow{\Vec{a}|_I} \overline{r'}(\Vec{n})$ for a transition $(\underline{r}, \Vec{a}|_I, \overline{r'})$, then $\Vec{n}' := (\overline{\Vec{u}} \cdot \Vec{n}(1) + \overline{\Vec{v}} \cdot \Vec{n}(2) + \delta_{r'}^f) / c$ is also an integer vector.
        \end{claim}
        \begin{claimproof}
            Notice that 
            \begin{align}
            \begin{split}
                \Vec{n}' &= \bigl( \overline{\Vec{u}} \cdot (\Vec{m}(1) + \Vec{a}(i)) + \overline{\Vec{v}} \cdot (\Vec{m}(2) + \Vec{a}(j)) + \delta_{r}^f + (\delta_{r'}^f - \delta_r^f) \bigr)/c \\
                &= \Vec{m}' + (\overline{\Vec{u}} \cdot \Vec{a}(i) + \overline{\Vec{v}} \cdot \Vec{a}(j) + \delta(\Vec{a}))/c \\ 
                &= \Vec{m}' + \Vec{a}.
            \end{split}
            \end{align}
            Hence $\Vec{n}'$ is integral as long as $\Vec{m}'$ is.
        \end{claimproof}
        \item For each pair of states $\overline{r}$ and $\underline{r}$, we will add a gadget to test non-negativity of the recovered configuration. For any $\Vec{m} \in \mathbb{N}^2$ we need to test if $\Vec{m}' := (\overline{\Vec{u}} \cdot \Vec{m}(1) + \overline{\Vec{v}} \cdot \Vec{m}(2) + \delta_r^f) / c \ge \Vec{0}$. Recall that the coordinates $\overline{\Vec{u}}$ and $\overline{\Vec{v}}$ are nonnegative. For each coordinate $\ell \in [d]$, let $V_\ell := V_{\overline{\Vec{u}}(\ell), \overline{\Vec{v}}(\ell), \delta_r^f(\ell)}$ be the 2-VASS obtained from \autoref{lem:exists-2-vass-testing-inequ}. We concatenate all these gadgets $V_\ell$ and put them between the states $\overline{r}$ and $\underline{r}$. This finishes the construction of $V_f$.
    \end{itemize}
    The source and target shall naturally be $s_f := \underline{p}(\Vec{x}|_I)$ and $t_f := \underline{q}(\Vec{y}|_I)$. Now we take the demanded family of 2-VASSes to be $\mathcal{V} := \{(V_f, s_f, t_f) \mid f \in \mathcal{F}\}$.

    For any run $\pi$ from $s$ to $t$, let $f_\pi : Q \to \mathbb{Z}^d$ be the function obtained from \autoref{prop:shift-from-state-of-fixed-run} (notice that we may extend $f_\pi$ to all states, by imposing $f_\pi(r) := 0$ if $r$ is not visited on $\pi$). Let $f := r \mapsto f_\pi(\Vec{r}) - \delta(\Vec{x})$. Then by \autoref{prop:shift-sub-shift-src-norm-le} we have $f \in \mathcal{F}$. It is routine to verify that $\pi$ can be mapped to a run in $V_f$ from $s_f$ to $t_f$ whose length is $(2d + 1)|\pi|$, since after each ``real'' transition, there is a gadget testing non-negativity using $2d$ transitions.

    On the other hand, consider any run $\pi'$ in some $V_f$ from $s_f$ to $t_f$. The structure of $V_f$ guarantees that $\pi'$ is a concatenetion of segments $\pi' = \pi_1\pi_2\ldots \pi_k$, each $\pi_i = u_i' \rho_i$ consists of a transition $u_i'$ projected from some transition $u_i \in T$ and a path $\rho_i$ of length $2d$ in some non-negativity testing gadget. The choice of $s_f$ and $t_f$ ensures that $\pi := u_1u_2 \ldots u_k$ is a $\mathbb{Z}$-run from $s$ to $t$ in $V$. As all non-negativity tests are passed, every configuration visited on $\pi$ is non-negative. Thus $\pi$ is a run from $s$ to $t$. We conclude that $(2d+1) \cdot \Len(V, s, t) = \bigcup_{(V', s', t') \in \mathcal{V}} \Len(V', s', t')$.

    Considering the size of each $V_f$, notice that if we ignore all the non-negativity testing gadgets, then the size of $(V_f, s_f, t_f)$ is no larger than $2\Size(V, s, t)$. There are $|Q|$ gadgets, each consisting of $d$ gadgets from \autoref{lem:exists-2-vass-testing-inequ} whose sizes are bounded by $O(\norm{\delta_f^r})^2 \le O(\norm{\delta(\Vec{x})} + 6\Size(V)^3)^2 \le O(6\Size(V)^2 \cdot \norm{\Vec{x}} + 6\Size(V)^3)^2 \le O(\Size(V, s, t))^6$. Thus $\Size(V_f, s_f, t_f) \le H_1(\Size(V, s, t)$ for some polynomial $H_1$.
\end{proof}

\subsection{Proof of Lemma~\ref{lem:normal-vec-of-non-projective}}

\normalVecOfNonProjective*

\begin{proof}
    Denote $\Vec{r}_\ell := (\Vec{u}(\ell), \Vec{v}(\ell))$. Let $R \subseteq \{\Vec{r}_1, \ldots, \Vec{r}_d\}$ be minimal such that $\Cone(R) = \Cone\{\Vec{r}_1, \ldots, \Vec{r}_d\}$. As $X$ is not projective, $R$ contains at least 3 vectors $\Vec{r}_i, \Vec{r}_j, \Vec{r}_k$. These 3 vectors in $\mathbb{Z}^2$ must be linearly dependent. But two of them must be linearly independent, otherwise one vector is a positive multiple of another and can be removed from $R$. We assume $\Vec{r}_i$ and $\Vec{r}_j$ are linearly independent. Then there exist $\alpha, \beta \in \mathbb{Q}$ such that $\Vec{r}_k = \alpha \Vec{r}_i + \beta \Vec{r}_j$. We may assume $\alpha, \beta$ have the same sign by renaming  $\Vec{r}_i, \Vec{r}_j, \Vec{r}_k$ properly. We cannot have $\alpha, \beta \ge 0$ as otherwise $\Vec{r}_k$ can be removed from $R$, which contradicts the minimality of $R$. So $\alpha, \beta \le 0$, i.e.\ $\Vec{r}_k + (-\alpha) \Vec{r}_i + (-\beta) \Vec{r}_j = \Vec{0}$. We may pick $\Vec{n} := c\Vec{e}_k - c\alpha\Vec{e}_i - c\beta\Vec{e}_j$, where $\Vec{e}_\ell$ is the unit vector in the $\ell$-th coordinate, and $c \in \mathbb{N}$ is a number to make sure $\Vec{n} \in \mathbb{N}^d$. By Cramer's rule, it is enough to take $c = |\Vec{u}(i)\Vec{v}(j) - \Vec{u}(j)\Vec{v}(i)|$. Thus $\norm{\Vec{n}} \le 2M^2$, and clearly $|\Supp(\Vec{n})| \le 3$. Now we have $\Inner{\Vec{n}, \Vec{u}} = \Inner{\Vec{n}, \Vec{v}} = 0$. So for any $\Vec{x} \in X = \Span(\Vec{u}, \Vec{v})$ it holds $\Inner{\Vec{n}, \Vec{x}} = 0$. Notice that as none of $\Vec{r}_i, \Vec{r}_j, \Vec{r}_k$ can be zero vectors, we have $i, j, k \in \Supp(\{\Vec{u}, \Vec{v}\}) \subseteq \Supp(X)$. Thus $\Supp(\Vec{n}) \subseteq \Supp(X)$.
\end{proof}

\subsection{Proof of Proposition~\ref{prop:bounded-counters-of-non-projective}}

\boundedCountersOfNonProjective*

\begin{proof}
    Similar to the proof of \autoref{prop:shift-sub-shift-src-norm-le}, we may write $\Vec{y} - \Vec{x}$ as $\Vec{c} + \Vec{z}$ where $\Vec{c} \in \CycleSpace(V)$ and $\Vec{z}$ is the effect of a simple path, so $\norm{\Vec{z}}\le \Size(V)$. Now $\Inner{\Vec{n}, \Vec{y}} = \Inner{\Vec{n}, \Vec{c}} + \Inner{\Vec{n}, \Vec{z} + \Vec{x}} = \Inner{\Vec{n}, \Vec{z} + \Vec{x}} \le |\Supp(\Vec{n})| \cdot \norm{\Vec{n}} \cdot (\norm{\Vec{z}} + \norm{\Vec{x}}) \le 6 \Size(V, p(\Vec{x}), q(\Vec{y}))^3$. As $\Vec{n}, \Vec{y} \in \mathbb{N}^d$ are non-negative, for all $k \in K = \Supp(\Vec{n})$ we must have $\Vec{y}(k) \le 6 \Size(V, p(\Vec{x}), q(\Vec{y}))^3$. Thus $\norm{ \Vec{y}|_K } \le 6 \Size(V, p(\Vec{x}), q(\Vec{y}))^3$.
\end{proof}

\subsection{Proof of \autoref{lem:projective-geo-2d-from-non-projective-geo-2d}}

\projectiveGeoTwoDFromNonProjGeoTwod*

\begin{proof}
    Let $s =: p(\Vec{x})$ and $t =: q(\Vec{y})$. We assume that $t$ is indeed reachable from $s$, as otherwise we just let $(V', s', t')$ be a trivial unreachable VASS. Suppose $V =: (Q, T)$. Denote $M = [0, 6\Size(V, s, t)^3]^{|K|}$. We first construct a VASS $(V'', s', t')$ where $V'' = (Q'', T'')$ as follows. 
    \begin{itemize}
        \item $Q'' := Q \times M$. For a state $(q, \Vec{m}) \in Q''$ we write it as $q_{\Vec{m}}$.
        \item $T'' := \{(r_{\Vec{m}}, \Vec{a}, r'_{\Vec{m}'}) \mid (r, \Vec{a}, \Vec{r}') \in T, \Vec{a}|_K = \Vec{m}' - \Vec{m}\}$.
        \item $s' := p_{\Vec{m}}(\Vec{x})$ and $t' := q_{\Vec{n}}(\Vec{y})$, where $\Vec{m} := \Vec{x}|_K$ and $\Vec{n} := \Vec{y}|_K$. By \autoref{prop:bounded-counters-of-non-projective}, as $t$ is reachable from $s$, we have $\norm{\Vec{n}} \le 6\Size(V)^3$ thus $\Vec{n} \in M$.
    \end{itemize}
    Any run in $V''$ can be turned into a run in $V$ by simply ignoring the markings at the states. On the other hand, by \autoref{prop:bounded-counters-of-non-projective}, a run starting from $s$ in $V$ can also be mapped to a run in $V''$ with source $s'$. Hence $\Len(V'', s', t') = \Len(V, s, t)$. Notice that $|M| \le (6\Size(V, s, t)^3 + 1)^3$. Thus $\Size(V'', s', t')$ is bounded polynomially in $\Size(V, s, t)$. We are not guaranteed that $\CycleSpace(V'')$ is projective now. But it is routine to verify that 
    \begin{align}
        \CycleSpace(V'') \subseteq \{\Vec{c} \in \CycleSpace(V) \mid \Vec{c}|_K = \Vec{0}\} \subsetneq \CycleSpace(V),
    \end{align}
    where the second inclusion is proper because $K \subseteq \Supp(\CycleSpace(V))$. Therefore, $\dimcyc(V'') < \dimcyc(V) = 2$. We let $V'$ be obtained from $V''$ by adding a new state and self-loops on this state, whose effects are chosen as follows to make $\CycleSpace(V')$ projective:
    \begin{itemize}
        \item If $\dimcyc(V'') = 0$ then we may add two self-loops of effects $\Vec{e}_i$ and $\Vec{e}_j$ for any $i \ne j \in [d]$.
        \item If $\dimcyc(V'') = 1$, let $\CycleSpace(V'')$ be spanned by effect $\Vec{u}$ of a simple cycle. 
        \begin{itemize}
            \item If $\Vec{u}$ contains two different entries (so $\Vec{u}$ is not in parallel with $\Vec{1}$), we can add a self-loop with effect $\Vec{v} = \Vec{1}$. To see that $\CycleSpace(V')$ is projective, one may apply \autoref{lem:2d-space-projective-iff} to $\CycleSpace(V') = \Span\{\Vec{u}, \Vec{v}\}$ with $\Vec{n} = (0, 1)$.
            \item Otherwise, $\Vec{u} = \alpha \cdot \Vec{1}$ for some $\alpha \in \mathbb{Q}$. We add a self-loop with effect $\Vec{v} = \Vec{e}_i - \Vec{e}_j$ for any $i \ne j \in [d]$. Now $\CycleSpace(V')$ is projective by applying \autoref{lem:2d-space-projective-iff} with $\Vec{n} = (\alpha, 0)$. \qedhere
        \end{itemize}
    \end{itemize}
\end{proof}

\subsection{Proof of \autoref{cor:min-len-le-poly-of-geo-2d}}

\minLenLEPolyOfGeoTwoD*

\begin{proof}
    Assume $s \xrightarrow{*} t$. By \autoref{lem:geo-2d-to-2d-poly} there exists a family $\mathcal{V}$ of 2-VASSes such that $(2d+1)\cdot\Len(V, s, t) = \bigcup_{(V', s', t') \Len(V', s', t')}$ and $\Size(V', s', t') \le \Size(V, s, t)^{O(1)}$ for all $(V', s', t') \in \mathcal{V}$. By \autoref{thm:min-len-le-poly-of-2d}, for any $(V', s', t') \in \mathcal{V}$ with $s' \xrightarrow{*} t'$ we have $\min \Len(V, s, t) \le \min \Len(V', s', t') / (2d+1) \le (\Size(V, s, t)^{O(1)})^{O(1)} \le \Size(V, s, t)^{O(1)}$. 
\end{proof}

\subsection{Proof of Theorem \ref{thm:geo-2-vass}}

\complexityGeoTwoVASS*

\begin{proof}
    Consider a reachability instance $(V, s, t)$. Under binary encoding, the length bound in \autoref{cor:min-len-le-poly-of-geo-2d} becomes exponential in the input size. But polynomial space is enough to enumerate every run of exponential length. This confirms the \PSPACE{} membership. The \PSPACE{} hardness is inherited from that of 2-VASS \cite{DBLP:journals/jacm/BlondinEFGHLMT21}.

    Under unary encoding, we may decide reachability by searching for a run of polynomial length, which can be done in \NP{}. When the actual dimension $d$ is a fixed constant in the problem, the searching can be improved to \NL{} since in the searching we may only keep record of the number of steps and the current configuration whose size is bounded by $\log|Q| + d\cdot \log(\Size(V, s, t)^{O(1)}) \le d \cdot O(\log\Size(V, s, t))$. Notice, however, that when $d$ is not fixed, storing a single configuration requires polynomial space.

    The \NL{}-hardness follows from the trivial reduction from the graph reachability problem. For \NP{}-hardness, we reduce from the \textsc{Subset-Sum} problem, viewing the $d$-coordinates as a $d$-bit binary representation. Details are deferred to the appendix. We remark that a similar reduction can be found in \cite[Lemma 2]{DBLP:conf/apn/Leroux21}.

    \newcommand{\scnum}{\OpNameSC{num}}
    \newcommand{\scbin}{\OpNameSC{bin}}

    Consider a \textsc{Subset-Sum} instance $\langle S = \{a_1, \dots, a_n \}, s \rangle$, where $S \subseteq \mathbb{N}$ and the goal is to find a subset of $S$ whose sum is $s$. Let $d$ be the minimum integer such that $2^d > a_1 + a_2 + \dots + a_n$, so that the sum of any possible subset can be represented as a $d$-bit binary number. For a vector $\Vec{v} \in \mathbb{N}^d$ we define $\scnum(\Vec{v}) = \sum_{i=1}^d 2^{i - 1}\Vec{v}(i)$. For each number $a < 2^d$, we define $\scbin(a)$ to be the unique vector in $\mathbb{N}^d$ such that $\scnum(\scbin(a)) = a$. 

    We will construct a $d$-VASS $V = (Q, T)$ with no cycles so that the \textsc{Subset-Sum} instance has a solution if and only if there is a run in $V$ from $\Vec{0}$ to $\scbin(s)$. First we add to $V$ the states $p_1, q_1, \dots, p_n, q_n, p_{n+1}$. For each $j \in [n]$ we add the following transitions:
    \[
        p_j \xrightarrow{\Vec{0}} q_j, \quad p_j \xrightarrow{\scbin(a_j)} q_j.
    \]
    Intuitively, at state $p_j$ we can non-deterministically choose $a_j$ by adding its binary representation. Now before moving from $q_j$ to $p_{n+1}$ we need to make sure that the configuration holds a legal binary value. For this purpose we devise a VASS $G$ as an adjusting gadget. $G$ consists of states $r_1, \dots, r_d$ and the following transitions for each $i \in [d - 1]$:
    \[
        r_i \xrightarrow{\Vec{0}} r_{i+1}, \quad r_i \xrightarrow{-2 \Vec{e}_i + \Vec{e}_{i + 1}} r_{i+1}
    \]
    where $\Vec{e}_i$ is the unit vector in the $i$-th coordinate. One can easily see that $G$ does not change the ``binary'' number represented by the counters:
    \begin{claim}
        If $r_1(\Vec{m}) \xrightarrow{*} r_d (\Vec{n})$ in $G$ then $\scnum(\Vec{m}) = \scnum(\Vec{n})$.
    \end{claim}
    
    Moreover, $G$ is able to adjust the counters properly after each addition:
    
    \begin{claim}
        Let $a, b \in \mathbb{N}$ be such that $a + b < 2^d$ and let $\Vec{m} = \scbin(a) + \scbin(b)$. Then there exists $\Vec{n} \in \{0, 1\}^d$ such that $r_1 (\Vec{m}) \xrightarrow{*} r_d (\Vec{n})$ in $G$.
    \end{claim}
    
    \begin{claimproof}
        Note that $\Vec{m}(i) \in \{0, 1, 2\}$ for each $i \in [d]$. Once a counter $\Vec{m}(i) \geq 2$ we perform the transition $r_i \xrightarrow{-2 \Vec{e}_i + \Vec{e}_{i + 1}} r_{i+1}$. Under this strategy at each state $r_k$ we are guaranteed that (i) for $i < k$ the $i$-th counter is in $\{0, 1\}$, and (ii) the $k$-th counter is no greater than $3$. Finally, at state $r_d$ it must be valid that the $d$-th counter is at most one as $a + b < 2^d$.
    \end{claimproof}

    Now for each $j \in [n]$ we add a copy of $V$ to connect the state $q_j$ to $p_{j+1}$. One can easily check from the above claims that the \textsc{Subset-Sum} instance has a solution if and only if $p_{n+1} (\scbin(s))$ is reachable from $p_1(\Vec{0})$.
\end{proof}

\subsection{Proof of Propostition~\ref{prop:cone-generated-by-4-vec}}

\coneGeneratedByFourVec*

\begin{proof}
    Denote $X := \{\Vec{r}_1,\dots,\Vec{r}_d\}$. We may assume w.l.o.g.\ that $\Vec{r}_\ell\ne \Vec{0}$ for each $\ell \in [d]$. Observe that if $\Vec{u}$ and $\Vec{v}$ are linearly dependent, i.e.\ $\Span(X)$ is one-dimensional, then the result holds easily even with $|I| \le 2$. Now we assume $\Vec{u}$ and $\Vec{v}$ are linearly independent.
    \begin{claim}
        There exists a vector $\Vec{n}\in\mathbb{N}^2$ such that $\Inner{\Vec{r}_\ell,\Vec{n}}\ne 0$ for every $\ell\in[d]$.
    \end{claim}
    \begin{claimproof}
        Let $\Vec{n}:=(\Size(V)+1, 1)$. For any vector $\Vec{r}_\ell=(a,b)\in X$, if $a=0$, then $\Inner{\Vec{r}_\ell,\Vec{n}}=b \ne 0$ must hold, as $\Vec{0}\notin X$. Otherwise, it holds that 
        \begin{equation}
            |\Inner{\Vec{r}_\ell,\Vec{n}}|\ge |a|\cdot (\Size(V)+1)-|b|\ge \Size(V)+1- \Size(V)>0.
        \end{equation}
        Therefore, $\Inner{\Vec{r}_\ell,\Vec{n}}\ne 0$ for every $\ell\in [d]$.
    \end{claimproof}
    Define $I_1:=\{\ell \in [d]\mid \Inner{\Vec{r}_\ell,\Vec{n}}>0\}$ and $I_{2}:=\{\ell\in[d]\mid \Inner{\Vec{r}_\ell,-\Vec{n}}>0\}$. By \autoref{lem:2d-space-projective-iff}, there exist two indices $i_1, j_1 \in I_1$ such that $\Cone(\{\Vec{r}_\ell\mid \ell \in I_1\})=\Cone(\{\Vec{r}_{i_1}, \Vec{r}_{j_1}\})$ and symmetrically, two indices $i_2, j_2 \in I_2$ such that $\Cone(\{\Vec{r}_\ell\mid \ell \in I_2\})=\Cone(\{\Vec{r}_{i_2}, \Vec{r}_{j_2}\})$. Let $I:=\{i_1, i_2, j_1, j_2\}$, clearly $|I|\le 4$. Note that $[d]:=I_1\cup I_2$. It is straightforward that $\Cone(\{\Vec{r}_1,\dots,\Vec{r}_d\})=\Cone(\{\Vec{r}_i\mid r\in I\})$. 
\end{proof}

\subsection{Proof of Lemma~\ref{cor:exists-subset-card-4}}

\existsSubsetCardFour*

\begin{proof}
    By \autoref{prop:cone-generated-by-4-vec}, there exists a $I\subseteq[d]$ such that $|I|\le 4$ and $\Cone(\{\Vec{r}_1,\dots,\Vec{r}_d\})=\Cone(\{\Vec{r}_i\mid r\in I\})$. A simple combinatorial algorithm can enumerate all the subsets $I$ with $|I|\le 4$ (only $O(d^4)$ choices) and then verify in logarithmic space whether every $\Vec{r}_\ell$ can be expressed as a non-negative linear combination of vectors in $\{\Vec{r}_i\mid r\in I\}$. Indeed, for any $\ell \in [d]$, if $\Vec{r}_\ell \in \Cone(\{\Vec{r}_i\mid r\in I\})$, by Carathéodory’s Theorem~\cite[Corollary 7.1i]{DBLP:books/daglib/0090562}, we have $c_\ell\cdot \Vec{r}_\ell=\alpha_\ell \cdot \Vec{r}_{i_\ell}+\beta_\ell\cdot\Vec{r}_{j_\ell}$ for some $i_\ell,j_\ell\in I$ and $c_\ell ,\alpha, \beta\in \mathbb{N}$, where $c_\ell>0$. By Cramer's rule, these coefficients can be bounded by a polynomial in $\Size(V)$. We can compute these coefficients by brute-force searching in logarithmic space.
    
    Now assume that we have correctly determined $I$ and the coefficients in $c_\ell\cdot \Vec{r}_\ell=\alpha_\ell \cdot \Vec{r}_{i_\ell}+\beta_\ell\cdot\Vec{r}_{j_\ell}$ for every $\ell \in [d]$. Then given any $\Vec{x}\in \CycleSpace(V)$, we may assume that $\Vec{x}=\lambda\cdot \Vec{u}+\mu\cdot \Vec{v}$ for some $\lambda, \mu \in \mathbb{Q}$. Then we can deduce that
    \begin{equation}
        \begin{aligned}
            c_\ell\cdot \Vec{x}(\ell)&=\Inner{(\lambda,\mu),c\cdot\Vec{r}_\ell}\\
            &=\Inner{(\lambda,\mu),\alpha_\ell \cdot \Vec{r}_{i_\ell}+\beta_\ell\cdot\Vec{r}_{j_\ell}}\\
            &=\alpha_\ell \cdot \Vec{x}(i_\ell)+\beta_\ell\cdot\Vec{x}(j_\ell),
        \end{aligned}
    \end{equation}
    which completes the proof.
\end{proof}

\end{document}